\documentclass{LMCS}

\def\dOi{10(4:6)2014}
\lmcsheading%
{\dOi}
{1--73}
{}
{}
{Mar.~28, 2013}
{Dec.~\phantom09, 2014}
{}

\ACMCCS{[{\bf Theory of computation}]: Semantics and
  reasoning---Program reasoning---Program verification; [{\bf Software
and its engineering}]: Software organization and properties---Software
  functional properties---Formal methods---Model checking}


\usepackage{xcolor,pgf,tikz,pgflibraryarrows,pgffor,pgflibrarysnakes}
\usepackage{stmaryrd}
\usepackage{amssymb,amsmath}
\usepackage{subfigure}
\usepackage{url,xspace}
\usepackage{dsfont}
\usepackage{yhmath,wasysym,longtable,array}

\usepackage{hyperref}

\def\runningexample{\begin{tikzpicture}[yscale=.9,xscale=1.1]
        \everymath{\scriptstyle}
        \path (-4.5,0) node[draw,circle,inner sep=2pt] (q0) 
        {$\ell_0$};
        \path (-4.5,-.7) node[] (q0b) 
        {$x \le 1$};
        
        \path (-1.5,0) node[draw,circle,inner sep=2pt] (q1) 
        {$\ell_1$};
        \path (-1.5,-.7) node[] (q0b) {$x \le 4$}; 
        \path (-1.5,.7) node[] (q1b) {$\{p_1\}$};
        
        \path (2,0) node[draw,circle,inner sep=2pt] (q2) 
        {$\ell_2$};
        \path (2,-.7) node[] (q2b)  {$x \le 2$};
        \path (2,.7) node[] (q1b) {$\{p_2\}$};
        
        \path (4,0) node[draw,circle,inner sep=2pt] (q3) 
        {$\ell_3$};
        \path (4,.7) node[] (q1b) {$\{p_1\}$};
        
        \draw[arrows=-latex'] (q0) .. controls +(20:45pt) .. (q1) node[pos=.5,
        above,sloped] {$e_2,~x \le 1$};
        
        \draw[arrows=-latex'] (q1) .. controls +(200:45pt) .. (q0) node[pos=.5,
        below,sloped] {$e_3,~x = 1$};
        
        \draw[arrows=-latex'] (q1) .. controls +(20:55pt) .. (q2) node[pos=.5,
        above,sloped] {$e_4,~x \ge 3,~x:=0$};
        
        \draw[arrows=-latex'] (q2) .. controls +(200:55pt) .. (q1) node[pos=.5,
        below,sloped] {$e_5,~x \le 2$};
        
        \draw[arrows=-latex'] (q2) -- (q3) node[pos=.5,above,sloped] {$e_6,~x
          =0$};
        
        \draw [-latex'] (q0) .. controls +(210:40pt) and +(150:40pt) .. (q0)
        node [midway,left] {$e_1,~x \leq 1$};
        
        \draw [-latex'] (q3) .. controls +(330:40pt) and +(30:40pt) .. (q3) node
        [midway,right] {$e_7,~x \leq 1$};
      \end{tikzpicture}}

\newcommand{\bs}{{\bf s}}
\newcommand{\be}{{\bf e}}

\newcommand{\bluegraph}{\ensuremath{\mathcal{G}_{\textsf{b}}}}

\newcommand{\sem}[1]{\ensuremath{\llbracket #1 \rrbracket}}

\newcommand{\ligne}{\begin{center}
    \vrule width 8cm height 0.3mm
  \end{center} \null\bigskip}

\newcommand\Prob{\ensuremath{\mathbb{P}}\xspace}

\newcommand\Cyl{\ensuremath{\textsf{Cyl}}}
\newcommand\Pol{\ensuremath{\textsf{\upshape Pol}}}
\newcommand\Proj{\ensuremath{\textsf{\upshape Proj}}}

\newcommand\IN{{\mathbb N}}
\newcommand\IZ{{\mathbb Z}}
\newcommand\IQ{{\mathbb Q}}
\newcommand\IR{{\mathbb R}}

\newcommand\ud{\mathrm{d}}
\newcommand\car[1]{\mathds{1}_{#1}}

\renewcommand{\widering}[1]{\stackrel{{\scriptscriptstyle \circ}}{\wideparen{#1}}}

\newcommand{\Runs}{\ensuremath{\textsf{Runs}}}

\newcommand\A{{\mathcal{A}}} 
\newcommand\B{{\mathcal{B}}} 
\newcommand\Fcal{{\mathcal{F}}}

\renewcommand\rho{\varrho}
\renewcommand\phi{\varphi}
\renewcommand{\path}[2][]{\pi_{#1}(#2)}
\newcommand{\cell}{\textsf{guard}}
\newcommand{\Zeno}{\textsf{{\upshape Zeno}}}

\newcommand{\RA}{\ensuremath{\mathsf{R}(\A)}}
\newcommand{\R}[1]{\ensuremath{\mathsf{R}(#1)}}

\newcommand{\AP}{\textsf{\upshape{AP}}\xspace}
\def\robust{\mathrel|\joinrel\approx} 
\def\true{\ensuremath{\mathtt{t\!t}}}
\def\false{\ensuremath{\mathtt{f\!f}}}

\makeatletter
\def\definelogic{\@ifnextchar[{\@defloglong}{\@deflogshort}}
\def\@defloglong[#1]#2{\expandafter\gdef\csname #2\endcsname
  {\ensuremath{\textsf{\upshape #1}}\xspace}}
\def\@deflogshort#1{\expandafter\gdef\csname #1\endcsname
  {\ensuremath{\textsf{\upshape #1}}\xspace}}
\definelogic{LTL}
\definelogic{CTL}
\definelogic[CSL$_{\textsf{TA}}$]{CSLta}
\definelogic{CSL}
\makeatother

\newcommand{\ComplexityFont}[1]{{\sffamily\upshape #1}}

\newcommand{\PSPACE}{\ComplexityFont{PSPACE}\xspace}

\newcommand{\NLOGSPACE}{\ComplexityFont{NLOGSPACE}\xspace}

\newcommand*{\modality}[2][]{\def\@rgone{#1}%
  \ifx\@rgone\@empty
  \ensuremath{\text{\rmfamily\upshape\bfseries {#2}}}%
  \else
  \ensuremath{\text{\rmfamily\upshape\bfseries {#2}}_{#1}}%
  \fi}
\newcommand*{\dmodality}[2][]{\def\@rgone{#1}%
  \ifx\@rgone\@empty
  \ensuremath{\smash{\widetilde{\text{\rmfamily\upshape\bfseries {#2}}}}%
        {\vphantom{\text{\rmfamily\upshape\bfseries {#2}}}}}%
  \else
  \ensuremath{\smash{\widetilde{\text{\rmfamily\upshape\bfseries {#2}}}_{#1}}%
        {\vphantom{\text{\rmfamily\upshape\bfseries {#2}}_{#1}}}}%
  \fi}
\newcommand*{\Until}[1][]{\ifmmode\,\fi\modality[{#1}]{U}\ifmmode\,\else\expandafter\xspace\fi}
\let\U\Until
\newcommand*{\WUntil}[1][]{\ifmmode\,\fi\dmodality[{#1}]{U}\ifmmode\,\else\expandafter\xspace\fi}
\let\W\WUntil
\newcommand*{\Rel}[1][]{\ifmmode\,\fi\dmodality[{#1}]{U}\ifmmode\,\else\expandafter\xspace\fi}

\newcommand*{\F}[1][]{\modality[{#1}]{F}\ifmmode\,\else\expandafter\xspace\fi}
\newcommand*{\G}[1][]{\modality[{#1}]{G}\ifmmode\,\else\expandafter\xspace\fi}
\newcommand*{\X}[1][]{\modality[{#1}]{X}\ifmmode\,\else\expandafter\xspace\fi}
\let\WUntil\W
\newcommand*{\All}[1][]{\modality[{#1}]{A}\ifmmode\,\else\expandafter\xspace\fi}
\newcommand{\egdef}{\stackrel{\mbox{\begin{scriptsize}\text{{\rmfamily def}}\end{scriptsize}}}{=}}

\newcommand\pref{{\mathsf{pref}}}
\newcommand{\thickgraph}{\ensuremath{\mathcal{G}_{\textsf{t}}}}

\newcommand{\source}{\ensuremath{\mathsf{source}}}
\newcommand{\target}{\ensuremath{\mathsf{target}}}
\newcommand{\last}{\ensuremath{\mathsf{last}}}
\newcommand{\MC}{\ensuremath{\mathsf{MC}}\xspace}

\newcommand{\corn}{\null\hfill \hbox to 0pt{$\lrcorner$\hss}}

\makeatletter
\let\lem\relax
\let\defi\relax
\let\thm\relax
\let\c@thm\relax
\let\rem\relax
\let\exa\relax
\let\cor\relax
\let\prop\relax
 \newtheorem{lem}{Lemma}[section]
 \newtheorem{defi}{Definition}[section]
 \newtheorem{thm}{Theorem}[section]
 \newtheorem{rem}{Remark}[section]
 \newtheorem{exa}{Example}[section]
 \newtheorem{cor}{Corollary}[section]
 \newtheorem{prop}{Proposition}[section]
\let\c@defi\c@thm
\let\c@lem\c@thm
\let\c@cor\c@thm
\let\c@rem\c@thm
\let\c@exa\c@thm
\let\c@prop\c@thm
\makeatother
\usepackage{thm-restate}

\begin{document}

\title[Stochastic timed automata]{Stochastic timed automata}

\author[Bertrand]{Nathalie Bertrand\rsuper a}	
\address{{\lsuper a}Inria Rennes, France}	
\email{nathalie.bertrand@inria.fr}  

\author[Bouyer]{Patricia Bouyer\rsuper b}	
\address{{\lsuper b}LSV, CNRS \& ENS Cachan, France}	
\email{bouyer@lsv.ens-cachan.fr}  

\author[Brihaye]{Thomas Brihaye\rsuper c}	
\address{{\lsuper{c,d}}Universit\'e de Mons, Belgium}	
\email{\{thomas.brihaye,quentin.menet\}@umons.ac.be}  

\author[Menet]{Quentin Menet\rsuper d}	
\address{\vspace{-18 pt}}	

\author[Baier]{Christel Baier\rsuper e}	
\address{{\lsuper{e,f}}TU Dresden, Germany}	
\email{\{baier,groesser\}@tcs.inf.tu-dresden.de}  

\author[Gr\"osser]{Marcus Gr\"osser\rsuper f}	
\address{\vspace{-18 pt}}	

\author[Jurdzi{\'n}ski ]{Marcin Jurdzi{\'n}ski\rsuper g}	
\address{{\lsuper g}University of Warwick, UK}	
\email{mju@dcs.warwick.ac.uk}  




\begin{abstract}
  \noindent  A stochastic timed automaton is a purely stochastic process defined
  on a timed automaton, in which both delays and discrete choices are
  made randomly. We study the almost-sure model-checking problem for
  this model, that is, given a stochastic timed automaton $\A$ and a
  property $\varphi$, we want to decide whether $\A$ satisfies
  $\varphi$ with probability $1$. In this paper, we identify several
  classes of automata and of properties for which this can be
  decided. The proof relies on the construction of a finite
  abstraction, called the thick graph, that we interpret as a finite
  Markov chain, and for which we can decide the almost-sure
  model-checking problem. Correctness of the abstraction holds when
  automata are almost-surely fair, which we show, is the case for two
  large classes of systems, single-clock automata and so-called
  weak-reactive automata. Techniques employed in this article gather tools
  from real-time verification and probabilistic verification, as well
  as topological games played on timed automata.
\end{abstract}

\maketitle
\vfill
\section{Introduction}
\subsection*{Timed automata and their extensions.}
In the last twenty years a huge effort has been made to design
expressive models for representing computerised systems. As part of
this effort the model of timed automata~\cite{AD90,AD94} has been
proposed in the early 90's as a suitable model for representing
systems with real-time constraints. Numerous works have focused on
that model, and it has received an important tool support, with for
instance the development of tools like Uppaal~\cite{BDL+06} or
Kronos~\cite{BDM+98}.

Given the success of the timed-automata-based technology for verifying
real-time systems, several extensions have been proposed, with the aim
of representing more faithfully real systems. They include timed
games~\cite{AMPS98} for modeling control problems and priced timed
automata~\cite{ATP01,BFH+01,BFLM11} for modeling various quantities
in timed systems, like energy consumption.

\subsection*{Stochastic extensions of timed automata.}
Many applications like communication protocols require models 
integrating both real-time constraints and randomised aspects (see
e.g.~\cite{stoelinga03}). The development of such models and
corresponding verification algorithms is a challenging task, since it
requires combining techniques from both fields of real-time verification
and probabilistic verification. In the literature we distinguish two
main different approaches.

A first approach consists in modeling the system as a purely
stochastic process, and to express soft real-time constraints in the
property that is checked. A model of choice for the system is that of
continuous-time Markov chains (CTMC for short), while a rather wide
spectrum of property formalisms has been considered, going from the
logic \CSL (continuous stochastic logic) and extensions
thereof~\cite{ASSB00,BHHK03,DHS09,ZJNH11} to (deterministic) timed
automata~\cite{CHKM11}. In this context several exact and approximate
model-checking algorithms have been developed.

Another approach consists in integrating both features into a complex
model (e.g. an extension of timed automata or Petri nets with
stochastic evolution rules), and to analyse this model. This allows
one to represent hard timing contraints such as deadlines.  In this
article we focus on automata-based models, and therefore only review
related work on models based on timed automata. Such models include
probabilistic timed automata~\cite{KNSS02} where discrete
distributions are assigned to actions and for which the tool
Prism~\cite{KNP11} has been developed. Delays or durations of events
can also be made randomised. This is done for instance
in~\cite{ACD91a,ACD91b} and later in~\cite{KNSS00}, yielding either
independent events and exact model-checking algorithms (for a
probabilistic and timed extension of computation tree logic), or
approximate model-checking algorithms.

The current work follows this last approach, and surveys and extends
results based on the model of \emph{stochastic timed automata}.  This
model has been proposed and studied in a series of
papers~\cite{BBBBG07,BBBBG08,BBJM12}. The semantics of a stochastic
timed automaton is a purely stochastic process based on a timed
automaton, in which both delays and discrete choices are made
randomly.  This model has later been extended with non-determinism and
interaction~\cite{BF09,BS12}, but in this article we focus on the
original purely stochastic model.

\subsection*{Overview of the contributions.}
In this article we are interested in the almost-sure model-checking of
stochastic timed automata. This problem asks, given a stochastic timed
automaton $\A$ and a property $\varphi$, whether $\A$ almost-surely
satisfies $\varphi$ (that is, with probability $1$). Our approach to
solve this problem relies on the construction of a finite Markov chain
$\MC(\A)$\footnote{In the core of the article, $\MC(\A)$ is the
  so-called thick graph $\thickgraph(\A)$, that we interpret as a
  finite Markov chain, putting the uniform distributions over edges.}
on which we will check whether $\varphi$ almost-surely holds or
not. We will then say that the abstraction $\MC(\A)$ is correct
w.r.t. $\varphi$ whenever $\varphi$ almost-surely holds equivalently
in $\A$ and in $\MC(\A)$. Unfortunately, this will not be the case in
general, and beyond the introduction of stochastic timed automata, the
main goal of this article is to identify subclasses of stochastic timed
automata and subclasses of properties for which the abstraction
$\MC(\A)$ is correct.

More precisely, we show that $\MC(\A)$ is a correct abstraction
w.r.t. property $\varphi$ in the following cases: (i) if $\varphi$
is a safety property, (ii) if $\varphi$ is an $\omega$-regular (or
\LTL) property and $\A$ is a single-clock stochastic timed automaton,
and (iii) if $\varphi$ is an $\omega$-regular (or \LTL) property and
$\A$ belongs to a subclass of stochastic timed automata called
\emph{weak reactive}.  In fact, cases (ii) and (iii) are
consequences of a more general result stating that if the runs in a
stochastic timed automaton $\A$ are almost-surely
fair\footnote{Roughly, a run is fair if any edge which is enabled
  infinitely often is taken infinitely often.}, then $\MC(\A)$ is a
correct abstraction. The results then follow from the (highly non
trivial) proof that both weak reactive and single-clock stochastic
timed automata are almost-surely fair.

We also establish the exact complexity of the almost-sure
model-checking problem in the three above cases. More precisely, we
prove that the almost-sure model-checking problem is (i)
\PSPACE-complete on stochastic timed automata against safety
properties, (ii) \PSPACE-complete (resp. \NLOGSPACE-complete) on
single-clock stochastic timed automata against properties given as
\LTL formulas (resp. $\omega$-regular properties), and (iii)
\PSPACE-complete on weak reactive stochastic timed automata against
$\omega$-regular properties. We finally extend this last result to
specifications given as deterministic timed automata. Let us point out
that the decidability status of the almost-sure model checking problem
for \LTL properties on the general class of stochastic timed automata
is still an open problem.

\subsection*{A model which relaxes timed automata assumptions.}
Let us mention that one initial motivation for defining stochastic
timed automata was the robustness of timed systems. Indeed, the model
of timed automata is an idealised mathematical model, which makes
strong assumptions on the behaviour of the represented real system: it
assumes for instance infinite precision of the clocks, instantaneous
events and communications, whereas a real system will have slightly
different behaviours (like measure time with digital clocks). This
topic of research is very rich, and many models and results have
already been described.

We review some of the frameworks which have been studied in this
context, but will not give a long list of references. We better point
to a survey made in 2011~\cite{markey11}, and to a recent PhD
thesis~\cite{sankur13}, which review in details the literature on the
subject. Let us first mention two models of implementable controllers
proposed in~\cite{DDR04} and in~\cite{SBM11}, where constraints and
precision of clocks are somewhat relaxed.  In this framework, if the
model satisfies a property, then, on a simple model of processor, its
implementation will also satisfy this property. This implementation
model induces a very strong notion of robustness, suitable for really
critical systems (like rockets or X-by-wire systems in cars), but
maybe too strong for less critical systems (like mobile phones or
network applications).  Another robustness model has been proposed at
the end of the 90's in~\cite{GHJ97} with the notion of tube
acceptance: a metric is put on the set of traces of the timed
automaton, and a trace is robustly accepted if and only if a tube
around that trace is classically accepted.  This language-focused
notion of acceptance is however not completely satisfactory for
implementability issues, because it does not take into account the
structure of the automaton.

In this context, the model of stochastic timed automata alleviates
some disadvantages to the strong mathematical assumptions made in
timed automata.  First, randomising delays and the choice of
transitions removes unlikely behaviours (like those requiring
satisfaction of very precise clock constraints), and only important
and meaningful sets of behaviours are then taken into account in the
verification process. Then, the assumptions made in timed automata
mentioned above lead to the existence of unreal(istic) behaviours of
the model, such as Zeno behaviours\footnote{That is, time-converging
  behaviours.}  that one would like to ignore. We will then realise
that, unless the underlying timed automaton is inherently Zeno, the
probability of Zeno behaviours will be $0$ (at least in the classes of
models we have identified). This allows us to convincingly claim that
stochastic timed automata can be used as a possible solution for
relaxing side-effects of mathematical assumptions made in timed
automata.

As a motivating example, we describe a model of the IPv4 Zeroconf
protocol using stochastic timed automata, see
Figure~\ref{fig:IPv4}. This protocol aims at configuring IP addresses
in a local network of appliances.  When a new appliance is plugged, it
selects an IP address at random, and broadcasts several probe messages
to the network to know whether this address is already used or not. If
it receives in a bounded delay an answer from the network informing
that the IP is already used, then a new IP address is chosen. It may
be the case that messages get lost, in which case there is an error.
In~\cite{BvSHV03}, a simple model for the IPv4 Zeroconf protocol is
given as a discrete-time Markov chain, which abstracts away timing
constraints. Using stochastic timed automata, expressing the delay
bound is feasible.
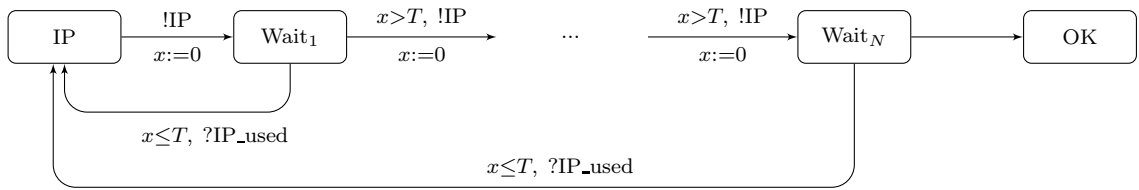
\begin{figure}[h]
\begin{center}
\begin{tikzpicture}
\everymath{\scriptstyle}
\draw (0,0) node [draw,inner sep=2.5pt,rounded corners=3pt,minimum width=1.5cm,minimum height=.7cm] (s0) {$\textrm{IP}$};
\path (-.15,0) node [minimum width=1.5cm,minimum height=.7cm] (s0bis) {};
\draw (3,0) node [draw,inner sep=2.5pt,rounded corners=3pt,minimum width=1.5cm,minimum height=.7cm] (s1) {$\textrm{Wait}_1$};
\draw (6.5,0) node [inner sep=2.5pt,rounded corners=3pt,minimum width=1.5cm,minimum height=.7cm] (s20) {};
\draw (7,0) node [inner sep=2.5pt,rounded corners=3pt,minimum width=1.5cm,minimum height=.7cm] (s21) {};
\draw (10.5,0) node [draw,inner sep=2.5pt,rounded corners=3pt,minimum width=1.5cm,minimum height=.7cm] (s3) {$\textrm{Wait}_N$};
\draw (13.5,0) node [draw,inner sep=2.5pt,rounded corners=3pt,minimum width=1.5cm,minimum height=.7cm] (s4) {$\textrm{OK}$};
\path(6.75,0) node {$\dots$};

\draw [-latex'] (s0) -- (s1) node[pos=0.5,above] {$\textrm{!IP}$} node [pos=.5,below]{$x:=0$};
\draw [-latex'] (s1) -- (s20) node[pos=0.5,above] {$x>T,\ \textrm{!IP}$} node [pos=.5,below]{$x:=0$};
\draw [-latex'] (s21) -- (s3) node[pos=0.5,above] {$x>T,\ \textrm{!IP}$}node [pos=.5,below]{$x:=0$};
\draw [-latex'] (s3) -- (s4) node[pos=0.5,above] {};

\draw[-latex',rounded corners=3mm] (s1) |- +(-1,-1) node[pos=1,below=2pt] {$x \le T,\ \textrm{?IP\_used}$} -| (s0) node[pos=.5,above right] {};

\draw[-latex',rounded corners=3mm] (s3) |- +(-2,-2) node[pos=1.2,above
left] {$x \le T,\ \textrm{?IP\_used}$} -| (s0bis) node[pos=.5,above right] {};
\end{tikzpicture}
\end{center}
\caption{Modelling the IPv4 Zeroconf using stochastic timed automata.\label{fig:IPv4}}
\end{figure}

The example of Figure~\ref{fig:IPv4} illustrates an important feature
of stochastic timed automata.  Compared to CTMC-like models,
stochastic timed automata allow one to express hard timing contraints
such as deadlines (constraint $x \le T$ in this example). Another
important feature of stochastic timed automata, as we will show in
this article, is that the almost-sure satisfaction of properties is
independent of the precise probability distributions over delays. This
is a major advantage since it avoids the problem of finding realistic
probability distributions which is known to be a difficult task, see
\emph{e.g.}~\cite{BDE+14}.

\subsection*{The Cantor topology: a useful tool.}
In a former paper~\cite{VV06}, Varacca and V{\"o}lzer show a strong
correspondence between a standard Markov-chain-based probabilistic
semantics of a finite automaton, and the Cantor topology over the set
of infinite executions of this automaton.  They show in particular
that almost-sure sets of executions (that is, sets of executions which
have probability $1$) coincide with topologically large sets of
executions. Following this idea we also define a topological semantics
\textit{\`a la} Cantor for a timed automaton. In our framework, the
above equivalence does not always hold, but in many cases however we
will be able to prove it. This characterisation is incredibly useful
in order to prove our results. The key tool in our techniques is a
topological game called Banach-Mazur game~\cite{oxtoby57}.

\subsection*{Related work.}
The literature on stochastic processes is huge. We already mentioned
several related works, but we would like to discuss a bit more the
works~\cite{DHS09,CHKM11}, which we think are the closest to the
current article. In both papers the model is that of CTMCs. Timing
constraints are expressed in the properties, either given as
deterministic timed automata~\cite{CHKM11} or as an extension of \CSL
called \CSLta~\cite{DHS09}, which extends \CSL with properties given as
single-clock deterministic timed automata.

Paper~\cite{CHKM11} is interested in quantitative model-checking, that
is, given a CTMC $\mathcal{C}$ and a property given as a deterministic
(Muller) timed automaton $\A$, the aim is to compute the probability
that runs of $\mathcal{C}$ are accepted by $\A$. This probability is
characterised using Volterra integral equations, which can be
transformed into linear equations when $\A$ has a unique
clock. Therefore quantitative verification can be done for
single-clock specifications but can only be approximated in the
general case. Our results are somehow incomparable since we allow for
a more general model (stochastic timed automata instead of CTMCs) but
prove decidability only for the qualitative model-checking problem.

Paper~\cite{DHS09} is interested in model-checking of CTMCs against
properties expressed as formulas of \CSLta. This logic involves
probability formulas, and uses single-clock deterministic timed
automata as predicates. Model-checking of the general logic can be
approximated, but if formulas only have qualitative subformulas, the
exact model-checking can be decided. We do not consider logics, but we
allow general deterministic timed automata in our specifications.

\subsection*{Organisation}  
Section~\ref{sec:prelim} summarises our notations for timed automata,
and specifications languages (such as \LTL and $\omega$-regular
properties).  Section~\ref{sec:proba-def} presents stochastic timed
automata, the notion of almost-sure satisfaction and the almost-sure
model-checking problem, while Section~\ref{sec:toposem} presents the
topological semantics and the notion of large satisfaction. In
Section~\ref{sec:thickgraph}, we define a finite abstraction of a
stochastic timed automaton, named \emph{thick graph}, which will be
essential in order to solve the almost-sure model-checking problem.
In Section~\ref{sec:match}, we show that the topological and the
probabilistic semantics coincide first if we restrict to safety
properties and then for $\omega$-regular properties but under the
restriction that the system is almost-surely fair. In
Section~\ref{sec:appli}, we identify two subclasses of stochastic
timed automata which are almost-surely fair, namely weak reactive and
single-clock. Finally, the algorithmic issues and the complexity
results are given in Section~\ref{sec:algo}. To improve readability of
the article, technical proofs are postponed to the Appendix.

This article presents results from~\cite{BBBBG07,BBBBG08,BBJM12} in a
uniform way, provides the complete proofs, and generalises the results
from~\cite{BBJM12} to a larger class of stochastic timed automata.

\section{Preliminaries}\label{sec:prelim}
\subsection{The timed automaton model}
We denote by $X = \{x_1, \ldots, x_k\}$ a finite set
of~\emph{clocks}. A \emph{clock valuation} over $X$ is a mapping $\nu
: X \to \IR_+$, where $\IR_+$ denotes the set of nonnegative reals. We
write $\IR_+^X$ for the set of clock valuations over $X$, and
$\mathbf{0}_X$ (or simply $\mathbf{0}$ if $X$ is clear in the context)
for the valuation assigning $0$ to every clock of $X$. Given a clock
valuation $\nu$ and $\tau \in \IR_+$, $\nu + \tau$ is the clock
valuation defined by $(\nu+\tau)(x) = \nu(x)+\tau$ for every $x \in
X$. If $Y \subseteq X$, the valuation $[Y \leftarrow 0] \nu$ is the
valuation $\nu'$ such that $\nu'(x) = 0$ if $x \in Y$, and $\nu'(x) =
\nu(x)$ otherwise. A \emph{guard} over $X$ is a finite conjunction of
expressions of the form $x \sim c$ where $x \in X$ is a clock, $c \in
\IN$ is an integer, and $\sim$ is one of the symbols $\{ \mathord< ,
\mathord\le , \mathord= , \mathord\ge , \mathord>\}$. We denote by
${\mathcal G}(X)$ the set of guards over $X$. The satisfaction relation
for guards over clock valuations is defined in a natural way, and we
write $\nu \models g$ if the clock valuation $\nu$ satisfies the guard
$g$. We denote by $\AP$ a finite set of atomic propositions.

We now define the timed automaton model, which has been introduced in
the early nineties~\cite{AD90,AD94}.

\begin{defi}\label{def:timed_aut}
  A \emph{timed automaton} over \AP is a tuple $\A = (L,X,E,{\mathcal
    I},{\mathcal L})$ such that: (i) $L$ is a finite set of locations,
  (ii) $X$ is a finite set of clocks, (iii) $E \subseteq L \times
  {\mathcal G}(X) \times 2^X \times L$ is a finite set of edges, (iv)
  ${\mathcal I} : L \to {\mathcal G}(X)$ assigns an invariant to each
  location, and $(v)$ ${\mathcal L} : L \to 2^{\AP}$ is a labelling
  function.
\end{defi}
We may omit the labelling function (in case we are only interested in
an internal accepting condition, \emph{i.e.} that only depends on the
locations). Note that we could also specify an initial location, but
that will not be really useful later, that is why we removed that
component from standard timed-automata definition.

If $e$ is an edge of $\A$, we write $\mathsf{source}(e)$
(resp. $\mathsf{target}(e)$) the source (resp. target) of $e$ defined
by $\ell$ (resp. $\ell'$) if $e = (\ell,g,Y,\ell')$.  The semantics of
a timed automaton $\A$ is a timed transition system $T_{\A}$ whose
states are pairs $s= (\ell,v) \in L \times \IR_+^{X}$ with $v \models
{\mathcal I}(\ell)$, and whose transitions are of the form $(\ell,v)
\xrightarrow{\tau,e} (\ell',v')$ if there exists an edge $e =
(\ell,g,Y,\ell')$ such that for every $0 \leq \tau' \leq \tau$,
$v+\tau' \models {\mathcal I}(\ell)$, $v+\tau \models g$, $v' = [Y
\leftarrow 0]v$, and $v' \models {\mathcal I}(\ell')$. We extend the
labelling function $\mathcal{L}$ to states: $\mathcal{L}((\ell,v)) =
\mathcal{L}(\ell)$ for every state $(\ell,v)$.  A finite
(resp. infinite) \emph{run} $\rho$ of $\A$ is a finite (resp.
infinite) sequence of transitions, \textit{i.e.},
\[
\rho = s_0 \xrightarrow{\tau_1,e_1} s_1 \xrightarrow{\tau_2,e_2} s_2
\ldots
\]
We write $\Runs_f(\A,s_0)$ (resp. $\Runs(\A,s_0)$) for the set of
finite runs (resp.  infinite runs) of $\A$ from state~$s_0$. If $\rho$
is a finite run in $\A$, we write $\mathsf{last}(\rho)$ for the last
state of $\rho$.  For $s$ a state of $\A$, $(e_i)_{1 \leq i \leq n}$ a
finite sequence of edges of $\A$, and $\mathcal{C}$ a constraint over
$n$ variables $(t_i)_{1 \leq i \leq n}$, the \emph{(symbolic) path}
starting from $s$, determined by $(e_i)_{1 \leq i \leq n}$, and
constrained by $\mathcal{C}$, is the following set of runs:
\[
\path[\mathcal{C}]{s,e_1\ldots e_n} \ = \ \{ \rho = s
\xrightarrow{\tau_1,e_1} s_1 \ldots \xrightarrow{\tau_n,e_n} s_n \in
\Runs_f(\A,s)\ \mid (\tau_i)_{1 \leq i \leq n} \models
\mathcal{C}\}\,,
\]
where $ (\tau_i)_{1 \leq i \leq n} \models \mathcal{C}$ stands for
``$\tau_i$'s satisfy the constraint $\mathcal{C}$'' with the intuitive
meaning.

If $\mathcal{C}$ is equivalent to `true', we simply write
$\path{s,e_1\ldots e_n}$. Let $\pi_{\mathcal{C}} =
\path[\mathcal{C}]{s,e_1 \ldots e_n}$ be a finite symbolic path, we
define the \emph{cylinder} generated by $\pi_{\mathcal{C}}$ as:
\[
\Cyl(\pi_{\mathcal C}) = \{\rho \in \Runs(\A,s) \mid \exists \rho' \in
\Runs_f(\A,s), \text{finite prefix of}\ \rho,\ \text{s.t.}\ \rho' \in
\pi_{\mathcal{C}}\}\,.
\]
In the following, we will also use infinite symbolic paths defined,
given $s$ a state of $\A$ and $(e_i)_{i \ge 1}$ an infinite sequence
of edges, as:
\[
\path{s,e_1\ldots} = \{ \rho = s \xrightarrow{\tau_1,e_1} s_1 \ldots
\in \Runs(\A,s)\}\,.
\]
If $\varrho \in \Runs(\A,s)$, we write $\pi_\varrho$ for the unique
symbolic path containing $\varrho$.  Given $s$ a state of $\A$ and $e$
an edge, we define $I(s,e) = \{\tau \in \IR_+ \mid \exists s' \textrm{
  s.t. } s \xrightarrow{\tau,e} s'\}$ and $I(s) = \bigcup_e
I(s,e)$. Note that $I(s,e)$ is an interval, whereas $I(s)$ is a finite
union of intervals.

The timed automaton $\A$ is \emph{non-blocking} if, for every state
$s$, $I(s) \neq \emptyset$.  The timed automaton $\A$ is
\emph{reactive} if, for every state $s$, $I(s)=\IR_+$; in this case we
may omit the invariant function and simply write $\A = (L,X,E,{\mathcal
  L})$.

\subsection{The timed region automaton}
\label{regionautomaton}
The well-known region automaton construction is a finite abstraction
of timed automata which can be used for verifying many properties like
$\omega$-regular untimed properties~\cite{AD94}.  Roughly, the region
automaton of $\A$ is the quotient of $T_{\A}$ by a finite-index
equivalence relation over clock valuations. Here we will use a timed
version of this construction, that we define now.

Let $\A = (L,X,E,{\mathcal I},{\mathcal L})$ be a timed automaton, and write
$M$ for the maximal constant used in guards and invariants in $\A$.
We define its region equivalence $\equiv_{\A}$ over the set of
valuations $\IR_+^X$ as follows: given $v,v' \in \IR_+^X$, $v
\equiv_{\A} v'$ if and only if the following conditions are
satisfied:
\begin{itemize}
\item for every $x \in X$, either $v(x),v'(x)>M$, or $\lfloor v(x)
  \rfloor = \lfloor v'(x) \rfloor$, and in the last case, $\{v(x)\}=0$
  iff $\{v'(x)\}=0$;\footnote{$\lfloor \cdot \rfloor$
    (resp. $\{\cdot\}$) denotes the integral (resp. fractional) part.}
\item for every $x,y \in X$ such that $v(x),v(y) \le M$, $\{v(x)\} \le
  \{v(y)\}$ iff $\{v'(x)\} \le \{v'(y)\}$.
\end{itemize}
This equivalence relation has finite index, the equivalence classes
are called \emph{regions}, and we write $R_{\A}$ for the set of
regions. If $v$ is a valuation, we write $[v]_{\A}$ or simply $[v]$
for the (unique) region to which $v$ belongs. Also, for $r$ a region,
$\cell(r)$ denotes the minimal guard characterising $r$.

\begin{rem}
  The above region equivalence is the most standard one, but several
  rougher equivalences could also be used, as soon as they yield a
  time-abstract bisimilar quotient. For instance, for single-clock
  timed automata, we will later use a rougher notion of region
  equivalence~\cite{LMS04} that will improve the complexity of our
  algorithms.
\end{rem}

The \emph{timed region automaton} of $\A$ is the timed automaton $\RA
= (Q,X,T,\kappa,\lambda)$ such that $Q = L \times R_{\A}$, and:
\begin{itemize}
\item $\kappa((\ell,r)) = {\mathcal I}(\ell)$, and $\lambda((\ell,r)) =
  {\mathcal L}(\ell)$ for all $(\ell,r) \in L \times R_{\A}$;
\item $T \subseteq (Q \times \cell(R_{\A}) \times E \times 2^X \times
  Q)$, and $(\ell,r) \xrightarrow{\cell(r''),e,Y} (\ell',r')$ is in
  $T$ iff there exists $e = \ell \xrightarrow{g,Y} \ell'$ in $E$
  s.t. there exist $v \in r$, $\tau \in \IR_{+}$ with $(\ell,v)
  \xrightarrow{\tau,e} (\ell',v')$, $v+\tau \in r''$ and $v' \in
  r'$.
\end{itemize}
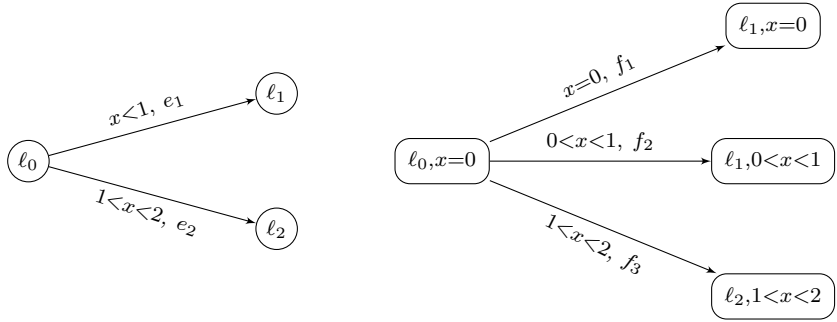
\begin{figure}[hbt]
  \begin{center}
    \begin{tikzpicture}[yscale=.9,xscale=1.1]
      \everymath{\scriptstyle} 
      \path (-1,0) node[draw,circle,inner sep=2pt] (q0) {$\ell_0$}; 
      \path (2,1) node[draw,circle,inner sep=2pt] (q1) {$\ell_1$}; 
      \path (2,-1) node[draw,circle,inner sep=2pt] (q2) {$\ell_2$};
      \draw [arrows=-latex'] (q0) -- (q1) node[pos=.5, above,sloped] {$x < 1,~e_1$};
      \draw [arrows=-latex'] (q0) -- (q2) node[pos=.5, below,sloped] {$1< x < 2,~e_2$};

      \path (4,0) node[draw,rectangle,rounded corners=2mm,inner sep=5pt] (q00)
      {$\ell_0,x=0$};
      \path (8,2) node[draw,rectangle,rounded corners=2mm,inner sep=5pt] (q11)
      {$\ell_1,x=0$};
      \path (8,0) node[draw,rectangle,rounded corners=2mm,inner sep=5pt] (q12)
      {$\ell_1,0<x<1$};
      \path (8,-2) node[draw,rectangle,rounded corners=2mm,inner sep=5pt] (q21)
      {$\ell_2,1<x<2$};
      \draw [arrows=-latex'] (q00) -- (q11) node[pos=.5, above,sloped] {$x = 0,~f_1$};
      \draw [arrows=-latex'] (q00) -- (q12) node[pos=.5, above,sloped] {$0 <x < 1,~f_2$};
      \draw [arrows=-latex'] (q00) -- (q21) node[pos=.5, below,sloped] {$1< x < 2,~f_3$};
    \end{tikzpicture}
  \end{center}
\caption{An automaton and its timed region automaton.\label{fig:cex-homeo}}
\end{figure}
As an example, a timed automaton and its associated timed region
automaton are depicted in Figure~\ref{fig:cex-homeo}. We recover the
usual region automaton of~\cite{AD94} by labelling the transitions
with `$e$' instead of `$\cell(r''),e,Y$', and by interpreting $\RA$ as
a finite automaton.  The above timed interpretation satisfies strong
timed bisimulation properties that we do not detail here.  To every
finite symbolic path $\path{(\ell,v),e_1 \ldots e_n}$ in $\A$
corresponds a finite set of paths $\path{((\ell,[v]),v),f_1 \ldots
  f_n}$ in $\RA$, each one corresponding to a choice in the regions
that are crossed. If $\varrho$ is a run in $\A$, we denote
$\iota(\varrho)$ its unique image in $\RA$.
\label{def:iota} 
Note that if $\A$ is non-blocking (resp. reactive), then so is~$\RA$.

\medskip In the rest of the paper we assume that timed automata are
non-blocking, even though general timed automata could also be handled
(but at a technical extra cost). 

\subsection{Specification languages}
\label{subsec:spec}

We fix a finite set of atomic propositions \AP, and a timed automaton
$\A = (L,X,E,\mathcal{I},\mathcal{L})$ over \AP.

\subsubsection{Properties over \AP.}
A \emph{property over \AP} is a subset $P$ of
$\left(2^\AP\right)^\omega$. An infinite run $\rho = s_0
\xrightarrow{\tau_1,e_1} s_1 \xrightarrow{\tau_2,e_2} \ldots$
satisfies the property whenever $\mathcal{L}(s_0) \mathcal{L}(s_1)
\mathcal{L}(s_2) \ldots \in P$.

More generally a \emph{timed property over \AP} is a subset $P$ of
$2^\AP \cdot \left(\mathbb{R}_+ \cdot 2^\AP\right)^\omega$. An
infinite run $\rho = s_0 \xrightarrow{\tau_1,e_1} s_1
\xrightarrow{\tau_2,e_2} \ldots$ satisfies the property $P$ whenever
$\mathcal{L}(s_0)\, \tau_1\, \mathcal{L}(s_1)\, \tau_2\,
\mathcal{L}(s_2) \ldots \in P$.

In both cases, we write $\varrho \models P$ if $\varrho$ satisfies the
property $P$, and we write:
\[
\sem{P}_{\A,s} \egdef \{\varrho \in \Runs(\A,s) \mid \varrho \models
P\}\,.
\]

\begin{rem}
  Obviously, timed properties generalise (untimed) ones. 
\end{rem}

\subsubsection{$\omega$-regular properties.}
$\omega$-regularity is a standard notion in computer science to
characterise simple sets of infinite behaviours. We will only define
here $\omega$-regularity for untimed properties, though the concept
exists for timed properties as well.

Typical $\omega$-regular properties are B\"uchi and Muller properties.
A \emph{B\"uchi property} over \AP is a(n untimed) property $P$ such
that there exists $F \subseteq \AP$ with $P = \{u_0 u_1 \ldots \mid
\{j \mid F \cap u_j \ne \emptyset\}\ \text{is infinite}\}$.  A
\emph{Muller property} over \AP is a property $P$ such that there
exists $\mathcal{F} \subseteq 2^\AP$ with $P = \{u_0 u_1 \ldots \mid
\{j \mid u_j \in \mathcal{F}\}\ \text{is infinite}\}$.

An $\omega$-regular property will be said \emph{internal} for $\A$
whenever there is a bijection $\beta$ between $L$ and $\AP$, and for
each $\ell \in L$, $\mathcal{L}(\ell)$ is the singleton
$\{\beta(\ell)\}$. That is, this allows to specify which states are
visited infinitely often. In that case, we will interpret such
properties on timed automata even though no labelling function has
been given (it is then implicit).  It is well known that (untimed)
automata equipped with internal B\"uchi or Muller acceptance
conditions capture untimed $\omega$-regular properties. This is also
the case for deterministic Muller automata.

\subsubsection{Safety, reachability, and prefix-independent properties.}
We now define simple $\omega$-regular properties.

According to~\cite{CMP-icalp92}, a property $P$ over \AP is a
\emph{safety property} whenever for every $w = u_0 u_1 \ldots \in
\left(2^\AP\right)^\omega$, $w \notin P$ iff there exists $i$ such
that for every $w' = u_0 \ldots u_i u'_{i+1} u'_{i+2} \ldots \in
\left(2^\AP\right)^\omega$, $w' \notin P$. That is, a safety property
is violated by a finite prefix. A \emph{simple safety property} $P$ is
characterised by $F \subseteq \AP$, and is defined by $P = \{u_0 u_1
\ldots \mid \forall j,\ u_j \in F\}$.

The negation of a safety property is a \emph{reachability property}:
$P$ is a reachability property whenever for every $w = u_0 u_1 \ldots
\in \left(2^\AP\right)^\omega$, $w \in P$ iff there exists $i$ such
that for every $w' = u_0 \ldots u_i u'_{i+1} u'_{i+2} \ldots \in
\left(2^\AP\right)^\omega$, $w' \in P$. That is, a reachability
property is validated by a finite prefix. A \emph{simple reachability
  property} $P$ is characterised by $F \subseteq \AP$, and is defined
by $P = \{u_0 u_1 \ldots \mid \exists j,\ u_j \in F\}$.

Another interesting notion is the one of \emph{prefix-independent
  property} $P$, which is such that for every $w = u_0 u_1 \ldots \in
\left(2^\AP\right)^\omega$, $w \in P$ iff for every $i$, $w' = u_{i}
u_{i+1} \ldots \in P$. That is, the property is satisfied or not
independently of its prefix. In particular, B\"uchi and Muller
properties are prefix-independent.  Note that a property $P$ is
prefix-independent if and if its validity only depends on the set of
elements of $2^{\AP}$ which is encountered infinitely often.

\subsubsection{The temporal logic \LTL.}
We consider the linear temporal logic \LTL~\cite{pnueli77} over
\AP, defined inductively as:
\[
\LTL \ni \varphi\ ::=\ p\ \mid\ \varphi \vee \varphi\ \mid\ \varphi
\wedge \varphi\ \mid\ \neg \varphi\ \mid\ \varphi \U \varphi
\]
where $p \in \AP$ is an atomic proposition. We use classical
shorthands like $\true \egdef p \vee \neg p$, $\false \egdef p \wedge
\neg p$, $\F \varphi \egdef \true \U \varphi$, and $\G \varphi \egdef
\neg \F (\neg \varphi)$. We assume the reader is familiar with the
semantics of \LTL, that we interpret here on infinite runs of a timed
automaton.

Each formula $\varphi$ of \LTL gives rise to a property $P_\varphi$,
in the sense given above.  Let $w = u_0 u_1 \ldots \in
\left(2^\AP\right)^\omega$, then:
\begin{eqnarray*}
  w \in P_{p} & ~\Leftrightarrow~ & p \in u_0 \\
  w \in P_{\varphi_1 \vee \varphi_2} & ~\Leftrightarrow~ & w \in P_{\varphi_1} \cup P_{\varphi_2} \\
  w \in P_{\varphi_1 \wedge \varphi_2} & ~\Leftrightarrow~ & w \in P_{\varphi_1} \cap P_{\varphi_2} \\
  w \in P_{\neg \varphi} & ~\Leftrightarrow~ & w \notin P_\varphi \\
  w \in P_{\varphi_1 \U \varphi_2} & ~\Leftrightarrow~ & \exists i \ge 0\ \text{s.t.}\ w_{\ge i} \in P_{\varphi_2}\ \text{and}\ \forall 0 \le j <i,\ w_{\ge j} \in P_{\varphi_1}
\end{eqnarray*}
where $w_{\ge k} = u_k u_{k+1} \ldots$ for every index $k$.

The semantics of a formula $\varphi$ over infinite runs of $\A$ is
derived from that of property $P_\varphi$. One can easily be convinced
that we recover the standard semantics of \LTL. If $\varphi$ is an
\LTL formula and $\rho \in \Runs(\A,s)$, we write $\rho \models
\varphi$ whenever $\rho \models P_\varphi$. We also write
$\sem{\varphi}_{\A,s}$ for $\sem{P_\varphi}_{\A,s}$.

\subsubsection{Specifications given as deterministic timed automata.}
A \emph{specification $\omega$-regular timed automaton} is a tuple $\B
= (\mathsf{L},\mathsf{i}_0,\mathsf{X},\AP,\mathsf{E},\mathcal{F})$
such that:
\begin{itemize}
\item $\mathsf{L}$ is a finite set of locations, and $\mathsf{i}_0:
  2^{\AP} \to \mathsf{L}$ is an input function;
\item $\mathsf{X}$ is a finite set of clocks;
\item $\AP$ is a finite set of atomic propositions;
\item $\mathsf{E} \subseteq \mathsf{L} \times \mathcal{G}(\mathsf{X})
  \times 2^\AP \times 2^{\mathsf{X}} \times \mathsf{L}$ is a finite
  set of edges;
\item $\mathcal{F}$ is an internal $\omega$-regular
  prefix-independent condition;
\item it is deterministic: for all edges $(\mathsf{l}
  \xrightarrow{\mathsf{g}_1,u,\mathsf{Y}_1} \mathsf{l}_1)$ and
  $(\mathsf{l} \xrightarrow{\mathsf{g}_2,u,\mathsf{Y}_2}
  \mathsf{l}_2)$ in $\mathsf{E}$, $\mathsf{g}_1 \wedge \mathsf{g}_2$
  is not satisfiable;
\item it is complete: for every every $\mathsf{l} \in \mathsf{L}$, for
  every $u \in 2^{\AP}$, for every $\mathsf{v} \in
  \mathbb{R}_+^{\mathsf{X}}$, for every $\tau \in \mathbb{R}_+$, there
  exists $(\mathsf{l} \xrightarrow{\mathsf{g},u,\mathsf{Y}}
  \mathsf{l}') \in \mathsf{E}$ such that $\mathsf{v}+\tau \models
  \mathsf{g}$.
\end{itemize}
Runs in $\B$ will be defined in a very similar way as runs in standard
timed automata. Only labels of transitions will be slightly different.
The runs of $\B$ are therefore of the form:
\[
(\mathsf{l}_0,\mathsf{v}_0) \xrightarrow{\tau_1,u_1}
(\mathsf{l}_1,\mathsf{v}_1) \xrightarrow{\tau_2,u_2} \ldots
\]
where conditions on valuations are those expected, and labels $u_i$'s
are those given by the edges that are taken.  Such a run is accepted
by $\B$ whenever the sequence $(\mathsf{l}_i)_{i \ge 0}$ satisfies the
$\omega$-regular condition $\mathcal{F}$.

Such a specification automaton $\B$ naturally gives rise to a timed
property $P_{\B}$ defined as follows. Let $w = u_0 \tau_1 u_1 \tau_2
\ldots \in 2^{\AP} \cdot \left(\mathbb{R}_+ \cdot
  2^{\AP}\right)^\omega$. There is a unique run $\kappa_w =
(\mathsf{l}_0,\mathbf{0}_{\mathsf{X}}) \xrightarrow{\tau_1,u_1}
(\mathsf{l}_1,\mathsf{v}_1) \xrightarrow{\tau_2,u_2} \ldots$ in
automaton $\B$ where $\mathsf{l}_0 = \mathsf{i}_0(u_0)$. The existence
of $\kappa_w$ follows from the completeness of $\B$ and its uniqueness
from the determinism. Then, $w \in P_{\B}$ iff $\kappa_w$ is an
accepting run in $\B$.

The semantics of a specification timed automaton $\B$ over infinite
runs of $\A$ is derived from that of property $P_\B$. If $\B$ is a
specification timed automaton and $\rho \in \Runs(\A,s)$, we write
$\rho \models \B$ whenever $\rho \models P_\B$. We also write
$\sem{\B}_{\A,s}$ for $\sem{P_\B}_{\A,s}$.

\begin{rem}
  If the accepting condition $\mathcal{F}$ is a B\"uchi (or Muller)
  condition, then $\B$ will be called a specification B\"uchi (or
  Muller) timed automaton. If $\mathsf{X} = \emptyset$, we will speak
  of a specification $\omega$-regular (untimed) automaton.  It is well
  known~\cite{VW94} and \cite[Chapter~3]{lncs2500} that for any \LTL
  formula $\varphi$, there is a (deterministic) specification Muller
  untimed automaton $\B_\varphi$ that characterises $\varphi$, that
  is: for every run $\varrho$, $\varrho \models \varphi$ iff $\varrho
  \models \B_\varphi$. In that case, obviously, $\sem{\varphi}_{\A,s}
  = \sem{\B_\varphi}_{\A,s}$ for every $\A$ and $s$.
\end{rem}

\section{Stochastic timed automata: the semantics}
\label{sec:proba-def}

In this section we define a probabilistic semantics for timed
automata.  Probabilities will rule time elapsing as well as choices
between enabled events. This will define a purely stochastic process:
intuitively, from a state, we will first randomly choose a delay among
all possible delays, then we will randomly choose an edge among all
those which are enabled.

The sequel assumes some basics of measure theory and probabilities,
that can be found in classical text books. 

\subsection{Probability measure on runs of a timed automaton}
\label{subsec:proba}
Let $\A = (L,X,E,{\mathcal I},{\mathcal L})$ be a timed automaton. We will
assign probability distributions from every state of $\A$ both over
delays and over enabled moves. Let $s$ be a state of $\A$. The
probability distribution from $s$ over delays is a probability measure
$\mu_s$ over $\IR_+$ (equipped with the standard Borel
$\sigma$-algebra) which satisfies the following requirements, denoted
$(\star)$ in the sequel:
\begin{description}
\item[(H1)] $\mu_s(I(s)) = \mu_s(\IR_+)= 1$,\footnote{Note that this is
    possible, as we assume $\A$ is non-blocking, hence $I(s) \neq
    \emptyset$ for every state $s$ of $\A$.}
\item[(H2)] Writing $\lambda$ for the standard Lebesgue measure
  on $\IR_+$, if
  $\lambda(I(s))>0$, then $\mu_s$ is equivalent\footnote{Two measures
    $\nu$ and $\nu'$ are \emph{equivalent} whenever for each
    measurable set $A$, $\nu(A) =0 \Leftrightarrow \nu'(A)=0$.} to
  $\lambda$ on $I(s)$; Otherwise, $\mu_s$ is equivalent to the uniform
  distribution over points of $I(s)$.
\end{description}
This last condition denotes some kind of fairness w.r.t. enabled
transitions when only punctual delays are possible, in that we cannot
disallow one transition by putting a probability $0$ to delays
enabling that transition.

We also assume a probability distribution $p_s$ over edges, such that
for every edge $e$, $p_s(e) > 0$ iff $e$ is enabled in $s$
(\textit{i.e.}, $s \xrightarrow{e} s'$ for some $s'$).  Moreover, to
simplify, we assume that $p_s$ is given by weights on transitions, as
it is classically done for resolving non-determinism: we associate
with each edge $e$ a weight $w(e) > 0$, and for every state $s$, for
every edge $e$, $p_s(e) = 0$ if $e$ is not enabled in $s$, and $p_s(e)
= w(e)/(\sum_{e'\ \text{enabled in}\ s} w(e'))$ otherwise. As a
consequence, if $s$ and $s'$ are region equivalent, then for every
edge $e$, $p_s(e) = p_{s'}(e)$.

\begin{defi}[Stochastic timed automaton]
\label{def:sta}
A \emph{stochastic timed automaton} is a tuple $\langle \A,\mu,w
\rangle$ consisting of a timed automaton $\A$ equipped with
probability measures $\mu= (\mu_s)_{s \in L \times \IR_+^X}$
satisfying $(\star)$, and positive weights $w = (w_e)_{e \in E}$.
\end{defi}
Note that when measures are clear or implicit in the context, we will
simply write $\A$ for $\langle \A,\mu,w \rangle$. We now show how we
define a probability measure $\Prob_\A$ on infinite runs of $\langle
\A,\mu,w \rangle$. 

We fix a stochastic timed automaton $\langle \A,\mu,w \rangle$, which
satisfies $(\star)$. We define a measure, that we also note
$\Prob_\A$, over finite symbolic paths from state $s$ as follows:
\[
\Prob_\A(\path{s,e_1\ldots e_n}) \ = \ \int_{t \in I(s,e_1)}
p_{s+t}(e_1) \, \Prob_\A(\path{s_t,e_2 \ldots e_n}) \, \ud \mu_{s}(t)
\]
where $s \xrightarrow{t} (s+t) \xrightarrow{e_1} s_t$, and we
initialise with $\Prob_\A(\path{s}) = 1$.  The formula for $\Prob_\A$
relies on the fact that the probability of taking transition $e_1$ at
time $t$ coincides with the probability of waiting $t$ time units and
then choosing $e_1$ among the enabled transitions, \textit{i.e.},
$p_{s+t}(e_1) \ud \mu_s(t)$.  Note that, time passage and actions are
independent events.

The value $\Prob_\A(\path{s,e_1\ldots e_n})$ is the result of $n$
successive one-dimensional integrals, but it can also be viewed as the
result of an $n$-dimensional integral. Hence, we can easily extend the
above definition to finite constrained paths $\path[\mathcal{C}]{s,e_1
  \ldots e_n}$ when $\mathcal{C}$ is Borel-measurable.  This extension
to constrained paths is needed to measure rather complex properties,
like Zeno behaviours or those expressed as specification timed
automata.  The measure $\Prob_{\A}$ can then be defined on cylinders,
letting $\Prob_{\A}(\Cyl(\pi)) = \Prob_{\A}(\pi)$ if $\pi$ is a finite
(constrained) symbolic path. Finally we extend $\Prob_{\A}$ in a
standard and unique way to the $\sigma$-algebra generated by these
cylinders (using Caratheodory's theorem), that we note
$\Omega_{\A}^s$.

We first check that $\Prob_{\A}$ defined as such is a probability
measure.

\begin{restatable}{prop}{propprobameasure}
  \label{prop:proba_measure}
  For every state $s$ of $\A$, $\Prob_\A$ is a probability measure
  over $(\Runs(\A,s),\Omega_{\A}^s)$.
\end{restatable}
The proof of this proposition justifies \textit{a posteriori} the
above construction for the probability measure $\Prob_{\A}$.  It goes
as follows: first prove that $\Prob_\A$ is a probability measure on
the set of constrained symbolic paths of length $n$ (for all $n$),
then extend this result to the ring generated by all constrained
symbolic paths and finally use Caratheodory's extension theorem to
establish that $\Prob_\A$ is a probability measure on the set of all
runs. The complete proof is rather technical, and therefore postponed
to Appendix~\ref{annex:proba-measure},
page~\pageref{annex:proba-measure}.

\begin{exa}\label{ex:ltlfirst}
  Consider the running stochastic timed automaton
  $\A_{\textsf{running}}$ on Figure~\ref{fig:running}. Assume for all
  states $s_t = (\ell_0,t)$ both uniform distributions over delays and
  discrete moves: $\mu_{s_0} = \lambda$ is the uniform distribution
  over $[0,1]$ and $\mu_{s_t} = \frac{\lambda}{1-t}$ is the uniform
  distribution over $[t,1]$; the weight of each edge is $1$. Then
  $\Prob_{\A_{\textsf{running}}}(\Cyl(\path{s_0,e_1e_1}))=
  \frac{1}{4}$ and
  $\Prob_{\A_{\textsf{running}}}(\path{s_0,{e_1}^\omega}) =
  0$. Indeed:
\[\eqalign{
  \Prob_{\A_{\textsf{running}}}(\path{s_0,e_1e_1}) 
&=\int_{t \in I(s_0,e_1)} p_{s_0+t}(e_1) \,
  \Prob_{\A_{\textsf{running}}}(\path{(\ell_0,t),e_1}) \, \ud
  \mu_{s_0}(t) \cr
&=\int_0^1 \frac{1}{2} \,
  \Prob_{\A_{\textsf{running}}}(\path{(\ell_0,t),e_1}) \, \ud
  \lambda(t) \cr
&=\frac{1}{2} \int_0^1
      \left(\int_{t \in I(s_t,e_1)} p_{s_t+u}(e_1) \,
      \Prob_{\A_{\textsf{running}}}(\path{(\ell_1,u)}) \, \ud
      \mu_{s_t}(u)  \right) \ud \lambda(t) \cr
&=\frac{1}{2} \int_0^1 \left(\int_t^1 \frac{1}{2} \,
        \frac{1}{1-t} \, \ud \lambda(u) \right) \ud \lambda(t) =
      \frac{1}{4}.
  }
\]
  In a similar way one can show that
  $\Prob_{\A_{\textsf{running}}}(\path{s_0,{e_1}^n}) = \frac{1}{2^n}$,
  for $n \in \IN$; and thus conclude that
  $\Prob_{\A_{\textsf{running}}}(\path{s_0,{e_1}^\omega}) = 0$. \corn
\end{exa}

\begin{figure}[htb]
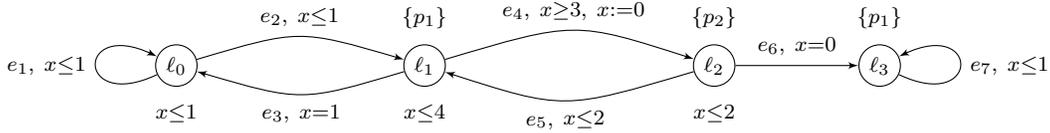

    \begin{center}
      \runningexample
    \end{center}
    \caption{The stochastic timed automaton
      $\A_{\textsf{running}}$.\label{fig:running}}
  \end{figure}

\subsection{Measuring Zeno runs}

In timed automata, and more generally in continuous-time models, some
behaviours are \emph{Zeno}. Recall that a run $\rho=s_0
\xrightarrow{\tau_1,e_1} s_1 \xrightarrow{\tau_2,e_2} \ldots$ of a
timed automaton is \emph{Zeno} if $\sum_{i=1}^\infty \tau_i < \infty$
(\textit{i.e.}, infinitely many actions happen in a finite amount of
time). Zeno behaviours are problematic since they most of the time
have no physical interpretation. As argued in~\cite{DP03}, some
fairness constraints are often put on executions, enforcing non-Zeno
behaviours, but in probabilistic systems, probabilities are supposed
to replace fairness assumptions, and it is actually the case in
continuous-time Markov chains in which Zeno runs are negligible (that
is, have probability~$0$)~\cite{BHHK03}.

We observe that, for any stochastic timed automaton $\A$, the set of
Zeno behaviours from a state $s$ is measurable.  It can indeed be
expressed as
\[
\bigcup_{M \in \mathbb{N}} \bigcap_{n \in \mathbb{N}}
\bigcup_{(e_1,\ldots,e_n) \in E^n} \path[\tau_1+\ldots+\tau_n \le
M]{s,e_1 \ldots e_n}\,.
\]
It therefore makes sense to compute the probability of Zeno
behaviours, and to check whether Zeno behaviours are negligible or
not. Being negligible would be a desirable property, as argued
before. However, in general this is hopeless since some timed automata
are \emph{inherently} Zeno. For instance, all runs are Zeno in the
automaton consisting of a single location with a non-resetting loop
guarded by $x \leq 1$. In the following, we will discuss Zenoness for
several classes of stochastic timed automata.

In the rest of the paper, if $\A$ is a stochastic timed automaton and
$\rho$ is a run of $\A$, we write $\rho \models \Zeno$ whenever $\rho$
is Zeno. If $s$ is a state of $\A$, we then also write $\Prob_{\A}(s
\models \Zeno)$ for the probability of the set of Zeno runs in $\A$
from $s$. We associate the following decision problem, that we call
\emph{almost-sure non-Zenoness} problem: given $\A$ a stochastic timed
automaton, and $s$ a state of $\A$, does $\Prob_{\A}(s \models
\Zeno)=0$?

\subsection{From timed automata to timed region automata}

In this part we establish a strong relation between a stochastic timed
automaton and its stochastic timed region automaton.  We let $\langle
\A,\mu^\A,w^\A \rangle$ be a stochastic timed automaton. The structure
of the corresponding stochastic timed region automaton is obviously
$\RA$. We need now to choose properly probability measures $\mu^{\RA}$
and weights $w^{\RA}$ for $\RA$ so that measures of runs are preserved
\textit{via} the mapping $\iota$ mentioned in
Section~\ref{regionautomaton}. We assume that the probability measures
in $\RA$ satisfy the following conditions: for every state $s$ in
$\A$, $\mu_s^{\A} = \mu_{\iota(s)}^{\RA}$, and for every edge $e \in
E$, $w^{\A}(e) = w^{\RA}(f)$ whenever $f$ corresponds to $e$.  Under
those conditions we show the following transfer properties between
$\A$ and $\RA$.

\begin{restatable}{lem}{lemmaproba}
  \label{lemma:proba}
  Let $\langle \A,\mu^\A,w^\A \rangle$ be a stochastic timed
  automaton, and let $\langle \RA,\mu^{\RA},w^{\RA} \rangle$ be the
  corresponding stochastic timed region automaton as defined above.
  Then, for every set $S$ of runs in $\A$ we have: $S \in
  \Omega_{\A}^s$ iff $\iota(S) \in \Omega_{\RA}^{\iota(s)}$, and in
  this case $\Prob_{\A}(S) = \Prob_{\RA}(\iota(S))$.
\end{restatable}

To establish Lemma~\ref{lemma:proba} it is sufficient to prove that
the measures coincide on finite constrained paths, since it implies
that they agree on cylinders and by uniqueness of the extension on any
measurable set of infinite runs. The complete proof is given in
Appendix~\ref{annex:region}, page~\pageref{annex:region}.

Thanks to Lemma~\ref{lemma:proba}, we will be able to lift results
proven on $\RA$ to $\A$. 

\subsection{Almost-sure satisfaction}

Let $\langle \A,\mu,w \rangle$ be a stochastic timed automaton over
\AP and $s$ be a state of $\A$.

We refine the notion of timed properties that were defined in
Section~\ref{subsec:spec}. Let $P \subseteq 2^\AP \cdot
\left(\mathbb{R}_+ \cdot 2^\AP \right)^\omega$ be a timed property
over \AP. We say that $P$ is \emph{$(\A,s)$-measurable} whenever
$\sem{P}_{\A,s} \in \Omega_{\A}^{s}$.  We say $P$ is
\emph{$\A$-measurable} (resp. \emph{measurable}) whenever it is
$(\A,s)$-measurable for every state $s$ (resp. it is $\A$-measurable
for every $\A$).

The following lemma establishes the measurability of several classes
of properties, and is proven in Appendix~\ref{app:mesurabilite},
page~\pageref{app:mesurabilite}.

\begin{restatable}{lem}{mesurabilite}
  $\omega$-regular properties and properties given as \LTL formulas
  are measurable. Timed properties given as specification B\"uchi or
  Muller timed automata are measurable.
\end{restatable}

In the sequel, if $P$ is a measurable property over \AP, we write
$\Prob_{\A}(s \models P)$ for $\Prob_{\A}\{\varrho \in \Runs(\A,s)
\mid \varrho \models P\}$. 

\begin{defi}
  Let $s$ be a state of $\A$.  Assume $P$ is an $(\A,s)$-measurable
  property over \AP.  We say that $\A$ \emph{almost-surely satisfies}
  $P$, from $s$, and we then write $\A,s \robust_{\Prob} P$, whenever
  $\Prob_{\A}(s \models P) = 1$. The \emph{almost-sure model-checking
    problem} asks, given $\A$, $s$ and $P$, whether $\A,s
  \robust_{\Prob} P$.
\end{defi}

The following corollary is an immediate consequence of Lemma~\ref{lemma:proba}. 

\begin{cor}\label{prop:proba-A-RA}
  Let $\A$ be a stochastic timed automaton, $s$ a state of $\A$, and
  $\varphi$ a measurable property over \AP.  Then,
  \[
  \A,s \robust_{\Prob} \varphi\ \Leftrightarrow\ \RA,\iota(s)
  \robust_{\Prob} \varphi\,.
  \]
\end{cor}

\begin{exa}\label{ex:ltl}
  Consider $\A_{\textsf{running}}$ again from
  Figure~\ref{fig:running}, and reproduced below, with initial state
  $s_0=(\ell_0,0)$ and assuming uniform distributions over delays and
  uniform distribution over discrete moves in all states. Then,
  $\A_{\textsf{running}},s_0 \robust_{\Prob} \F(p_1 \wedge \G (p_1
  \Rightarrow \F p_2))$. Indeed, in state $(\ell_0,\nu)$ with $0 \le
  \nu \le1$, the probability of firing $e_2$ (after some delay) is
  always $1/2$ (guards of $e_1$ and $e_2$ are the same, there is thus
  a uniform distribution over the two edges), the location $\ell_1$ is
  eventually reached with probability~$1$. In $\ell_1$, the transition
  $e_3$ will unlikely happen, because its guard $x=1$ is too much
  ``small'' compared to the guard $x \geq 3$ of the transition
  $e_4$. The same phenomenon arises in location $\ell_2$ between the
  transitions $e_5$ and $e_6$. In conclusion, the runs of the timed
  automaton $\A_{\textsf{running}}$ (from $s_0$) almost surely follow
  sequences of transitions of the form
  ${e_1}^*e_2(e_4e_5)^\omega$. Hence, with probability $1$, the
  formula $\F(p_1 \wedge \G (p_1 \Rightarrow \F p_2))$ is satisfied.
  Note that the latter formula is not satisfied in
  $\A_{\textsf{running}}$ from $s_0$ (under the classical \LTL
  semantics), since some runs violate it: `\emph{staying in $\ell_0$
    forever}', `\emph{reaching $\ell_3$}', etc... All these
  counter-examples are unlikely and vanish thanks to our probabilistic
  semantics. \corn
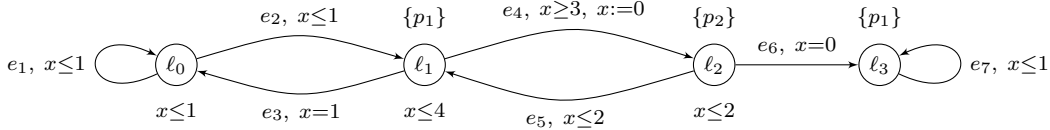
  \begin{figure}[h!]
    \begin{center}
      \runningexample
    \end{center}
  \end{figure}
\end{exa}

Our aim is to \emph{decide} the almost-sure model-checking problem.
It is clear that given a measurable property $P$, the value
$\Prob_\A(s \models P)$ depends on the measures $\mu$ and
$w$. However, we show later that whether $\A,s \robust_{\Prob} P$ is
independent of the precise values of $\mu$ and $w$. To prove this, we
design a finite abstraction, independent of $\mu$ and $w$, which is
correct for deciding the almost-sure model-checking problem in a
number of classes of stochastic timed automata.

\section{A topological semantics}\label{sec:toposem}
In~\cite{VV06} almost-sure model-checking of concurrent reactive
systems is characterised by a topological notion: largeness is
qualitative and captures the notion of ``many runs''. Inspired by that
work, we propose a topological semantics for timed automata, based on
the notion of large sets, which will help us characterise almost-sure
sets. In our context also, the topological semantics is purely
qualitative but nevertheless gives information on ``how big'' a set of
paths satisfying a given property is.

In this section, relying on a notion of thickness for symbolic paths,
we first define a natural topology over infinite runs of a given timed
automaton. This topology induces a \emph{large} semantics: an
$\omega$-regular property is satisfied if ``most of the runs'' satisfy
it. As pointed out already in~\cite{VV06}, \emph{largeness}, and its
complement \emph{meagerness}, are better appropriate than density to
express a notion such as ``most of the runs''. Indeed, ``small'' sets
can be dense, and the complement of a dense set can also be dense
(\textit{e.g.} $\IQ$ in $\IR$), whereas it is not the case for large
sets. Let us start this section by recalling the notion of large sets
and their characterisation using Banach-Mazur games.

\subsection{Largeness and the Banach-Mazur game}
\label{subsec:BMS} 
We refer to~\cite{munkres00} for basic notions of topology
(topological space, interior, closure, \textit{etc}). However we
recall here the more specific notion of \emph{largeness} and also
provide its elegant characterisation in terms of \emph{Banach-Mazur
  games}~\cite{oxtoby57}.

\subsubsection{Some topological notions.} 
Let $(A,{\mathcal T})$ be a topological space. If $B \subseteq A$, we
denote by $\mathring{B}$ (resp. $\overline{B}$) the \emph{interior}
(resp. \emph{closure}) of $B$.  Let us recall that a set ${\mathcal T}'
\subseteq {\mathcal T}$ is called a \emph{basis} for the topology $\mathcal T$
if every open set (\textit{i.e.}, elements of $\mathcal T$) can be
obtained as the union of elements of ${\mathcal T}'$. In this case, the
elements of ${\mathcal T}'$ are called \emph{basic opens}.  A set $B
\subseteq A$ is \emph{nowhere dense} if the interior of the closure of
$B$ is empty, \textit{i.e.}, $\mathring{\overline{B}} = \emptyset$.  A
set is \emph{meagre} if it is a countable union of nowhere dense
sets. Finally, a set is \emph{large} if its complement is meagre.

\begin{exa}
  Let $\IR$ be the set of real numbers equipped with its natural
  topology (that is, basic open sets are the open intervals).  The set
  of integers $\IZ$ is nowhere dense in $\IR$.  The set of rational
  numbers $\IQ$ is dense (in $\IR$) however $\IQ$ is meagre since is a
  countable union of singletons (which are clearly nowhere dense
  sets); this implies that $\IR \setminus \IQ$ is large. \corn
\end{exa}

\subsubsection{Banach-Mazur game.}
Although the notion of largeness is quite abstract, it admits a very
nice characterisation in terms of a two-player game, known as
\emph{Banach-Mazur game}. 

\begin{defi}[Banach-Mazur game]\label{BM:defi}
  Let $(A,{\mathcal T})$ be a topological space and $\mathcal B$ be a family
  of subsets of $A$ satisfying the two following properties:
  \begin{itemize}
  \item for all $B \in {\mathcal B}$, $\mathring{B} \ne \emptyset$, and
  \item for all $O$ a non-empty open set of $A$, there exists $B \in
    {\mathcal B}$ such that $B \subseteq O$.
  \end{itemize}
  Fix $C$ a subset of $A$. Two players alternate their moves:
  Player~$1$ starts and chooses an element $B_1$ of $\mathcal B$;
  Player~$2$ then responds by choosing an element $B_2$ of $\mathcal B$
  such that $B_1 \supseteq B_2$; Then Player~1 chooses $B_3$ in $\mathcal
  B$ such that $B_2 \supseteq B_3$, and so on. This way, they define a
  non increasing sequence of sets $B_i$:
  $$A \supseteq B_1 \supseteq B_2 \supseteq B_3 \cdots$$
  where the $B_{2i+1}$'s (resp. $B_{2i}$'s) are chosen by Player~$1$
  (resp. Player~$2$), for $i \in \IN$. Player~$1$ wins the game if the
  intersection of all $B_i$'s intersects $C$, \textit{i.e.},
  $$\bigcap_{i=1}^\infty B_i \cap C \neq \emptyset\,.$$
  Otherwise, Player~$2$ wins the game.
\end{defi}
Notice that typical examples of family $\mathcal B$ are provided by
topology bases.

Banach-Mazur games are not always determined, even for simple
topological spaces (see \cite[Remark~1]{oxtoby57}). Still a natural
question is to know when the players have winning strategies. The
following result gives a partial answer:

\begin{thm}[Banach-Mazur~\cite{oxtoby57}]\label{BM:thm}
  Player~$2$ has a winning strategy in the Banach-Mazur game with
  target set $C$ if and only if $C$ is meagre.
\end{thm}

To illustrate Theorem~\ref{BM:thm} we give the following simple
example.
\begin{exa}\label{ex:BMGS}
  Let $(\IR,{\mathcal T})$ be the set of real numbers equipped with the
  natural topology, $\mathcal B$ be the family of open intervals with
  rational bounds, and $C$ be the open set $(0,1)$. Intuitively,
  $(0,1)$ is not meagre, and to prove it using Theorem~\ref{BM:thm},
  it is sufficient to show that Player~$2$ does not have a winning
  strategy in the Banach-Mazur game on $\mathcal B$ with target set $C$.

  Assume Player~$1$'s first move is to pick the set $C$ for
  $B_1$. From then on, the game takes place within $C$, and the only
  way for Player~$2$ to win is to build a sequence of $B_i$'s
  converging to the empty set. Let us now explain how Player~$1$ can
  prevent this from happening. Let $B_{2i}=(a_i,b_i)$ be the $i$th set
  chosen by Player~$2$. Player~$1$'s response to $(a_i,b_i)$ is the
  interval $B_{2i+1} = (a_i+\epsilon_i,b_i-\epsilon_i)$, where
  $\epsilon_i = \frac{b_i-a_i}{3}$. Notice that $B_{2i+1}$ is a
  regular move for Player~$1$ since $B_{2i} \supseteq B_{2i+1}$ and
  $B_{2i+1} \in {\mathcal B}$. Notice also that the closed set $F_i =
  [a_i+\frac{\epsilon_i}{2},b_i-\frac{\epsilon_i}{2}]$ satisfies the
  following inclusions: $B_{2i} \supseteq F_i \supseteq
  B_{2i+1}$. Given any strategy of Player~$2$, using the above
  strategy Player~$1$ ensures:
  $$\bigcap_{i=1}^\infty {B_i} \cap {C} \ = \ \bigcap_{i=1}^\infty {F_i}.$$
  This intersection is guaranteed to be non-empty by Heine-Borel
  theorem, since the sequence is included in $[0,1]$ which is a
  compact set. As a consequence, Player~$2$ does not have a winning
  strategy, and thus $C$ is non meagre. \corn
\end{exa}

\begin{rem}\label{rem:Baire}
  Surprisingly, in general, there is no relation between meagre and
  open sets. Indeed non meagre open sets obviously exist, but one can
  identify topological spaces where open sets are meagre. Let us
  consider the set of rational numbers $\IQ$ with its natural topology
  (whose basic open sets are the open intervals). In this topology all
  sets, and in particular all open sets, are countable and thus
  meagre. 

  A topological space in which all non-empty open sets are not meagre
  is called a \emph{Baire space}.\footnote{In modern definitions, a
    topological space is a Baire space if each countable union of
    closed sets with an empty interior has an empty interior. However,
    the two definitions coincide, see~\cite[p.295]{munkres00}.} As we
  have just noticed, not all topological spaces are Baire
  spaces. Indeed, under their natural topologies, $\IQ$ is not a Baire
  space whereas $\IR$ is (a proof follows the same ideas than in
  Example~\ref{ex:BMGS}). This remark suggests that largeness and
  meagerness are even more relevant in Baire spaces.
\end{rem}

\subsection{Thick and thin symbolic paths}\label{subsec:dim}

We fix a timed automaton $\A=(L,X,E,{\mathcal I},{\mathcal L})$ for the rest
of the section. In order to attach a topology to sets of infinite
runs in $\A$, we first define thick symbolic paths. In $\IR^n$, open
sets are among those sets of maximal dimension. Symbolic paths do not
exactly lie in $\IR^n$, but we will see that each symbolic constrained
path of length $n$ can be embedded in some ambient space $\IR^m$, with
$m \leq n$. \emph{Thick} symbolic paths will then naturally arise as
symbolic paths of maximal dimension in their ambient space. Before
going to the definition, let us explain through an example the
intuition behind this notion.

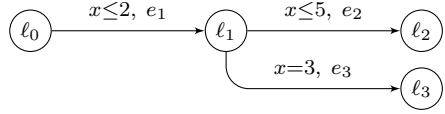
\begin{figure} 
 \begin{center}
    \begin{tikzpicture}[xscale=1.3,yscale=1.1]
      \everymath{\scriptstyle}
      \path[use as bounding box] (0,-.8) -- (4,.4);
      \draw (0,0) node [draw,circle,inner sep=2pt] (l0) {$\ell_0$};
      \draw (2,0) node [draw,circle,inner sep=2pt] (l1) {$\ell_1$};
      \draw (4,0) node [draw,circle,inner sep=2pt] (l2) {$\ell_2$};
      \draw (4,-.7) node [draw,circle,inner sep=2pt] (l4) {$\ell_3$};
      \draw [arrows=-latex'] (l0) -- (l1) node [above,midway] {$x \le 2,\ {e_1}$};
      \draw [arrows=-latex'] (l1) -- (l2) node [above,midway] {$x \le 5,\ {e_2}$};
      \draw [arrows=-latex', rounded corners=3mm] (l1) -- +(0,-.7) -- (l4)
      node [above,midway] {$x = 3,\ e_3$};
    \end{tikzpicture}
  \end{center}
\caption{Thick and thin symbolic paths on an example.}
\label{fig:dim-ex}
\end{figure}

\begin{exa}\label{ex:dim}
  Consider the single-clock timed automaton depicted in
  Figure~\ref{fig:dim-ex}, $s_0 = (\ell_0,0)$ and $\pi$ be the
  (unconstrained) symbolic path $\path{s_0,e_1 e_2}$. One can
  naturally associate a polyhedron of $(\IR_+)^2$ with $\pi$:
  \begin{align*} 
    \Pol(\pi) & = \{(\tau_1,\tau_2) \in (\IR_+)^2 \mid \rho = s_0
    \xrightarrow{\tau_1,e_1} s_1 \xrightarrow{\tau_2,e_2} s_2\}\\
    & = \{(\tau_1,\tau_2) \in (\IR_+)^2 \mid (0 \le \tau_1 \le 2)
    \wedge (0 \le \tau_1 + \tau_2 \le 5)\}
  \end{align*}
  $\Pol(\pi)$ has dimension 2 in $\mathbb{R}^2$. Since it is of
  maximal dimension, we say that the symbolic path $\pi$ is
  \emph{thick}. Consider now the symbolic path $\pi'=\path{s_0,e_1
    e_3}$.  The polyhedron $\Pol(\pi')$ associated with $\pi'$ has
  dimension $1$, and is somehow embedded in a two-dimensional space
  (due to the existence of the edge $e_2$). In that case, we say that
  it is \emph{thin}. \corn
\end{exa}

The above example is simplistic and could give the wrong impression
that symbolic paths with singular transitions
(\emph{i.e.}. transitions that do not increase the dimension of the
polyhedron) are necessarily thin; or equivalently that in order to be
thick, a symbolic path of length $n$ should have an associated
polyhedron of dimension $n$. This is not always the case, and singular
transitions can play an important role. Consider a slight modification
of the automaton of Figure~\ref{fig:dim-ex} where edge $e_1$ is
guarded by $x=2$. In this modified automaton, $\path{s_0,e_1 e_2}$ is
still thick, $\path{s_0,e_1 e_3}$ is thick too and the dimension of
the ambient space of any symbolic path of length $2$ is $1$.

To formally define thick and thin paths, we introduce the notion of
associated polyhedron, and some notations. Given $\pi_{\mathcal{C}} =
\path[\mathcal{C}]{s,e_1 \ldots e_n}$ a constrained path of a timed
automaton $\A$, its \emph{associated polyhedron} is defined as
follows:
\[
\Pol(\pi_{\mathcal{C}}) = \{(\tau_i)_{1 \leq i\leq n} \in (\IR_+)^n \mid s
\xrightarrow{\tau_1,e_1} s_1 \cdots \xrightarrow{\tau_n,e_n} s_n \in
\path[\mathcal{C}]{s,e_1 \ldots e_n}\}\,.
\]
Moreover, for each $0 < i \leq n$, we write $\mathcal{C}_{i}$ for the
constraint induced by the projection of $\Pol(\pi_{\mathcal{C}})$ over
the variables corresponding to the $i$ first coordinates (with the
convention that $\mathcal{C}_0$ is true).

\begin{defi}
  The constrained path $\pi_{\mathcal{C}} = \path[\mathcal{C}]{s,e_1
    \ldots e_n}$ is \emph{thin} whenever there exists some index $1
  \leq i \leq n$ such that
  \[
  \dim\Big(\Pol\big(\path[\mathcal{C}_i]{s,e_1 \ldots e_{i-1} e_i}\big)\Big) <
  \dim\Big(\bigcup_e
  \Pol\big(\path[\mathcal{C}_{i-1}]{s,e_1 \ldots e_{i-1} e}\big)\Big)\,.
  \]
  Otherwise $\pi_{{\mathcal C}}$ is \emph{thick}.
\end{defi}

Clearly enough all extensions of thin symbolic paths are thin as well.
Let us examine an example illustrating some subtlety of this notion.

\begin{exa}\label{ex:dimsubtil}
  Consider the timed automaton depicted in Figure~\ref{fig:ex-thick}
  where $e_i$ denotes the transition ending in $\ell_i$, for
  $i=1,2,3$. 
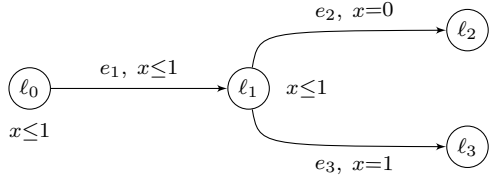
\begin{figure}
    \begin{center}
      \begin{tikzpicture}[scale=.97]
        \everymath{\scriptstyle}
        \path (-3.5,0) node[draw,circle,inner sep=2pt] (q0) 
        {$\ell_0$};
        \path (-3.5,-0.6) node[] (q0bis) {$x \le 1$};
        
        \path (-0.5,0) node[draw,circle,inner sep=2pt] (q1) 
        {$\ell_1$};
        \path (0.3,0) node[] (q1bis) {$x \le 1$};
        
        \path (2.5,.8) node[draw,circle,inner sep=2pt] (q2) 
        {$\ell_2$};
        
        \path (2.5,-.8) node[draw,circle,inner sep=2pt] (q4) 
        {$\ell_3$};
        
        \draw[arrows=-latex'] (q0) --  (q1)
        node[pos=.5, above,sloped] {$e_1,~x \le 1$};
        
        \draw[arrows=-latex'] (q1) .. controls +(80:0.8cm) ..  (q2)
        node[pos=.8, above,sloped] {$e_2,~x = 0$};
        
        \draw[arrows=-latex'] (q1) .. controls +(100:-0.8cm).. (q4)
        node[pos=.8, below,sloped] {$e_3,~x = 1$};        
      \end{tikzpicture}
    \end{center}

\caption{A subtle example illustrating thickness and thinness.}
\label{fig:ex-thick}
\end{figure}
Consider the (unconstrained) symbolic paths $\path{s_0,e_1e_2}$ and
$\path{s_0,e_1e_3}$, where $s_0=(\ell_0,0)$, and let us argue that
$\path{s_0,e_1e_2}$ is thin, whereas $\path{s_0,e_1e_3}$ is
thick. Intuitively, this difference comes from the fact that
``reaching $\ell_2$'' is only possible when transition $e_1$ has been
taken exactly when the value of the clock $x$ is $0$, although
``reaching $\ell_3$'' is always possible after transition $e_1$ has
been taken.  Formally: $\Pol(\path{s_0,e_1e_2}) = \{(0,0)\}$ whereas
$\Pol(\path{s_0,e_1e_3}) = \{(\tau_1,\tau_2) \in [0,1]^2 \mid
\tau_1+\tau_2=1\}$. Hence,
\[
\dim\Big(\Pol\big(\path{s_0,e_1e_2}\big)\Big) <
\dim\Big(\Pol\big(\path{s_0,e_1e_3}\big)\Big) = \dim\Big(\bigcup_e
\Pol\big(\path{s_0,e_1 e}\big)\Big)\,
\]
  proving the desired result. \corn
\end{exa}

The latter example shows that thickness cannot be tested locally:
edges $e_2$ and $e_3$ are both guarded by equality constraints, but do
not behave the same with respect to thickness. This phenomenon cannot
happen in timed region automata, in which, as we shall establish later
(see Proposition~\ref{prop:thickness-local},
page~\pageref{prop:thickness-local}), thickness coincides with local
thickness.

The notion of thickness naturally extends to infinite symbolic paths.
\begin{defi}\label{def:diminf}
  An infinite symbolic path $\path[\mathcal{C}]{s,e_1e_2 \ldots}$ is
  \emph{thick} if for all $n \geq 1$, $\path[\mathcal{C}]{s,e_1 \ldots
    e_n}$ is thick. Otherwise, it is \emph{thin}.
\end{defi}

We illustrate these notions on our running example.

\begin{exa}
  On $\A_{\textsf{running}}$ of Figure~\ref{fig:running}, also
  reproduced below, with $s_0=(\ell_0,0)$, let us explain why
  $\pi(s_0,{e_1}^\omega)$ is thick and
  $\pi(s_0,e_2e_3 e_1^\omega)$ is thin.
  \begin{figure}[h!]
    \begin{center}
      \runningexample
    \end{center}
  \end{figure}
  First observe that all finite prefixes $\path{s_0,e_1^n}$ of
  $\path{s_0,{e_1}^\omega}$ are thick. Indeed,
  \begin{multline*}
    \Pol(\path{s_0,{e_1}^n}) = \\
    \left\{(\tau_1,\ldots,\tau_n) \in (\IR_+)^n \mid (0 \le \tau_1 \le
      1) \wedge \cdots \wedge(0 \le \tau_1+\cdots+\tau_n \le 1)
    \right\}\,.
  \end{multline*}
  Thus, clearly enough $\dim(\Pol(\path{s_0,{e_1}^n})) = n$, which is
  maximal for an $n$-dimension polyhedron. This proves that
  $\path{s_0,{e_1}^n}$ is thick, for all $n \in \IN$, and thus
  $\path{s_0,{e_1}^\omega}$ is thick too.

  Consider now the infinite path $\pi(s_0,e_2 e_3 e_1^\omega)$ and
  show that it is thin by exhibiting a thin finite prefix. Observe
  that:
  \begin{align*}
    \Pol(\path{s_0,e_2 e_3}) = \left\{(\tau_1,\tau_2) \mid (0 \le
      \tau_1 \le 1)\wedge(0 \le \tau_1+\tau_2 =
      1) \right\},\\
    \Pol(\path{s_0,e_2 e_4}) = \left\{(\tau_1,\tau_2) \mid (0 \le
      \tau_1 \le 1)\wedge(3 \le \tau_1+\tau_2) \right\},
  \end{align*}
  thus $\dim(\Pol(\path{s_0,e_2 e_3}))=1 < \dim(\Pol(\path{s_0,e_2
    e_4}))=2$ which implies that $\path{s_0,e_2 e_3}$ is thin.  Hence
  we conclude that $\path{s_0,e_2 e_3 e_1^\omega}$ is thin. \corn
  \end{exa}

\subsection{A topology on infinite runs}\label{subsec:topo}
The goal of this subsection is to define a topology on the set of
infinite runs of a given timed automaton. Keeping our analogy with
$\IR^n$, where the open sets are among the sets of maximal dimension,
we use the notion of thickness introduced in the latter subsection in
order to define a topology on runs. More precisely, the basic open
sets will be cylinders over thick constrained symbolic paths whose
associated polyhedra are open in their ambient spaces.

\begin{defi}\label{def:topo}
  Let $s$ be a state of $\A$. Let $\mathcal{T}_{\A}^s$ be the topology
  over the set of runs of $\A$ starting in $s$ defined with the
  following basic open sets:\footnote{We recall that open sets of
    $\mathcal{T}_{\A}^s$ are then built from those basic open sets
    using union.}  either the set $\Runs(\A,s)$, or the empty set
  $\emptyset$, or the cylinders $\Cyl(\pi_{\mathcal{C}})$ where
  $\pi_{\mathcal{C}} = \path[\mathcal{C}]{s,e_1e_2 \ldots e_n}$ is a
  finite constrained symbolic path of $\A$ such that: (i)
  $\pi_{\mathcal{C}}$ is thick, (ii) $\mathcal{C}$ is
  Borel-measurable, and (iii) $\Pol(\pi_{\mathcal{C}})$ is open in
  $\Pol(\pi)$ for the classical topology on $\IR^n$, where $\pi =
  \path{s,e_1e_2 \ldots e_n}$.
\end{defi}

Before illustrating our definition, let us draw the reader's attention
on the two following points.  First notice that
Definition~\ref{def:topo} only makes sense if the intersection of two
basic open sets is still a basic open set; This is proven in the
Appendix as Lemma~\ref{lemma:intersection}
(page~\pageref{lemma:intersection}). Second, regarding our initial
objective of expressing a notion of ``most of the runs'' using
\emph{largeness}, we need, for consistency, the space to be Baire (see
Remark~\ref{rem:Baire}); This is stated below and proven in
Appendix~\ref{app:Baire}, page~\pageref{app:Baire}.

\begin{restatable}{prop}{propBaire}
  \label{prop:Baire}
  For every state $s$ of $\A$, the topological space
  $(\Runs(\A,s),{\mathcal T}_{\A}^s)$ is a Baire space.
\end{restatable}

Now that the validity of our topology is clear, we illustrate it on
our running example.

\begin{exa}\label{ex:largetaBMS}
  On the running automaton $\A_{\textsf{running}}$ of
  Figure~\ref{fig:running} with initial state $s_0=(\ell_0,0)$, the
  set $C \egdef \bigcup_{i \in \IN} \path{s_0,{e_1}^i e_2 (e_4
    e_5)^\omega}$ is large.  To prove it (or equivalently to prove
  that the complement of $C$ is meagre) we use a Banach-Mazur game,
  and show that Player~2 has a strategy to avoid the complement of
  $C$, that is to reach $C$. The game is played with the basic open
  sets of $\Runs(\A_{\textsf{running}},s_0)$. A winning strategy for
  Player~2 goes as follows:
  \begin{itemize}
  \item Assume Player~1 has chosen a cylinder
    $\Cyl(\path{s_0,{e_1}^{n_1}})$, for some $n_1 \in \IN_0$ (if
    Player~1 leaves $\ell_0$ at her first move, we skip the first move
    of Player~2)
  \item Player~2 chooses $\Cyl(\path{s_0,{e_1}^{n_1} e_2})$,
  \item Notice that Player~1 is not allowed to extend the symbolic
    path $\path{s_0,{e_1}^{n_1} e_2}$ with sequences of transitions
    including $e_3$ or $e_6$, since both symbolic paths
    $\path{s_0,{e_1}^{n_1} e_2 e_3}$ and $\path{s_0,{e_1}^{n_1} e_2
      e_4 e_6}$ are thin. Therefore Player~1 can only play moves of the
    form $\Cyl(\path{s_0,{e_1}^{n_1} e_2 (e_4 e_5)^{n_2}})$ or
    $\Cyl(\path{s_0,{e_1}^{n_1} e_2 (e_4 e_5)^{n_2}e_4})$.
  \item Player~2 then responds $\Cyl(\path{s_0,{e_1}^{n_1} e_2 (e_4
      e_5)^{n_3}})$, with $n_3 > n_2$.
  \end{itemize}
  One can easily be convinced that repeating infinitely often the last
  two moves, the play forms a run of $C$, proving that Player~2 wins
  the game and thus that $C$ is large.

  Notice that both players could also play with constrained
  paths. This would not be interesting for Player~1, since it may
  cause the intersection to be empty (in which case Player~2 wins as
  well). \corn
\end{exa}

We now give a simple characterisation of the basic open sets
considered in Definition~\ref{def:topo}, whose proof is given in
Appendix~\ref{app:basicopensets}, page~\pageref{app:basicopensets}.

\begin{restatable}{lem}{basicopensets}
  \label{lemma:basicopensets}
  In the topological space $(\Runs(\A,s),\mathcal{T}_{\A}^s)$, a
  finite symbolic path $\pi$ defines a basic open set if and only if
  there exist an open constraint $\mathcal{C}$ of $\mathbb{R}^n$ and
  thick edges $e_1,\dots,e_n$ such that $\pi =
  \path[\mathcal{C}]{s,e_1 \dots e_n}$.
\end{restatable}

\subsection{Large satisfaction}
\label{subsec:semtopo}
We are now in a position where we can define a notion of large
satisfaction.

\begin{defi}
  Let $s$ be a state of $\A$, and $P$ be a timed property over \AP.
  We say that $\A$ \emph{largely satisfies} $P$ from $s$, and write
  $\A,s \robust_{\mathcal T} P$, if the set $\{\varrho \in \Runs(\A,s)
  \mid \varrho \models P\}$ is topologically large (in
  $(\Runs(\A,s),\mathcal{T}_{\A}^{s})$).
\end{defi}

Let us illustrate this notion of large satisfaction on our running
example.

\begin{exa}\label{ex:ltlthird}
  On the running example $\A_{\textsf{running}}$ depicted below, with
  $s_0=(\ell_0,0)$,
  \[
  \A_{\textsf{running}},s_0 \robust_{\mathcal T} \F(p_1 \wedge \G (p_1 \Rightarrow \F p_2))\,.
  \]
  
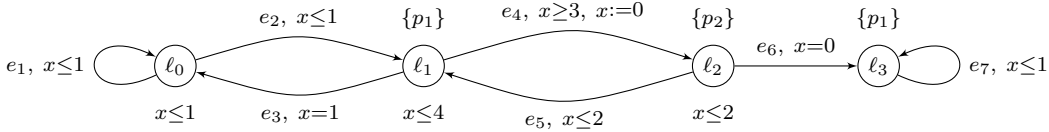
  \begin{figure}[h!]
    \begin{center}
      \runningexample
    \end{center}
  \end{figure}

  Indeed, in Example~\ref{ex:largetaBMS}, we proved that the set $C
  \egdef \{\path{s_0,{e_1}^i e_2 (e_4 e_5)^\omega} \mid i \in \IN\}$
  is large.  Moreover each run of $C$ satisfies $\varphi \egdef \F(p_1
  \wedge \G (p_1 \Rightarrow \F p_2))$, and thus
  $\A_{\textsf{running}},s_0 \robust_{\mathcal T} \varphi$, since
  largeness is closed under subsumption.  Note that the previous
  formula is not satisfied with the classical \LTL
  semantics. `\emph{Staying in $\ell_0$ forever}', `\emph{reaching
    $\ell_3$}', etc are examples of behaviours in
  $\A_{\textsf{running}}$ that violate the \LTL formula $\varphi$. \corn
\end{exa}

\subsection{From timed automata to timed region automata}

As in the context of probabilities, we can relate the topologies in
$\A$ and in $\RA$. Although the topological spaces given by $\A$ and
$\RA$ are not homeomorphic, the topologies in $\A$ and in $\RA$ are
somehow equivalent, as stated by the next lemma. This will allow us to
lift result from $\RA$ to $\A$.

\begin{restatable}{lem}{homeo}
  \label{lemma:homeo} Let $\iota: \Runs(\A,s) \rightarrow
  \Runs(\RA,\iota(s))$ be the projection of runs in $\A$ onto the
  region automaton $\RA$.  Then $\iota$ is continuous, and for every
  non-empty open set $O \in \mathcal{T}_{\A}^s$, $\widering{\iota(O)}
  \neq \emptyset$.
\end{restatable}
The proof of this lemma is given in Appendix~\ref{app:homeo},
page~\pageref{app:homeo}.

\begin{rem}
  Note that $\iota: \Runs(\A,s) \to \Runs(\RA,\iota(s))$ is not an
  homeomorphism since $\iota^{-1} : \Runs(\RA,s) \rightarrow
  \Runs(\A,\iota^{-1}(s))$ is not continuous. Indeed, let us consider
  the automaton $\A$ of Figure~\ref{fig:cex-homeo}, page
  \pageref{fig:cex-homeo}, with $s_0 =(\ell_0,0)$. The set of runs $O
  = \Cyl(\path{s_0,e_1})$ is open in $\mathcal{T}_{\A}^{s_0}$ since
  $\path{s_0,e_1}$ is a thick symbolic path. However,
  $\iota(\Cyl(\path{s_0,e_1})) = \Cyl(\path{s_0,f_1}) \cup
  \Cyl(\path{s_0,f_2})$ is not open in
  $\mathcal{T}_{\RA}^{\iota(s_0)}$ as $\path{s_0,f_1}$ is thin and
  hence $\Cyl(\path{s_0,f_1})$ is not a basic open. Thus $\iota(O)$ is
  not open and $\iota^{-1}$ is not continuous.
\end{rem}

Lemma~\ref{lemma:homeo} allows one to simulate a Banach-Mazur game
from $\A$ to $\RA$ and \emph{vice-versa}. Therefore, the large
satisfaction relations in $\A$ and $\RA$ coincide, see
Appendix~\ref{app:toporegions}, page~\pageref{app:toporegions}.

\begin{restatable}{prop}{toporegions}
  \label{prop-topo-A-RA}
  Let $s$ be a state of $\A$, and $P$ be a timed property over
  \AP. Then,
  \[
  \A,s \robust_{\mathcal T} P\ \Leftrightarrow\ \RA,\iota(s) \robust_{\mathcal
    T} P\,.
  \]
\end{restatable}

\section{Construction of the thick graph}
\label{sec:thickgraph}

In this section, we construct the so-called thick graph.  The idea is
to remove \emph{locally thin} edges from the region automaton, and it
will be used to characterise (globally) thin paths.

We fix a timed automaton $\A = (L,X,E,{\mathcal I},{\mathcal L})$, and we
start by defining a local notion of thinness.
\begin{defi}
  Let $e$ be an edge of $\RA$, and $q$ its source. We say $e$ is
  \emph{thin} whenever $\dim(I(s,e))<\dim(I(s))$ for some (or
  equivalently, for every) $s \in q$. Otherwise the edge is said
  \emph{thick}.
\end{defi}

An equivalent definition is that an edge $e$ with source $q$ is thin
whenever for every $s \in q$, the length-$1$ constrained path
$\path{s,e}$ is thin. That is why this notion of thinness is local.
Next we will write $\dim(I(q,e))$ (resp. $\dim(I(q))$) instead of
$\dim(I(s,e))$ (resp. $\dim(I(s))$) for every $s \in q$ since this is
independent on $s \in q$.

\begin{defi}
  The \emph{thick graph} of $\A$, denoted $\thickgraph(\A)$, is
  obtained from $\RA$ by deleting all the thin edges.
\end{defi}

In particular, $\thickgraph(\A)$ has only thick edges. We first state
a lemma, which explains how the dimension of symbolic paths grows in
the (timed) region automaton. Its proof can be found in
Appendix~\ref{app:thickgraph}, page~\pageref{app:thickgraph}

\begin{restatable}{lem}{thinlocal}
  \label{lem:thin-local}
  Let $\path{s,e_1 \ldots e_n}$ be a symbolic path of $\RA$. Assuming
  $\Pol(\path{s,e_1 \ldots e_n}) \ne \emptyset$ and letting $q$ be the
  target region of $\path{s,e_1 \ldots e_{n-1}}$,
  \[
  \dim(\Pol(\path{s,e_1 \ldots e_n})) = \dim(\Pol(\path{s,e_1 \ldots
    e_{n-1}}))+\dim(I(q,e_n)).
  \]
\end{restatable}
That is, we are able to compute the global dimension of a symbolic
path, given the dimension of each of its edges. Notice that this is
only true in the region automaton. As an example, in the timed
automaton of Example~\ref{ex:dimsubtil}, $\Pol(\path{s,e_1e_2}) =
\{(0,0)\}$, whereas $\Pol(\path{s,e_1}) = [0,1]$.

The following proposition states the correctness of the thick graph,
in the sense that a symbolic path is (globally) thin if and only if it
traverses a (locally) thin edge. Its proof is given in
Appendix~\ref{app:thicknesslocal}, page~\pageref{app:thicknesslocal}.
 
\begin{restatable}{prop}{thicknesslocal}
  \label{prop:thickness-local}
  Let $\pi = \path{s,e_1 \ldots e_n}$ be a symbolic path in
  $\RA$. Then, $\pi$ is thin in $\RA$ iff there exists $1 \le i \le n$
  such that $e_i$ is thin.
\end{restatable}

The thick graph will be used in the next section for characterising
large and almost-sure satisfaction. It will later be used for
algorithmic issues.

\begin{exa}
  We illustrate the construction of the thick graph on the running
  example. Figure~\ref{fig:bluegraph-running} represents the classical
  region graph where thick (resp. thin) edges are depicted bold
  (resp. dashed). From that region graph, by removing thin edges and
  keeping only the states reachable from the initial one, we obtain
  the thick graph, represented on
  Figure~\ref{fig:bluegraph-running-end}. \corn
\begin{figure}[htb]
\begin{center}
      \begin{tikzpicture}[xscale=1.5,yscale=1.5]
        \everymath{\scriptstyle}

        \path (-4.5,0) node[draw,inner sep=2pt] (q00)
        {$\ell_0,0$};
        \path (-4.5,-1.5) node[draw,inner sep=2pt] (q01)
        {$\ell_0,(0,1)$};
        \path (-4.5,-3) node[draw,inner sep=2pt] (q02)
        {$\ell_0,1$};
        
        \path (-3,0) node[draw,inner sep=2pt] (q10)
        {$\ell_1,0$};
        \path (-3,-1.5) node[draw,inner sep=2pt] (q11)
        {$\ell_1,(0,1)$};
        \path (-3,-3) node[draw,inner sep=2pt] (q12)
        {$\ell_1,1$};
        
        \path (-1.5,0) node[draw,inner sep=2pt] (q2)
        {$\ell_2,0$};
      
        \path (0,0) node[draw,inner sep=2pt] (q30)
        {$\ell_3,0$};
        \path (0,-1.5) node[draw,inner sep=2pt] (q31)
        {$\ell_3,(0,1)$};
        \path (0,-3) node[draw,inner sep=2pt] (q32)
        {$\ell_3,1$};

        \draw [-latex',dashed] (q00) .. controls +(180:25pt) and +(120:25pt) .. (q00) node [above,pos=0.5] {$e_1$};
        \draw [-latex',thick] (q01) .. controls +(180:25pt) and +(120:25pt)
        .. (q01) node [above,pos=0.5] {$e_1$};
        \draw [-latex',dashed] (q02) .. controls +(180:25pt) and +(120:25pt) .. (q02) node [above,pos=0.5] {$e_1$};
        \draw[-latex',thick] (q00) -- (q01) node [left,pos=.3] {$e_1$};
        \draw[-latex',dashed] (q01) -- (q02) node [left,pos=.3] {$e_1$};
        \draw[-latex',dashed] (q00) .. controls +(190:50pt) and +(180:50pt).. (q02) node[pos=.5,left] {$e_1$};

        \draw[-latex',dashed] (q00) -- (q10) node [midway,above] {$e_2$};
        \draw[-latex',thick] (q00) -- (q11) node [midway,above] {$e_2$};
        \draw[-latex',dashed] (q00) -- (q12) node [midway,above] {$e_2$};
        \draw[-latex',thick] (q01) -- (q11) node [midway,above] {$e_2$};        
        \draw[-latex',dashed] (q01) -- (q12) node [midway,above] {$e_2$};
        \draw[-latex',thick] (q02) -- (q12) node [midway,above] {$e_2$};

        \draw[arrows=-latex',dashed] (q12) .. controls +(200:25pt) and +(-20:25pt) .. (q02) node[pos=.5, below,sloped] {$e_3$};
      
        \draw[arrows=-latex',thick] (q10) .. controls +(20:25pt) and +(-200:25pt) .. (q2) node[pos=.5, above,sloped] {$e_4$};
        \draw[arrows=-latex',thick] (q11) -- (q2) node[pos=.5, above,sloped] {$e_4$};
        \draw[arrows=-latex',thick] (q12) .. controls +(60:40pt) and +(-110:40pt).. (q2) node[pos=.3, above,sloped] {$e_4$};
      
        \draw[arrows=-latex',dashed] (q2) .. controls +(200:25pt) and +(-20:25pt) .. (q10) node[pos=.5, below,sloped] {$e_5$};
        \draw[arrows=-latex',thick] (q2) .. controls +(-120:25pt) .. (q11) node[pos=.8, below,sloped] {$e_5$};
        \draw[arrows=-latex',dashed] (q2) .. controls +(-100:80pt) .. (q12) node[pos=.5, below,sloped] {$e_5$};
      
        \draw[arrows=-latex',dashed] (q2) -- (q30) node[pos=.5,above,sloped] {$e_6$};
      
        \draw [-latex',dashed] (q30) .. controls +(0:25pt) and +(60:25pt)
        .. (q30) node [above,pos=0.5] {$e_7$};
        \draw [-latex',thick] (q31) .. controls +(0:25pt) and +(60:25pt)
        .. (q31) node [above,pos=0.5] {$e_7$};
        \draw [-latex',thick] (q32) .. controls +(0:25pt) and +(60:25pt)
        .. (q32) node [above,pos=0.5] {$e_7$};
        \draw [-latex',thick] (q30) -- (q31) node [midway,right] {$e_7$};
        \draw [-latex',dashed] (q31) -- (q32) node [midway,right] {$e_7$};
        \draw[-latex',dashed] (q30) .. controls +(-10:50pt) and +(0:50pt).. (q32) node[pos=.5,left] {$e_1$};

      \end{tikzpicture}
    \end{center}    
\caption{Construction of the thick graph on the running example.}
\label{fig:bluegraph-running}
\end{figure}
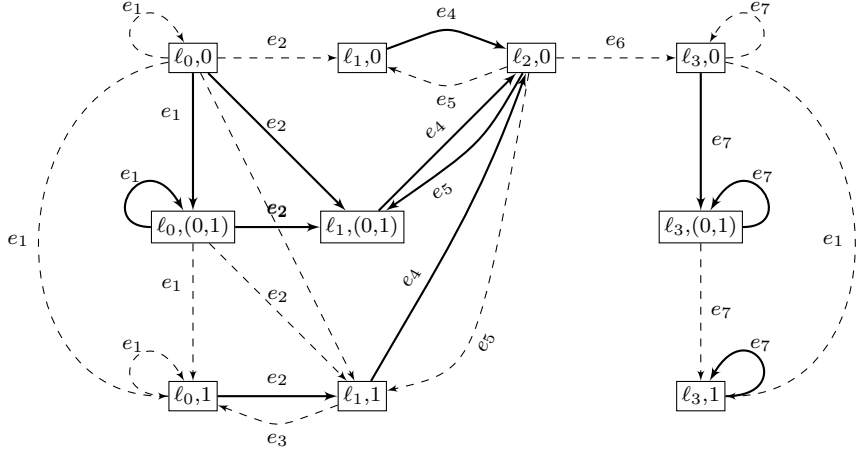
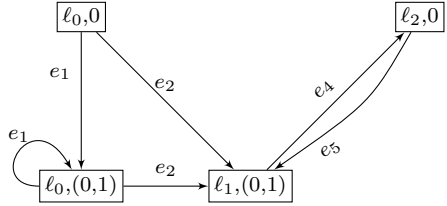
\begin{figure}[htb]
\begin{center}
   \begin{tikzpicture}[xscale=1.5,yscale=1.5]
   \everymath{\scriptstyle}
   \path (-4.5,0) node[draw,inner sep=2pt] (q00)
        {$\ell_0,0$};
        \path (-4.5,-1.5) node[draw,inner sep=2pt] (q01)
        {$\ell_0,(0,1)$};

        \path (-3,-1.5) node[draw,inner sep=2pt] (q11)
        {$\ell_1,(0,1)$};

        \path (-1.5,0) node[draw,inner sep=2pt] (q2)
        {$\ell_2,0$};

        \draw [-latex'] (q01) .. controls +(180:25pt) and +(120:25pt)
        .. (q01) node [above,pos=0.5] {$e_1$};

        \draw[-latex'] (q00) -- (q01) node [left,pos=.3] {$e_1$};

        \draw[-latex'] (q00) -- (q11) node [midway,above] {$e_2$};

        \draw[-latex'] (q01) -- (q11) node [midway,above] {$e_2$};        

        \draw[arrows=-latex'] (q11) -- (q2) node[pos=.5, above,sloped] {$e_4$};

        \draw[arrows=-latex'] (q2) .. controls +(-120:25pt) .. (q11) node[pos=.8, below,sloped] {$e_5$};
\end{tikzpicture}
\end{center}
\caption{$\thickgraph(\A_{\textsf{running}})$, the thick graph for the running example.}
\label{fig:bluegraph-running-end}
\end{figure}
\end{exa}

\begin{rem}
  In the following, we will denote by $\MC(\A)$ the finite Markov
  chain obtained by taking $\thickgraph(\A)$ as the support of the
  Markov chain, and assuming uniform distributions over edges. In the
  next sections, we write $\MC(\A),s \robust P$ whenever the finite
  Markov chain $\MC(\A)$ almost-surely satisfies property $P$ from
  state $s$.
\end{rem}

\section{When do the large and the almost-sure satisfaction coincide?}
\label{sec:match}

We know that large satisfaction and almost-sure satisfaction coincide
for finite automata for several classes of properties~\cite{VV06}. We
want here to discuss situations where almost-sure and large
satisfactions also match in our context. This will help giving
algorithmic solutions to the almost-sure model-checking problem using
the thick graph.

We fix for this section a stochastic timed automaton $\langle \A,\mu,w
\rangle$.

\subsection{Safety properties}

We first compare the two semantics in the restricted case of safety
properties.  Let us first state this simple result that, in a region
automaton, a finite symbolic path has probability $0$ iff it is
thin. This is the first easy link we can make between the
probabilities and the topology. Note that this correspondence does not
depend on the choice of the probability distributions.

\begin{restatable}{prop}{probthick}
  \label{prop:prob-thick} 
  Consider a finite symbolic path $\pi = \path{s,e_1 \ldots e_n}$ in
  $\RA$.  Then, $\Prob_{\RA}(\Cyl(\pi)) >0$ iff $\pi$ is
  thick. Equivalently, $\Prob_{\RA}(\Cyl(\pi)) =0$ iff $\pi$ is thin.
\end{restatable}

This proposition relies on condition $(\star)$ on the probability
distributions $\mu$ (\textit{cf} page~\pageref{subsec:proba}): if an
edge is thick, it is either because it has dimension $1$, or because
it has dimension $0$, but all other outgoing edges also have dimension
$0$. In the first case, the measure $\mu$ must be equivalent to the
Lebesgue measure, and in the second case, it will give a positive mass
to the edge. The full details of the proof are given in
Appendix~\ref{app:safety}, page~\pageref{app:safety}.

Using this result, we show that the large and the almost-sure
satisfaction always coincide when we restrict to safety properties.

\begin{restatable}{thm}{theosafety}
  \label{th:safety}
  Let $s$ be a state of $\A$, and $P$ be a(n untimed) safety
  property over \AP. Then the four following properties are
  equivalent:
  \begin{enumerate}[label=\({\alph*}]
  \item $\A ,s \robust_{\Prob} P$;
  \item $\A,s \robust_{\mathcal T} P$;
  \item every infinite thick symbolic path $\pi$ from $\iota(s)$ in
    $\RA$ satisfies $P$;
  \item every infinite path $\pi$ from $\iota(s)$ in $\thickgraph(\A)$
    satisfies $P$;
  \item $\MC(\A),\iota(s) \robust P$.
  \end{enumerate}
\end{restatable}

This result relies on the fact that safety properties are violated by
finite prefixes. Proposition~\ref{prop:prob-thick} then tells us that
such finite prefixes yield positive probability whenever they are
thick. We can then play a Banach-Mazur game in the topological space
of the automaton, where only thick paths can be used by the players as
moves. This allows to show that property $\neg P$ is meagre iff $P$ is
not violated by thick prefixes, that is when $\neg P$ has probability
$0$.  The details of the proof are given in Appendix~\ref{app:safety},
on page~\pageref{app:safety2}.

As said, the proof of Theorem~\ref{th:safety} heavily relies on the
fact that witnesses of violation (resp. validation) for safety
(resp. reachability) properties are finite prefixes. Not surprisingly,
Theorem~\ref{th:safety} does not hold for general \LTL or
$\omega$-regular properties, for which the violation cannot always be
witnessed by finite prefixes. As an example, consider the timed
automaton of Figure~\ref{fig:cex-safety}, and the property $\F p$. The
probability of $\F p$ is indeed $1$ in this automaton, although the
infinite symbolic path $\path{(\ell_0,0),e_0^\omega}$ violates $\F p$
and is thick.

\begin{figure}[h!]
\begin{center}
\begin{tikzpicture}[yscale=.9,xscale=1.1]
  \everymath{\scriptstyle}
\path (0,0) node[draw,circle,inner sep=2pt] (l0) {$\ell_0$};
\path (3,0) node[draw,circle,inner sep=2pt] (l1) {$\ell_1$};

\path (3,.6) node (l1') {$\{p\}$};

   \draw [-latex'] (l0) .. controls +(210:40pt) and +(150:40pt) .. (l0)
        node [midway,left] {$e_0,~x \leq 1$};

    \draw[arrows=-latex'] (l0) -- (l1) node[pos=.5,above,sloped] {$e_1,~x\leq 1$};
\end{tikzpicture}
\end{center}
\caption{A counterexample for Theorem~\ref{th:safety} beyond safety.}
\label{fig:cex-safety}
\end{figure}
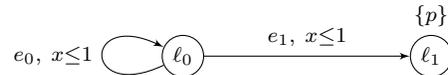

This kind of behaviours motivates the restriction to fair paths, which
is rather natural since probabilities and strong fairness are closely
related~\cite{Pnu83,PZ93,BK98}.

\subsection{Restriction to fairness: the case of prefix-independent properties}

Motivated by the counterexample of Figure~\ref{fig:cex-safety}, we
define a natural notion of fairness for infinite symbolic paths in
timed automata.

\begin{defi}
  An infinite region path $\pi = q_0 \xrightarrow{e_1} q_1
  \xrightarrow{e_2} q_2 \ldots$ of $\RA$ is \emph{fair} iff for every
  thick edge $e$, if $e$ is enabled in infinitely many $q_i$ ($i \in
  \IN$), then $e_i =e$ for infinitely many $i \in \IN$.
\end{defi}

Fairness extends to runs and symbolic paths in an obvious way as
detailed below. Region paths and symbolic paths in $\RA$ are closely
related: to any non-empty symbolic path $\path{s,e_1e_2\ldots}$, we
associate a unique region path $q_0 \xrightarrow{e_1} q_1
\xrightarrow{e_2} q_2 \ldots$ with $s \in q_0$ and $q_i$ is the target
region of edge $e_i$. We then say that a symbolic path
$\path{s,e_1e_2\ldots}$ in $\RA$ is fair whenever its corresponding
region path is fair. Finally we say that an infinite run $\varrho$ in
$\RA$ is fair whenever $\pi_\varrho$ (its underlying symbolic path) is
fair. Obviously, the set of fair infinite runs from $s$ is measurable
(that is, in $\Omega_{\RA}^s$), as fairness is an $\omega$-regular
property over infinite paths. We now turn the definition from $\RA$ to
$\A$: an infinite run $\varrho$ in $\A$ is fair whenever
$\iota(\varrho)$ is fair.

We write $\mathsf{fair}$ for this property, that is if an infinite run
$\varrho$ (in $\A$ or in $\RA$) is fair, then we write $\varrho
\models \mathsf{fair}$. We then write $\Prob_{\A}(s \models
\mathsf{fair})$ (resp. $\Prob_{\RA}(\iota(s) \models \mathsf{fair})$)
for the probability of fair runs in $\A$ from $s$ (resp. in $\RA$ from
$\iota(s)$). We say that $\A$ (resp. $\RA$) is \emph{almost-surely
  fair} from $s$ (resp. $\iota(s)$) whenever $\Prob_{\A}(s \models
\mathsf{fair})=1$ (resp. $\Prob_{\RA}(\iota(s) \models
\mathsf{fair})=1$).

Let us state the following straightforward lemma, which gives a useful
characterisation of thick and fair symbolic paths in the region
automaton.

\begin{lem}
  Let $s$ be a state of $\A$, and $P$ be a prefix-independent property
  over \AP. Then, the two following properties are equivalent:
  \begin{enumerate}[label=(\roman*)]
  \item all BSCCs\footnote{BSCC stands for `bottom strongly connected
      component', see \emph{e.g.}~\cite{book-Cormen}.} reachable from $[s]$ in $\thickgraph(\A)$ satisfy
    $P$;\footnote{We say that a prefix-independent property $P$ is
      satisfied by a BSCC $C$ whenever every run visiting infinitely
      often all states of $C$ satisfies $P$. As $P$ is
      prefix-independent, this reduces to the existence of a run
      satisfying $P$, which visits infinitely often all states of
      $C$.}
  \item every infinite thick \textbf{and fair} symbolic path $\pi$
    from $\iota(s)$ in $\RA$ satisfies $P$;
  \item $\MC(\A),\iota(s) \robust P$.
  \end{enumerate}
  We write   $\S (\A,s,P)$ for this property.
\end{lem}

\subsubsection{Fairness and topology.}
Even though fairness is introduced because of the probabilities, we
first realise that adding fairness to paths allows to characterise the
large satisfaction. More precisely, we prove the following result,
which is rather similar to Theorem~\ref{th:safety}, when restricted to
topology.

\begin{thm}
  \label{theo:large}
   Let $s$ be a state of $\A$, and $P$ be a prefix-independent property
  over \AP. Then:
  \[
  \A,s \robust_{{\mathcal T}} P\ \Leftrightarrow\ \S(\A,s,P)\,.
  \]
\end{thm}

\begin{proof}
  First assume that $\A,s \robust_{\mathcal{T}} P$. This equivalently
  means that $\RA,\iota(s) \robust_{\mathcal{T}} P$
  (Lemma~\ref{prop-topo-A-RA}), that is,
  $\sem{\neg P}_{\RA,\iota(s)}$ is meagre. We now apply the
  characterisation of meagre sets \textit{via} Banach-Mazur games
  (Theorem~\ref{BM:thm}), where the players play with basic open sets
  of $(\Runs(\RA,\iota(s)),\mathcal{T}_{\RA}^{\iota(s)})$.  Note that
  for every basic open set $\path[\mathcal{C}]{s,e_1 \dots e_n}$ in
  the above topology, all edges $e_1,\ldots,e_n$ are thick, the open
  set can therefore ``be read'' in $\thickgraph(\A)$.

  In this game, Player 2 has a strategy to ensure that $\bigcap_i B_i
  \cap \sem{\neg P}_{\RA,\iota(s)} = \emptyset$, where the $B_i$'s are
  the moves in the game. Let us denote $\Sigma_2$ this winning
  strategy.  Fix any BSCC $C$ of $\thickgraph(\A)$ reachable from
  $[s]$, and let $\{\ell_1,\ldots,\ell_p\}$ be an enumeration of the
  locations in $C$. In order to prove that $C$ satisfies $P$, we will
  build a symbolic path $\pi$, played according to $\Sigma_2$ (and
  thus satisfying $P$), witnessing that $C$ satisfies $P$.  Let
  Player~1 play as follows (against $\Sigma_2$): at her first move,
  Player~1 chooses $B_1$ leading to $\ell_1$; then, no matter which
  $B_2$ Player~2 chooses, Player~1 chooses $B_3$ (longer than $B_2$)
  leading to $\ell_2$; \textit{etc}. Furthermore applying a technique
  similar to the proof of Proposition~\ref{prop:Baire}
  (compactification), Player 1 can ensure that $\bigcap_i B_i \ne
  \emptyset$.  We then get that $\emptyset \ne \bigcap_i B_i \subseteq
  \sem{P}_{\RA,\iota(s)}$, since $\Sigma_2$ is winning for
  Player~2. Let $\pi$ be the infinite symbolic path underlying
  $\bigcap_i B_i$.  As $P$ is a property over \AP, $\pi \models
  P$. Furthermore, $\pi$ visits all locations of $C$ (since Player 1
  has ensured visiting all the locations of $C$ infinitely often),
  which means that the BSCC $C$ satisfies $P$. We conclude that all
  BSCCs satisfy property $P$.

  Conversely, assume that every BSCC of $\thickgraph(\A)$ satisfies
  $P$.  It is then easy to provide a winning strategy for Player 2 in
  the same Banach-Mazur game as described above. Once a BSCC $C$ has
  been reached (after Player~1's first move or Player~2's first move),
  Player~2 will ensure to visit all the locations of $C$ (as Player~1
  did in the above proof). The resulting infinite path will satisfy
  $P$ (by hypothesis), which implies the winning condition for
  Player~2. 
\end{proof}

\subsubsection{What about fairness and probabilities?}
The natural question is: can we fully extend Theorem~\ref{th:safety}
and therefore prove that almost-sure satisfaction is equivalent to
$(\S)$? We concentrate on the following equivalence:
\begin{equation}
  \label{eq:fair}
  \tag{\ddag}
  \A ,s \robust_{\Prob} P \ \Leftrightarrow \ 
  \left\{\begin{array}{@{}l}
      \text{every infinite thick \textbf{and fair} symbolic path}\ \pi \\
      \text{from}\ \iota(s)\ \text{in}\ \RA\ \text{satisfies}\ P
    \end{array}\right.
\end{equation}
We show now that this equivalence is unfortunately not true in
general.  The counter-example to that equivalence is much more
surprising than that of Figure~\ref{fig:cex-safety}. The reasons why
it fails is rather subtle and is due to complex \emph{time converging}
behaviours. As pointed out in~\cite{CHR02}, timed automata admit
various time converging behaviours, and, in our context, some of these
behaviours can lead to ``big'' sets of \emph{unfair}
executions. Inspired by an example presented in~\cite{CHR02}, we
design a two-clock timed automaton $\A_{\mathsf{unfair}}$ (see
Figure~\ref{fig:pacman}) for which the equivalence~\eqref{eq:fair}
does not hold.  In this automaton, every (infinite) fair and thick
symbolic path satisfies $\G\F p_1 \wedge \G\F p_2$, and in particular
$\F p_2$. Thus, letting $\varphi = \F p_2$, the right-hand side of
equivalence~\eqref{eq:fair} is true. However, when
$\A_{\mathsf{unfair}}$ is equipped with uniform distributions,
starting in $s_0= (\ell_0,(0,0))$, one can show that the probability
of the symbolic path $\path{s_0,(e_3 e_4 e_5)^\omega}$ is positive and
therefore $\A_{\mathsf{unfair}},s_0 \not \robust_{\Prob} \F p_2$. We
notice that this implies $\Prob_{\A_{\mathsf{unfair}}}(s_0 \models
\mathsf{fair})<1$.

All this is proven formally in Appendix~\ref{app:pacman},
page~\pageref{app:pacman}.

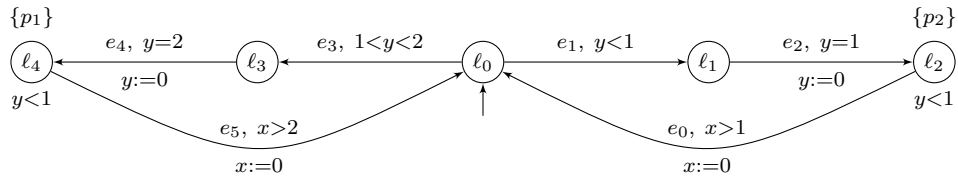
\begin{figure}[h]
  \begin{center}
    \begin{tikzpicture}[yscale=.85]
      \everymath{\scriptstyle}
      \path[use as bounding box] (-6,-1.4) -- (6,.7);
      \path (0,0) node[draw,circle,inner sep=2pt] (q0) {$\ell_0$};
          \path (0,-1) node[] (q0b) {};
      
      \path (3,0) node[draw,circle,inner sep=2pt] (q1) {$\ell_1$};
      
      \path (6,0) node[draw,circle,inner sep=2pt] (q3) {$\ell_2$};
      \path (6,-.6) node[] (q3b) {$y < 1$};
      \path (6,.7) node[] (q3a) {$\{p_2\}$};

      \path (-3,0) node[draw,circle,inner sep=2pt] (q4) {$\ell_3$};
          
      \path (-6,0) node[draw,circle,inner sep=2pt] (q6) {$\ell_4$};
      \path (-6,0.7) node[] (q6a) {$\{p_1\}$};
      \path (-6,-.6) node[] (q6b) {$y < 1$};

      \draw[arrows=-latex'] (q0b) -- (q0);
      
      \draw[arrows=-latex'] (q0) -- (q1) node[pos=.5, above,sloped]
      {$e_1,~y<1$};
      
      \draw[arrows=-latex'] (q1) -- (q3) node[pos=.5, above,sloped]
      {$e_2,~y=1$} node[pos=.5, below, sloped] {$y:=0$};
      
      \draw[arrows=-latex'] (q3) .. controls +(210:3.5cm).. (q0)
      node[pos=.5, above,sloped] {$e_0,~x>1$} node[pos=.5, below, sloped]
      {$x:=0$};
      
      \draw[arrows=-latex'] (q0) -- (q4) node[pos=.5, above,sloped]
      {$e_3,~1<y<2$};
      
      \draw[arrows=-latex'] (q4) -- (q6) node[pos=.5, above,sloped]
      {$e_4,~y=2$} node[pos=.5, below, sloped] {$y:=0$};
      
      \draw[arrows=-latex'] (q6) .. controls +(330:3.5cm).. (q0)
      node[pos=.5, above,sloped] {$e_5,~x>2$} node[pos=.5, below,
      sloped] {$x:=0$};
    \end{tikzpicture}
  \end{center}
  \caption{A two-clock automaton, $\A_{\mathsf{unfair}}$ with non
    negligible set of unfair runs.} \label{fig:pacman}
\end{figure}

\subsubsection{When fairness is almost-sure.}
We will show that a sufficient condition to have \eqref{eq:fair} is to
have $\Prob_{\A}(s_0 \models \mathsf{fair}) =1$.  We can now state the
following crucial theorem:
\begin{thm}
  \label{th:fairhelps}
  Let $s$ be a state of $\A$, and $P$ be a prefix-independent
  (untimed) property over \AP.  Assuming $\Prob_{\A}(s \models
  \mathsf{fair}) =1$, the following holds:
  \[
  \A ,s \robust_{\Prob} P\ \Leftrightarrow\ \S(\A,s,P)\,.
  \]
\end{thm}

\begin{proof}
  First define the property $\mathsf{thick}$ as follows: letting
  $\varrho$ be a run in $\RA$, $\varrho \models \mathsf{thick}$ iff
  $\pi_\varrho$ is thick. Applying Theorem~\ref{th:safety}, we get
  that $\Prob_{\RA}(s \models \mathsf{thick})=1$.

  We now prove that $\Prob_{\RA}(s \models P) >0$ iff there exists an
  infinite symbolic path $\pi$ from $s$, which is thick and fair, and
  such that $\pi \models P$. This will imply the expected result.  We
  prove the two implications separately.
  \begin{itemize}
  \item \textit{Proof of the left-to-right implication.} Let us assume
    that $\Prob_{\RA}(s \models P) >0$. We have seen that
    $\Prob_{\RA}(s \models \mathsf{thick}) =1$. Therefore, thanks to
    the fact that $\Prob_{\RA}(s \models \mathsf{fair}) =1$,
    $\Prob_{\RA}(s \models P) = \Prob_{\RA}(s \models P \wedge
    \mathsf{fair} \wedge \mathsf{thick})$. Hence,
    \[
    \Prob_{\RA}(s \models P \wedge \mathsf{fair} \wedge
    \mathsf{thick}) >0.
    \]
    In particular, there exists a fair thick infinite path from $s$
    which satisfies $P$.
  \item \textit{Proof of the right-to-left implication.} Let $\pi =
    \path{s_0,e_1 e_2 \ldots}$ be a fair thick symbolic path in $\RA$
    satisfying $P$.  We consider the thick graph $\thickgraph(\A)$ of
    $\A$.  Since $\pi$ is thick, $\pi$ is also a path in
    $\thickgraph(\A)$.  Let us consider the strongly connected
    components of $\thickgraph(\A)$. As $\pi$ is a fair path, it
    eventually reaches a BSCC in $\thickgraph(\A)$ and from then on
    takes each edge of the BSCC infinitely often. Otherwise, this
    would mean that $\pi$ ignores a thick edge forever, which would
    contradict the fairness assumption. Let ${B}_\pi$ be the BSCC that
    $\pi$ eventually reaches and $\pi_\pref$ the shortest prefix of
    $\pi$ leading from $s$ to ${B}_\pi$ (note that it is
    thick). Consider the following set of paths in~$\RA$:
    \[
    S \egdef \{\pi' \in \Cyl(\pi_\pref) \mid \pi' \textrm{ is thick
      and fair} \}\,.
    \]
    Since $\Prob_{\RA}(s \models \mathsf{thick}) =1$ and
    $\Prob_{\RA}(s \models \mathsf{fair}) =1$, we deduce that
    $\Prob_{\RA}(S) = \Prob_{\RA}(\Cyl(\pi_\pref))$. Moreover since
    $\pi_\pref$ is thick, we obtain $\Prob_{\RA}(S) >0$. It now
    suffices to observe that all paths in $S$ satisfy $P$.  Indeed,
    the satisfiability of prefix-independent (untimed) properties over
    \AP only depends on the set of states that are visited infinitely
    often, and all paths in $S$ visit infinitely often exactly the
    states in ${B}_\pi$, and $\pi \models P$.  
  \end{itemize}
\end{proof}

\subsubsection{Conclusion.} We can conclude this subsection by stating
the main result concerning prefix-independent properties. It is a
direct consequence of Theorems~\ref{theo:large}
and~\ref{th:fairhelps}.

\begin{cor}
  \label{coro:1}
  Let $s$ be a state of $\A$, and $P$ be a prefix-independent
  (untimed) property over \AP.
  Assuming $\Prob_{\A}(s \models \mathsf{fair}) =1$, the following
  equivalences hold:
  \[
  \A ,s \robust_{\Prob} P\ \Leftrightarrow\ \A ,s \robust_{\mathcal T} P\
  \Leftrightarrow\ \S(\A,s,P)\,.
  \]
\end{cor}
Note that characterisation~$(\S)$ will help with algorithmic issues.
Complexity issues will be discussed in Section~\ref{sec:algo}.

\subsection{Extension to richer properties}
\label{sec:richer}

In this section we extend the previous study to properties which are
richer than prefix-independent properties, in particular to \LTL
properties and properties expressed as specification timed automata.

Let $\langle \A,\mu,w \rangle$ be a stochastic timed automaton with
$\A = (L,X,E,{\mathcal I},{\mathcal L})$, and let $\B =
(\mathsf{L},\mathsf{i}_0,\mathsf{X},\AP,\mathsf{E},\mathcal{F})$ be a
specification B\"uchi or Muller timed automaton.  We build the product
stochastic timed automaton $\langle \A \ltimes \B, \overline\mu,
\overline{w}\rangle$ as follows. The timed automaton $\A \ltimes \B =
(\overline{L},X\cup \mathsf{X},\overline{E},\overline{\mathcal{I}})$
(with no labelling function) is such that:
\label{def:produit-TA-spec}
\begin{itemize}
\item $\overline{L} = L \times \mathsf{L}$;
\item $\overline{E}$ is composed of the following edges: if $e \egdef
  (\ell \xrightarrow{g,Y} \ell') \in E$ then for all $\mathsf{e}
  \egdef (\mathsf{l} \xrightarrow{\mathsf{g},\lambda(\ell),\mathsf{Y}}
  \mathsf{l}') \in\mathsf{E}$, there is an edge $((\ell,\mathsf{l})
  \xrightarrow{g \wedge \mathsf{g},Y \cup \mathsf{Y}}
  (\ell',\mathsf{l}'))$ in $\overline E$, which we write
  $e_{\mathsf{e}}$;
\item $\overline{\mathcal{I}}((\ell,\mathsf{l})) = \mathcal{I}(\ell)$
  for every $(\ell,\mathsf{l}) \in \overline{L}$.
 \end{itemize}
The measures $\overline\mu$ and the weights $\overline{w}$ are such
that:
\begin{itemize}
\item $\overline\mu_{(\ell,\mathsf{l})} = \mu_{\ell}$ for every
  $(\ell,\mathsf{l}) \in \overline{L}$, and
\item $\overline{w}_{\overline{e}} = w_e$ for every edge $\overline{e}
  \in \overline{E}$ which comes from edge $e$.
\end{itemize}

Given a state $s = (\ell,v)$ of $\A$, we define the initial state in
$\A \ltimes \B$ as $\mathsf{init}_{\A \ltimes \B}(s) \egdef
((\ell,\mathsf{i}_0(\mathcal{L}(s))),v\mathbf{0}_{\mathsf{X}})$; that
is, we start in $\B$ from the location specified by the label of state
$s$ (this is $\mathsf{i}_0(\mathcal{L}(s))$), with all clocks of $\B$
set to $0$.

Note that any run $\varrho$ in $\A$ from state $s$ has a unique image
in $\A \ltimes \B$ from $\mathsf{init}_{\A \ltimes \B}(s)$, denoted
$\varrho^{\B}$ (since $\B$ is complete and deterministic).  Note that
the converse also holds: for any run $\varrho'$ in $\A \ltimes \B$
from some $\mathsf{init}_{\A \ltimes \B}(s)$, there is a unique
preimage $\varrho$ in $\A$ from $s$, such that $\varrho' =
\varrho^\B$. We define the $\omega$-regular property $P_{\A \ltimes
  \B}$ in $\A \ltimes \B$ as the lifting of $\mathcal{F}$ in $\A
\ltimes \B$ (an infinite run in $\A \ltimes \B$ satisfies $P_{\A
  \ltimes \B}$ whenever its projection on $\B$ satisfies the accepting
condition $\mathcal{F}$). As $\mathcal{F}$ is an internal
prefix-independent condition in $\B$, $P_{\A \ltimes \B}$ is an
internal prefix-independent condition in $\A \ltimes \B$.

\begin{rem}
  It should be clear enough that $\A \ltimes \B$ is non-blocking
  assuming $\A$ is. Moreover, for all states $s=(\ell,v)$ of $\A$, for
  all states $\mathsf{s}=(\mathsf{l},\mathsf{v})$ of $\B$, and for
  every edge $e\in E$,
  \[
  I(s,e) = \bigcup_{\mathsf{e} = (\mathsf{l}
    \xrightarrow{\mathsf{g},\lambda(\ell),\mathsf{Y}} \mathsf{l}')}
  I(((\ell,\mathsf{l}),(v\mathsf{v})),e_{\mathsf{e}})
  \]
  This allows to properly define the probability measure $\Prob_{\A
    \ltimes \B}$.
 \end{rem}

We can now state the main theorem for specification timed automata,
which uses this product construction. Its proof is given in
Appendix~\ref{app:specauto}, page~\pageref{app:specauto}.

\begin{restatable}{thm}{specauto}
  \label{theo:specauto}
  Let $s$ be a state of $\A$, and $\B$ be a specification B\"uchi or
  Muller timed automaton.  Assuming $\Prob_{\A \ltimes
    \B}(\mathsf{init}_{\A \ltimes \B}(s) \models \mathsf{fair}) =1$,
  the following holds:
  \[
  \A ,s \robust_{\Prob} \B\ \Leftrightarrow\ \A ,s \robust_{\mathcal T}
  \B\ \Leftrightarrow\ \S(\A \ltimes \B,\mathsf{init}_{\A \ltimes
    \B}(s),P_{\A \ltimes \B})\,.
  \]
\end{restatable}

The above also applies to specification untimed automata.  In
particular, it also applies to specification automata corresponding to
\LTL formulas. It is now easy to be convinced that if $\B$ is a
specification untimed automaton, then $\Prob_{\A}(s \models
\mathsf{fair}) = \Prob_{\A \ltimes \B}(\mathsf{init}_{\A \ltimes
  \B}(s) \models \mathsf{fair})$ since $\B$ does not restrict guards
of edges, and in particular
\[
\Prob_{\A}(s \models \mathsf{fair}) =1\ \Leftrightarrow\ \Prob_{\A
  \ltimes \B}(\mathsf{init}_{\A \ltimes \B}(s) \models \mathsf{fair})
=1\,.
\]
This allows to get the following important corollary, which
characterises the almost-sure model-checking for \LTL.

\begin{cor}
\label{coro:match-fair}
  Let $s$ be a state of $\A$, and $\varphi$ be an \LTL formula over \AP.
Assuming $\Prob_{\A}(s \models \mathsf{fair}) =1$, the following
holds:
  \[
  \A ,s \robust_{\Prob} \varphi\ \Leftrightarrow\ \A ,s \robust_{\mathcal
    T} \varphi\ \Leftrightarrow\ \S(\A \ltimes
  \B_\varphi,\mathsf{init}_{\A \ltimes \B_\varphi}(s),P_{\A \ltimes
    \B_\varphi})\,.
  \]
\end{cor}

\section{Application to several classes of timed automata}
\label{sec:appli}
In the previous section, we showed that, provided fairness is
almost-sure, one could characterise almost-sure satisfaction and large
satisfaction using thick paths. We will describe here two classes of
stochastic timed automata for which this holds.

\subsection{Single-clock timed automata}
In this section, we focus on single-clock timed automata, and we show
that, under some minor additional technical hypotheses, single-clock
timed automata are almost-surely fair.  In particular,
Corollary~\ref{coro:match-fair} will apply, yielding the decidability
of the almost-sure model-checking problem for $\omega$-regular
properties on this class of stochastic timed automata.

Let $\A= (L,X,E,{\mathcal I},{\mathcal L})$ be a single-clock stochastic timed
automaton.  We assume the following conditions on $\A$, denoted
$(\dag)$:
\label{dagger}
\begin{description}
\item[(H3)]
For all $\ell \in L$, for all $[a,b] \subseteq \IR_+$, the
  function $v \mapsto \mu_{(\ell,v)}([a,b])$ is continuous;
\item[(H4)] If $s'=s+t$ for some $t \geq 0$, and $0 \notin I(s+t',e)$ for
  every $0 \leq t' \leq t$, then $\mu_s(I(s,e)) \leq
  \mu_{s'}(I(s',e))$;
\item[(H5)] There is $0<\lambda_0<1$ s.t. for every state $s$ with $I(s)$
  unbounded, $\mu_s([0,1/2]) \leq \lambda_0$.
\end{description}

\begin{rem}
  The three last requirements are technical and needed to properly
  define a probability measure over infinite runs, but they are
  natural and easily satisfiable. For instance, a timed automaton
  equipped with uniform (resp. exponential\footnote{With bounded
    transition rates, see~\cite{DP03}.}) distributions on bounded
  (resp. unbounded) intervals satisfy these conditions. If we assume
  exponential distributions on unbounded intervals, the very last
  requirement corresponds to the bounded transition rate condition
  in~\cite{DP03}, required to have reasonable and realistic
  behaviours.
\end{rem}

\begin{restatable}{thm}{oneclock}
  \label{theorem:fair-oneclock}
  Assuming $\A$ satisfies $(\dag)$, if $s$ is a state of $\A$,
  $\Prob_{\A}(s \models \mathsf{fair}) = 1$.
\end{restatable}

The proof of this theorem is very technical. We will describe the main
ingredients of the proof in the core of the paper, and postpone all
details to Appendix~\ref{app:single}.

Let $\{c_i \mid 0 \le i \le k\}$ be the set of constants appearing in
guards of $\A$, assuming w.l.o.g.  $c_0 = 0$. We know~\cite{LMS04}
that the following intervals are regions for $\A$:
\[
\{c_i\}\ \text{for}\ 0 \le i \le k; \qquad (c_i,c_{i+1})\ \text{for}\
0 \le i < k; \qquad (c_k,+\infty)
\]
We assume $\RA$ is built with these regions. It is polynomial-size
(contrary to standard region automaton which is exponential-size).

The proof of the above theorem then relies on
Lemma~\ref{lemma:technic} below. A \emph{subregion} of a region $q$ is
a pair $(q,J)$ such that $J \subseteq q$ is an interval. If $s \in J$,
we may write $s \in (q,J)$ as well. If $(q,J)$ and $(q',J')$ are
subregions, we write $(q,J) \xrightarrow{e} (q',J')$ to express that
$(q,v) \xrightarrow{\tau,e} (q',v')$ for some $v \in J$, $v' \in J'$
and $\tau \in \IR_+$. In the sequel to ease the reading, we will use
\LTL-like notations, like $\Prob_{\A}( s, \G \F (q,J) \xrightarrow{e}
(q',J') \mid \G \F (q,J))$ to denote the conditional probability of
the set of runs $s_0 \xrightarrow{\tau_1,e_1} s_1
\xrightarrow{\tau_2,e_2} s_2 \cdots$ such that $s_0 = s$ and $\{s_i
\xrightarrow{e_{i+1}} s_{i+1} \mid s_i \in J,\ e_{i+1}=e,\ \text{and}\
s_{i+1} \in J'\}$ is infinite, assuming that the set $\{s_i \mid s_i
\in J\}$ is infinite. We will use other similar notations, that we
expect are sufficiently explicit to be understandable.

\begin{restatable}{lem}{technica}
  \label{lemma:technic}\hfill
  \begin{enumerate}
  \item For every subregion $(q,J)$ of $q$ such that (i) $J$ is
    non-empty and open in $q$ (for the induced topology), and (ii)
    $\overline{J} \subseteq q$ is compact,
  \item for every thick edge $e$ enabled in $q$,
  \item for every subregion $(q',J')$ of $q'$ such that for every $s
    \in (q,J)$, $e(s) \cap J'$ is non-empty and open in $q'$ (for the
    induced topology), where $e(s) = \{s' \mid \exists \tau \in \IR_+\
    \text{s.t.}\ s \xrightarrow{\tau,e} s'\}$,
  \item for every state $s$ of $\RA$ such that $\Prob_{\RA}(s, \G \F
    (q,J)) > 0$,\footnote{This is for the next conditional probability
      to be defined.}
  \end{enumerate}
  \[
  \Prob_{\RA}( s, \G \F (q,J) \xrightarrow{e} (q',J') \mid \G \F
  (q,J)) = 1\,.
  \]\smallskip
\end{restatable}

\noindent The idea of this lemma is to provide a lower-bound on the probability
of firing the transition $(q,J) \xrightarrow{e} (q',J')$ each time we
visit $(q,J)$. By thickness of $e$, we know that the probability at
each visit is positive, but as $\overline{J}$ is compact, we infer a
positive uniform lower-bound $\lambda$. This is the main ingredient to
prove the result.

\begin{rem}
  This lemma holds for all timed automata, not only one-clock timed
  automata.
\end{rem}

We have shown the proof for a single edge, but this lemma can be
extended straightforwardly to finite sequences of edges as follows:
\begin{lem}
  \label{lemma:technic2}
  \begin{enumerate}
  \item For all regions $(q_i)_{0 \leq i \leq p}$,
  \item for all edges $(e_i)_{1 \leq i \leq p}$ such that $e_i$ is
    thick and enabled in $q_{i-1}$,
  \item for all subregions $((q_i,J_i))_{0 \leq i \leq p}$ such that
    for every $0 \le i <p$:
    \begin{enumerate}[label=\({\alph*}]
    \item $J_i$ is non-empty and open in $q_i$ (for the induced
      topology),
    \item $\overline{J_i} \subseteq q_i$ is compact, and
    \item for every $s \in J_i$, $e_i(s) \cap J_{i+1}$ is non-empty
      and open, where $e_i(s) = \{s' \mid \exists \tau \in \IR_+\
      \text{s.t.}\ s \xrightarrow{\tau,e_i} s'\}$,
    \end{enumerate}
  \item for every state $s$ of $\A$ such that $\Prob_{\RA}(s, \G \F
    (q_0,J_0) ) > 0$
  \end{enumerate}
  \noindent it holds that:
  \[
  \Prob_{\RA}( s, \G \F \sigma \mid \G \F (q_0,J_0) ) = 1
  \] 
  where $\sigma = (q_0,J_0) \xrightarrow{e_1} (q_1,J_1) \ldots
  \xrightarrow{e_p} (q_p,J_p)$.
\end{lem}

Now, we can turn back to Theorem~\ref{theorem:fair-oneclock}.

\begin{proof}[Sketch of proof of Theorem~\ref{theorem:fair-oneclock}]
  Let $s$ be a state in $\A$. We want to prove that $\Prob_{\A}(s
  \models \mathsf{fair})=1$. We will equivalently prove
  $\Prob_{\RA}(\iota(s) \models \mathsf{fair})=1$.  To that purpose,
  we decompose the set of infinite runs in $\RA$ from $\iota(s)$ into:
  \begin{enumerate}[label=$(F_{\arabic*})$]
  \item[$(F_1)$] the set of runs with infinitely many resets,
  \item[$(F_2)$] the set of runs with finitely many resets, and which
    are ultimately in the unbounded region $(c_k,+\infty)$,
  \item[$(F_3)$] the set of runs with finitely many resets, and which
    ultimately stay forever in a bounded region, either $\{c_i\}$ with
    $0 \leq i \leq k$, or $(c_i,c_{i+1})$ with $0 \leq i < k$. We
    write $(F_3^{(c_i,c_{i+1})})$ (resp. $(F_3^{c_i})$) for condition
    $F_3$ restricted to $(c_i,c_{i+1})$ (resp.  $\{c_i\}$).
  \end{enumerate}
  
  \noindent We write $\Prob_{\RA}(s,F_j)$ for the probability of the runs
  starting in $s$ and satisfying condition~$F_j$. The three sets of
  runs above are measurable and partition the set of all runs. Hence
  $\sum_{j=1,2,3} \Prob_{\RA}(s,F_j) = 1$, and applying Bayes formula:
  \begin{equation}
    \label{bayes}\tag{$\bullet$}
    \Prob_{\RA}(s \models \mathsf{fair}) = \sum_{j=1,2,3} \Prob_{\RA}(s
    \models \mathsf{fair} \mid F_j) \cdot \Prob_{\RA}(s,F_j)\,.
  \end{equation}
  We now distinguish between the three cases to prove that
  $\Prob_{\RA}(s \models \mathsf{fair} \mid F_j)=1$ (in case
  $\Prob_{\RA}(s,F_j)=0$ we remove the corresponding term
  from~\eqref{bayes}).
  \begin{description}
  \item[Case $F_1$.] In that case, for all states $(q,0)$ which are
    visited infinitely often we apply Lemma~\ref{lemma:technic2} to
    any sequence of thick edges, and get the expected result for
    $F_1$.
  \item[Case $F_2$.] Once the unbounded region is reached, precise
    values of clocks are irrelevant, and the timed automaton roughly
    behaves like a finite Markov chain, which yields the expected
    result for $F_2$.
  \item[Case $F_3$.] We consider runs which end up in a bounded region
    $r$. This case is only possible (with positive probability) if no
    thick edge is enabled infinitely often with a guard above $r$
    (otherwise it would be taken infinitely often~--~due to the
    hypothesis {\bf $(H4)$}). Therefore in this case as well the
    automaton ultimately behaves like a finite Markov chain, which
    allows to conclude with the expected property.\qedhere 
  \end{description}
\end{proof}
Note that the one-clock hypothesis is crucial for cases $F_1$ and $F_3$.

\subsubsection*{Checking almost-sure non-Zenoness in one-clock timed
  automata.}

Note that we cannot obtain the result on Zeno behaviours from
Corollary~\ref{coro:match-fair} since Zenoness is a timed property.
We can neither use the product construction of
Section~\ref{sec:richer} and Theorem~\ref{theo:specauto} since this
will increase the number of clocks to~$2$. Therefore we need a
specific proof, that we present now.

We first show the following crucial lemma.

\begin{restatable}{lem}{lemmazeno}
  \label{lemma:zeno1}
  Assuming $\A$ satisfies $(\dag)$, if $s$ is a state of $\A$, then: 
  \[
  \Prob_{\A}(s \models \Zeno) = \sum_{B\ \text{Zeno BSCC of}\ \thickgraph(\A)}
  \Prob_{\RA}(\iota(s) \models \F B)
  \]
  where a BSCC of $\thickgraph(\A)$ is said Zeno whenever it is
  bounded and the clock is never reset.
\end{restatable}

To prove the lemma, we need to realise that runs ending up in the
unbounded region or runs with infinitely many resets have probability
$0$ to be Zeno. Indeed, in the first case, there will be non
constraint on the clock for taking transition, and in particular, with
probability $1$ a delay of $1$ will elapse infinitely often; in the
second case, infinitely often a guard of the form $0<x<1$ will be
enabled, and therefore with probability $1$, a delay of at least $1/2$
will be done, yielding with probability $1$ a non-zeno run. It remains
those runs ending up in a bounded region, which will correspond to the
right-hand side of the equality in the statement (due to fairness, the
runs end up with probability $1$ in a BSCC). The details are given in
the Appendix on page~\pageref{app:zeno-lemma}.

\begin{restatable}{thm}{zeno}
  \label{thm:zeno}
   Assuming $\A$ satisfies $(\dag)$, if $s$ is a state of $\A$, then
  the three following properties are equivalent:
  \begin{enumerate}[label=\({\alph*}]
  \item $\A ,s \robust_{\Prob} \neg\Zeno$;
  \item $\A,s \robust_{\mathcal T} \neg\Zeno$;
  \item no Zeno BSCC
    is reachable in $\thickgraph(\A)$ from $\iota(s)$.
  \end{enumerate}
\end{restatable}

This theorem is a consequence of Lemma~\ref{lemma:zeno1}, and uses
once more Banach-Mazur games for what concerns the large satisfaction.
The details are also given in the Appendix on page~\pageref{app:zeno2}.

\begin{cor}
  \label{coro:zeno}
  The almost-sure non-Zenoness problem is \NLOGSPACE for single-clock
  stochastic timed automata which satisfy condition $(\dag)$.
\end{cor}

Condition $(c)$ in Theorem~\ref{thm:zeno} can be checked in \NLOGSPACE
(recall that in that case $\thickgraph(\A)$ is polynomial in the size
of $\A$).

\subsection{Reactive and weak-reactive timed automata}
\label{subsec:reactive}

Although fairness is not almost-sure in general for $n$-clock
stochastic timed automata with $n\ge 2$ (see Figure~\ref{fig:pacman}),
it is the case for the subclass of \emph{reactive} (and weak-reactive)
stochastic timed automata. Let us first recall that a (non-stochastic)
timed automaton $\A$ is reactive if $I(s)=\IR_+$ for all states $s$ of
$\A$.

We first focus on the class of reactive stochastic timed automata, and
later extend the results to the class of weak-reactive stochastic
timed automata.

\begin{defi}
  A stochastic timed automaton $\langle \A,\mu,w \rangle$
  is \emph{reactive} whenever the timed automaton $\A=(L,X,E,{\mathcal
    I},{\mathcal L})$ is reactive, and for every $\ell \in L$, there
  exists a probability distribution $\mu_\ell$ over $\IR_+$,
  equivalent to the Lebesgue measure,
  such that for every $v \in \IR_+^X$, $\mu_{(\ell,v)}
  = \mu_\ell$.
\end{defi}
Note that for any constant $C$, for any $\ell \in L$,
$\mu_\ell([0,C])<1$, this is due to the equivalence of $\mu_{\ell}$
with the Lebesgue measure. Note also that if $s =(\ell,v)$ and $s' =
(\ell,v')$ are such that $v \cong_{\A} v'$, then $p_{s}(e) =
p_{s'}(e)$ for every edge~$e$.

\begin{exa}
  Examples of distributions over delays that respect the above
  conditions are exponential distributions, but we can think of many
  other kinds of distributions. Later, all our examples will use
  exponential distributions. Each such distribution is characterised
  by a positive parameter $\lambda_\ell$, and its density is $t
  \mapsto \lambda_\ell \cdot e^{-\lambda_\ell \cdot t}$. \corn
\end{exa}

Note that reactive stochastic timed automata generalise
continuous-time Markov chains (CTMC for short). A CTMC is nothing else
than a single-clock reactive stochastic timed automaton in which (i)
on all transitions, the guard is trivial, and the clock is reset, and
(ii) each location is assigned an exponential distribution over
delays.

Let $\langle \A,\mu,w \rangle$ be a reactive stochastic timed
automaton.  The goal of this section is to prove the following result,
which will allow to apply results of Section~\ref{sec:match}, and in
particular Theorem~\ref{theo:specauto}.

\begin{prop}
  \label{prop:fair-reactive}
  Let $s$ be a state of $\A$. Then $\Prob_{\A}(s \models
  \mathsf{fair}) = 1$.
\end{prop}

We first recall some basic probability results.

\begin{lem}
	\label{petitlemmes}
Let $\Prob$ be a probability measure on some probabilistic space
$\Omega$. Let $A$, $B$ and $C$ be measurable sets such that
$\Prob(B)>0$ and $\Prob(C)>0$, then
\begin{enumerate}
\item If $\Prob(A)=1$, then ${\Prob(A \mid B) = 1}$.
\item If $\Prob(A \mid B) = 1$, $\Prob(B \mid C) = 1$, then
  $\Prob(A \mid C) = 1$.
\item If $\Prob(A \mid B) = 1$, $\Prob(A \mid C) = 1$, then
  $\Prob(A \mid B \cup C)=1$.
\end{enumerate}
\end{lem}

We write $\langle \RA,\mu,w \rangle$ for the stochastic timed region
automaton based on $\A$ (we abusively also write $\mu$ and $w$ since
we will not use that notation for $\A$), and we write $Q$ for the set
of locations of $\RA$, and $T$ for its set of edges.  Thanks to
Corollary~\ref{prop:proba-A-RA}, in order to prove
Proposition~\ref{prop:fair-reactive}, it is sufficient to prove that
$\RA$ is almost-surely fair.  In the following we denote by $M$ the
maximal constant appearing in $\A$ (or $\RA$), and we write $\Prob$
instead of $\Prob_{\RA}$.

To prove almost-sure fairness in $\RA$, we have to show that for every
thick edge $e$, the probability to visit $e$ infinitely often, knowing
we visit $\source(e)$ infinitely often, is equal to $1$. The key point
of this proof lies in the fact that, since the automaton $\RA$ is
reactive, there exists a subset of regions, called \emph{memoryless},
that will be visited infinitely often with probability $1$.  A region
$r$ is said \emph{memoryless} whenever the following holds for every
clock $x \in X$: either $v(x)=0$ for every $v\in r$, or $v(x)>M$ for
every $v \in r$. The interest of memoryless regions is that once you
reach such a region the future (and its probability)
is independent of both the finite prefix and the clock valuations. In
particular, visiting infinitely often memoryless regions prevents
converging phenomena as the one observed in
Figure~\ref{fig:pacman}. For each memoryless region $r$, we
distinguish a canonical valuation $v_r \in r$ defined by $v_r(x)=0$ or
$v_r(x)=M+1$ for every $x \in X$ (note that this valuation is uniquely
defined). If $q=(\ell,r) \in Q$ is such that $r$ is memoryless, we
distinguish the canonical configuration $s_q = (\ell,v_r)$ (or
$s_q=((\ell,r),v_r)$ in $\RA$).

The fact that memoryless regions are visited infinitely often
almost-surely is formalised in Lemma~\ref{lemma:eq1}. Then it remains
to show that knowing we visit infinitely often such a memoryless
region, we visit a reachable thick edge $e$ infinitely often with
probability $1$. To this end, we investigate the set of runs that
visit infinitely often memoryless regions and $e$, and we conclude
thanks to a judicious decomposition of this set and Borel-Cantelli
lemma, this is formalised in Lemma~\ref{lemma:eq2}.  A sketch of proof
is given for Lemma~\ref{lemma:eq1} and the complete proofs of both
lemmas are provided in Appendix~\ref{app:reactive},
page~\pageref{app:reactive}.

In order to formalise these ideas, we need to introduce further
notations. Let $s$ be a state of $\RA$, that we will take as initial.
If $e$ is a thick edge in $T$, and $q\in Q$, we write
$\mathfrak{R}^e(s)$ for the set of runs in $\RA$ that start in $s$ and
take $e$ infinitely often, and $\mathfrak{R}^{q}(s)$ for the set of
runs of $\RA$ that start in $s$ and visit $q$ infinitely often. In
particular, we write $\mathfrak{R}^{\source(e)}(s)$ for the set of
runs that start in $s$ and visit $\source(e)$ infinitely often (hence
along which $e$ is enabled infinitely often).

We fix a thick edge $e$ in $T$, and we let $\mathcal{Q}$ be the set of
pairs $q = (\ell,r)$ where $r$ is memoryless and $\mathcal{Q}'$ the
set of elements $q=(\ell,r)\in \mathcal{Q}$ such
that \[\Prob(\mathfrak{R}^{q}(s))>0 \quad\text{and}\quad
\Prob(\mathfrak{R}_0^{q,e}(s_q))>0\] where $\mathfrak{R}_0^{q,e}(s_q)$
is the set of runs that start from $s_q$ and take $e$ before any other
visit to $q$.

\begin{restatable}{lem}{equaun}
  \label{lemma:eq1}
Assuming the above notations,
\[
\Prob\left(\bigcup_{q \in \mathcal{Q}} \mathfrak{R}^{q}(s) \right) =
1\,.
\]
\end{restatable}

\begin{proof}[Sketch of proof]
  We notice that the set of runs that delay infinitely many times more
  than $M$ time units before taking a transition is a subset of
  $\bigcup_{q \in \mathcal{Q}} \mathfrak{R}^{q}(s)$. Indeed, if a run
  $\rho$ delays more than $M$ time units before taking the $n$-th
  transition then each clock is either reset on the $n$-th transition
  (hence its value is $0$), or its value exceeds $M$. Now, as for
  every $\ell\in L$, we have assumed $\mu_\ell$ is equivalent to the
  Lebesgue measure on $\IR_+$, it holds that $\mu_{\ell}([0,M])<1$.
  We can then prove that the probability of the set of runs that delay
  only finitely many times more than $M$ time units is zero, since $L$
  is finite, which concludes the proof.  
\end{proof}

\begin{restatable}{lem}{equadeux}
  \label{lemma:eq2}
  Assuming the above notations, for all $q \in \mathcal{Q}'$
  \begin{equation*}
    \Prob\left(\mathfrak{R}^e(s) \mid
      \mathfrak{R}^{q}(s)\right) = 1.
  \end{equation*}
\end{restatable}

Now assuming Lemma~\ref{lemma:eq1} and Lemma~\ref{lemma:eq2}, we will
prove Proposition~\ref{prop:fair-reactive}.\bigskip

\noindent\emph{Proof of Proposition~\ref{prop:fair-reactive}.}
  We want to prove that the probability of being fair is $1$, hence we
  want to prove that for every thick edge $e$ with
  $\Prob(\mathfrak{R}^{\source(e)}(s))>0$,
  \begin{equation*}
    \Prob\left(\mathfrak{R}^e(s) \mid
      \mathfrak{R}^{\source(e)}(s)\right) = 1\,.
  \end{equation*}
  Let $e$ be a thick edge with
  ${\Prob(\mathfrak{R}^{\source(e)}(s))>0}$.  By
  Lemma~\ref{lemma:eq1}, we have that
  \begin{equation}
    \label{eq2}
    \Prob\left(\bigcup_{q \in \mathcal{Q}} \mathfrak{R}^{q}(s) \right)
    = 1,
  \end{equation}
  and by Lemma~\ref{lemma:eq2}, we have that
  \begin{equation}
    \label{eq1}
    \Prob\left(\mathfrak{R}^e(s) \mid
        \mathfrak{R}^{q}(s)\right) = 1
  \end{equation}
  for any $q\in \mathcal{Q}'$.
  \medskip
  Applying Lemma~\ref{petitlemmes} (point 3.), we deduce from
  Equation~$(\ref{eq1})$ that
  \begin{equation}
    \Prob\left(\mathfrak{R}^e(s) \mid \bigcup_{q \in \mathcal{Q}'}
      \mathfrak{R}^{q}(s)\right) = 1
    \label{eq3}
  \end{equation}
  and applying Lemma~\ref{petitlemmes} (point 1.), we deduce from
  Equation~$(\ref{eq2})$ that:
  \begin{equation}
    \Prob\left(\bigcup_{q \in \mathcal{Q}} \mathfrak{R}^{q}(s) \mid
      \mathfrak{R}^{\mathsf{source}(e)}(s) \right) = 1.
    \label{eq4}
  \end{equation}
  Moreover, we can easily show that
  \begin{equation}
    \Prob\left(\bigcup_{q \in \mathcal{Q}} \mathfrak{R}^{q}(s)\mid
      \mathfrak{R}^{\mathsf{source}(e)}(s) \right)
    =
    \Prob\left(\bigcup_{q \in \mathcal{Q}'} \mathfrak{R}^{q}(s) \mid
      \mathfrak{R}^{\mathsf{source}(e)}(s) \right)\,.
    \label{eq5}
  \end{equation}
  Indeed, we just have to prove that
  \begin{equation*}
    \Prob\left(\bigcup_{q \in \mathcal{Q}} \mathfrak{R}^{q}(s)\cap
      \mathfrak{R}^{\mathsf{source}(e)}(s) \right)
    = \Prob\left(\bigcup_{q \in \mathcal{Q}'} \mathfrak{R}^{q}(s) \cap
      \mathfrak{R}^{\mathsf{source}(e)}(s) \right)
  \end{equation*}
  and it is thus sufficient to prove that
  \[
  \Prob\left(\bigcup_{q \in \mathcal{Q}\backslash \mathcal{Q}'}
    \mathfrak{R}^{q}(s)\cap \mathfrak{R}^{\mathsf{source}(e)}(s)
  \right)= 0\,.
  \] 
  However, if $q\in \mathcal{Q}\backslash \mathcal{Q}'$, we have
  $\Prob(\mathfrak{R}^{q}(s))=0$ or
  $\Prob(\mathfrak{R}_0^{q,e}(s_q))=0$.  Now, if
  $\Prob(\mathfrak{R}^{q}(s))=0$, we have
  \[
  \Prob\left(\mathfrak{R}^{q}(s)\cap
    \mathfrak{R}^{\mathsf{source}(e)}(s) \right)=0
  \]
  and if $\Prob(\mathfrak{R}_0^{q,e}(s_q))=0$, we also have
  \[\Prob\left(\mathfrak{R}^{q}(s)\cap
    \mathfrak{R}^{\mathsf{source}(e)}(s) \right)=0\,.
  \]
  We therefore deduce from Equations~$(\ref{eq4})$ and $(\ref{eq5})$
  that
  \begin{equation}
    \Prob\left(\bigcup_{q \in \mathcal{Q}'} \mathfrak{R}^{q}(s) \mid
      \mathfrak{R}^{\mathsf{source}(e)}(s) \right) = 1\,.
    \label{eq6}
  \end{equation}
  Applying Lemma~\ref{petitlemmes} (point 2.), we get the expected
  result from Equations~$(\ref{eq3})$ and $(\ref{eq6})$:
  \[
  \Prob\left(\mathfrak{R}^e(s) \mid
    \mathfrak{R}^{\mathsf{source}(e)}(s)\right) = 1\,.\eqno{\qEd}
  \]

\subsubsection{Extension to weak reactive stochastic timed automata.}
We now extend the almost-sure fairness from the subclass of reactive
stochastic timed automata to the larger class of \emph{weak reactive}
stochastic timed automata, defined as follows.
  
\begin{defi}
  A stochastic timed automaton $\langle \A,\mu,w \rangle$
  with $\A = (L,X,E,{\mathcal I},{\mathcal L})$
    is said to be \emph{weak reactive} whenever $L$ is the
  disjoint union of sets $L_u$ and $L_b$ such that
    $\ell \in L_u$ if and only if $I(\ell,v)$ is unbounded for all $v$
  such that $v \models \mathcal{I}(\ell)$, and:
  \begin{itemize}
  \item for every pair $s=(\ell,v)$, $s'=(\ell,v')$ satisfying for
    every $x\in X$, $v(x)=v'(x)$ or $\min(v(x),v'(x))>M$,
    \[\mu_s=\mu_{s'}\,;\]
  \item there exists $0<\lambda_0\le 1$ such that for every $\ell \in
    L_u$, for every $v \models \mathcal{I}(\ell)$, we have that
    \[\mu_{(\ell,v)}\left(\left[M,+\infty\right[\,\right) \ge
    \lambda_0\,;\]
  \item there exist $0<\lambda_1\le 1$ and $N\ge 1$ such that for any
    $\ell\in L_b$, for every $v \models \mathcal{I}(\ell)$, we have
    that
    \[\Prob_{\A}\Big(\bigcup_{(e_1,\dots,e_N)\in
      E_u}\path{(\ell,v),e_1,\dots,e_N}\Big)\ge \lambda_1\] where
    $E_u=\{(e_1,\dots,e_N) \mid \target(e_i)\in L_u\text{\ for some\ }
    1\le i< N\}$.
  \end{itemize}
\end{defi}

\noindent It is obvious that the class of reactive stochastic timed automata is
a subclass of weak reactive stochastic timed automata. In fact, the
main difference between these classes of automata lies on the
existence of some states $s$ such that $I(s)$ is bounded.
  
The proof of almost-sure fairness of reactive stochastic timed
automata is based on two lemmas: Lemma~\ref{lemma:eq1} and
Lemma~\ref{lemma:eq2}. The proof of Lemma~\ref{lemma:eq2} works in the
same way for weak reactive stochastic timed automata, using the fact
that for every pair $s=(\ell,v)$, $s'=(\ell,v')$ satisfying for every
$x\in X$, $v(x)=v'(x)$ or $\min(v(x),v'(x))>M$, we have
$\mu_s=\mu_{s'}$. The proof of Lemma~\ref{lemma:eq1} can also be
adapted in view of properties of weak reactive stochastic timed
automata in order to obtain the following lemma:

\begin{restatable}{lem}{weakreactive}
  \label{lemma:weakreactive}
  Let $\A$ be a weak reactive stochastic timed automaton and $s$ be a
  state of $\A$. Then 
  \[
  \Prob\left(\bigcup_{q \in \mathcal{Q}} \mathfrak{R}^{q}(s) \right)
  = 1\,.
  \]
\end{restatable}

Thanks to these two lemmas, we deduce the almost-sure fairness for weak reactive automata as in the case of reactive  stochastic timed automata:

\begin{prop}
  \label{prop:fair-weakreactive}
  Let $\A$ be a weak reactive stochastic timed automaton and $s_0$ be a
  state of $\A$. Then $\Prob_{\A}(s \models \mathsf{fair}) = 1$.
\end{prop}

We end up this subsection by exhibiting some examples of weak reactive
stochastic timed automata:

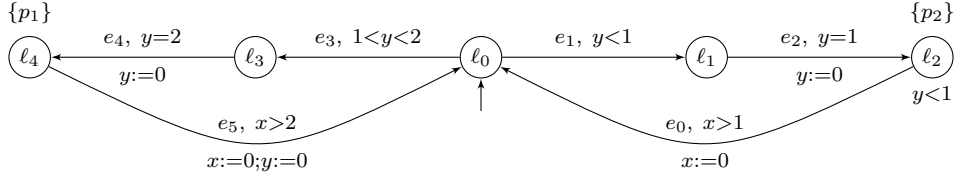
\begin{figure}[h]
  \begin{center}
    \begin{tikzpicture}[yscale=.85]
      \everymath{\scriptstyle}
      \path[use as bounding box] (-6,-1.4) -- (6,.7);
      \path (0,0) node[draw,circle,inner sep=2pt] (q0) {$\ell_0$};
           \path (0,-1) node[] (q0b) {};
      
      \path (3,0) node[draw,circle,inner sep=2pt] (q1) {$\ell_1$};
      
      \path (6,0) node[draw,circle,inner sep=2pt] (q3) {$\ell_2$};
      \path (6,-.6) node[] (q3b) {$y < 1$};
      \path (6,.7) node[] (q3a) {$\{p_2\}$};
      
      \path (-3,0) node[draw,circle,inner sep=2pt] (q4) {$\ell_3$};
           
      \path (-6,0) node[draw,circle,inner sep=2pt] (q6) {$\ell_4$};
      \path (-6,0.7) node[] (q6a) {$\{p_1\}$};
           
      \draw[arrows=-latex'] (q0b) -- (q0);
      
      \draw[arrows=-latex'] (q0) -- (q1) node[pos=.5, above,sloped]
      {$e_1,~y<1$};
      
      \draw[arrows=-latex'] (q1) -- (q3) node[pos=.5, above,sloped]
      {$e_2,~y=1$} node[pos=.5, below, sloped] {$y:=0$};
      
      \draw[arrows=-latex'] (q3) .. controls +(210:3.5cm).. (q0)
      node[pos=.5, above,sloped] {$e_0,~x>1$} node[pos=.5, below, sloped]
      {$x:=0$};
      
      \draw[arrows=-latex'] (q0) -- (q4) node[pos=.5, above,sloped]
      {$e_3,~1<y<2$};
      
      \draw[arrows=-latex'] (q4) -- (q6) node[pos=.5, above,sloped]
      {$e_4,~y=2$} node[pos=.5, below, sloped] {$y:=0$};
      
      \draw[arrows=-latex'] (q6) .. controls +(330:3.5cm).. (q0)
      node[pos=.5, above,sloped] {$e_5,~x>2$} node[pos=.5, below,
      sloped] {$x:=0;y:=0$};
    \end{tikzpicture}
  \end{center}
  \caption{Automaton $\A^1_{\mathsf{unfair}}$.}
  \label{fig:unfair1}
\end{figure}

\begin{exa}\hfill
  \begin{enumerate}
  \item Let $\A$ be a stochastic timed automaton such that for any
    $\ell\in L_u$, there exists a probability distribution $\mu_\ell$
    over $\IR_+$, equivalent to the Lebesgue measure, such that for
    every $v \in \IR_+^X$, $\mu_{(\ell,v)}=\mu_{\ell}$.  If $L_b$ does
    not contain any cycle then $\A$ is weak reactive.
  \item Let $\lambda>0$. The modification of $\A_{\mathsf{unfair}}$
    given on Figure~\ref{fig:unfair1}, where $\mu_{(\ell_4,v)}$ is the
    exponential distribution of parameter $\lambda$, is weak reactive.

    The other modification of $\A_{\mathsf{unfair}}$ given on
    Figure~\ref{fig:unfair2}, where $\mu_{(\ell_2,v)}$ is the
    exponential distribution of parameter $\lambda$, is not weak
    reactive. \corn
  \end{enumerate}
\end{exa}

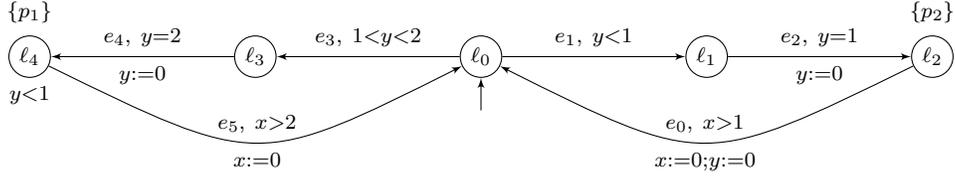
\begin{figure}[h]
  \begin{center}
    \begin{tikzpicture}[yscale=.85]
      \everymath{\scriptstyle}
      \path[use as bounding box] (-6,-1.4) -- (6,.7);
      \path (0,0) node[draw,circle,inner sep=2pt] (q0) {$\ell_0$};
         \path (0,-1) node[] (q0b) {};
      
      \path (3,0) node[draw,circle,inner sep=2pt] (q1) {$\ell_1$};
      
      \path (6,0) node[draw,circle,inner sep=2pt] (q3) {$\ell_2$};
         \path (6,.7) node[] (q3a) {$\{p_2\}$};
      
      \path (-3,0) node[draw,circle,inner sep=2pt] (q4) {$\ell_3$};
         
      \path (-6,0) node[draw,circle,inner sep=2pt] (q6) {$\ell_4$};
      \path (-6,0.7) node[] (q6a) {$\{p_1\}$};
      \path (-6,-.6) node[] (q6b) {$y < 1$};
      
      \draw[arrows=-latex'] (q0b) -- (q0);
      
      \draw[arrows=-latex'] (q0) -- (q1) node[pos=.5, above,sloped]
      {$e_1,~y<1$};
      
      \draw[arrows=-latex'] (q1) -- (q3) node[pos=.5, above,sloped]
      {$e_2,~y=1$} node[pos=.5, below, sloped] {$y:=0$};
      
      \draw[arrows=-latex'] (q3) .. controls +(210:3.5cm).. (q0)
      node[pos=.5, above,sloped] {$e_0,~x>1$} node[pos=.5, below, sloped]
      {$x:=0;y:=0$};
      
      \draw[arrows=-latex'] (q0) -- (q4) node[pos=.5, above,sloped]
      {$e_3,~1<y<2$};
      
      \draw[arrows=-latex'] (q4) -- (q6) node[pos=.5, above,sloped]
      {$e_4,~y=2$} node[pos=.5, below, sloped] {$y:=0$};
      
      \draw[arrows=-latex'] (q6) .. controls +(330:3.5cm).. (q0)
      node[pos=.5, above,sloped] {$e_5,~x>2$} node[pos=.5, below,
      sloped] {$x:=0$};
    \end{tikzpicture}
  \end{center}
  \caption{Automaton $\A^2_{\mathsf{unfair}}$.}
  \label{fig:unfair2}
\end{figure}

\subsubsection*{Discussion on Zenoness.}

A side-result of the proofs of Lemmas \ref{lemma:eq1} and
\ref{lemma:weakreactive} is that the set of Zeno runs in a (weak)
reactive stochastic timed automaton is negligible, which implies in
particular that the almost-sure non-Zenoness problem is trivial.

\begin{restatable}{prop}{reactivenonzeno}
  Let $\A$ be a (weak) reactive stochastic timed automaton, and $s$ be
  a state of $\A$. Then $\Prob_{\A}(s \models \mathsf{Zeno})=0$.
\end{restatable}

Intuitively, this is because almost-surely, time will increase by a
lower-bounded amount. The proof requires details of those of Lemmas
\ref{lemma:eq1} and \ref{lemma:weakreactive}, and is therefore
postponed to the Appendix, on page~\pageref{app:reactivenonzeno}.

\section{Decidability and complexity results}
\label{sec:algo}

We will use the thick graph construction for deciding 
almost-sure model-checking problem. For safety properties, we will use
the characterisation given in Theorem~\ref{th:safety}, whereas we will
use condition $(\S)$ for more general properties, under the assumption
that fairness is almost-sure.

\begin{lem}
  \label{lemma:complex}
  Let $\A$ be a timed automaton over \AP, $s$ a state of $\A$, and $P$
  be a B\"uchi or Muller property over \AP. Assume that
  $\thickgraph(\A)$ has size $f(\A)$, then we can decide in
  non-deterministic $\log(f(\A))$-space whether condition $\S(\A,s,P)$
  holds.
\end{lem}

\begin{proof}
  The algorithm for checking condition $\S(\A,s,P)$ is the following:
  \begin{itemize}
  \item guess a state $q$ of $\thickgraph(\A)$;
  \item check that there is a (thick) path from $[s]$ to $q$ in
    $\thickgraph(\A)$;
  \item check that $q$ belongs to a BSCC of $\thickgraph(\A)$ which
    satisfies property $P$.
  \end{itemize}
  Note that as $P$ is prefix-independent, we can assume the thick path
  from $[s]$ to $q$ is simple.  All this can be done on-the-fly and in
  non-deterministic $\log(f(\A))$-space. Note that checking $P$ in a
  BSCC can be done in $\log(f(\A))$-space as well. 
\end{proof}

\begin{lem}
  \label{lemma:complex2}
  Let $\A$ be a timed automaton over \AP, $s$ a state of $\A$, and $P$
  be a simple safety property over \AP. Assume that $\thickgraph(\A)$
  has size $f(\A)$, then we can decide in non-deterministic
  $\log(f(\A))$-space whether there is an infinite path $\pi$ from
  $\iota(s)$ in $\thickgraph(\A)$ which does not satisfy~$P$.
\end{lem}

\begin{proof}
  Here is a possible algorithm with the expected complexity:
  \begin{itemize}
  \item guess a state $q$ of $\thickgraph(\A)$;
  \item check that there is a (thick) path from $[s]$ to $q$ in
    $\thickgraph(\A)$;
  \item check that this finite path violates property $P$. 
  \end{itemize}
  Note that as $P$ is a simple safety property, we can assume that the
  path between $[s]$ and $q$ is simple. 
\end{proof}

We now state the following lemma, whose proof follows from the
definition of $\A \ltimes \B$ and from the definition of (weak)
reactiveness. This will imply that, under the mentioned conditions,
$\A \ltimes \B$ is almost-surely fair, which will allow to apply
Theorem~\ref{theo:specauto} to (weak) reactive stochastic timed
automata.

\begin{lem}
  Let $\A$ be a (weak) reactive stochastic timed automaton, and $\B$
  be a specification timed automaton. Then $\A \ltimes \B$ is a (weak)
  reactive stochastic timed automaton.
\end{lem}

We apply the results from Section~\ref{sec:match}, and obtain the
following principal theorem, which states decidability and complexity
results for the almost-sure model-checking. Notice that in all those
cases, almost-sure model-checking coincides with large model-checking.
\begin{restatable}{thm}{main}
  \label{theo:main}
  \begin{enumerate}[label=(\roman*)]
  \item The almost-sure model-checking of stochastic timed automata
    for simple safety properties is \PSPACE-complete.
  \item The almost-sure model-checking of single-clock stochastic
    timed automata for B\"uchi or Muller properties, or for
    properties given as specification (untimed) automata, is
    \NLOGSPACE-complete.\footnote{Note that simple safety or simple
      reachability properties can be expressed as small specification
      untimed automata, which yield an \NLOGSPACE upper bound in those
      cases as well.}
  \item The almost-sure model-checking of single-clock stochastic
    timed automata for properties given as \LTL formulas is
    \PSPACE-complete.
  \item The almost-sure model-checking of (weak) reactive stochastic
    timed automata for B\"uchi or Muller properties or properties
    given as specification timed automata is \PSPACE-complete.
  \end{enumerate}
\end{restatable}

\noindent All upper bounds are then obtained \textit{via}
Lemmas~\ref{lemma:complex} and~\ref{lemma:complex2}, since the size of
the region automaton (and therefore the thick graph) is
exponential~\cite{AD94}, except for single-clock timed automata, where
it is only polynomial-size~\cite{LMS04}. Note that if $\A$ is a timed
automaton and $\B$ a specification automaton, then the size of
$\thickgraph(\A \ltimes \B)$ is exponential in the two following
cases: $\B$ is a specification timed automaton, and $\B$ is a
specification untimed automaton of size at most exponential.  The
lower bounds are proven in Appendix~\ref{app:lower-bounds}. Note that
to establish the lower-bound in (i), the classical \PSPACE-hardness
proof of reachability in timed automata has to be adapted, since it is
based on punctual guards that would yield negligible behaviours in the
context of stochastic timed automata.

\section{Conclusion}
In this article we introduced and studied the model of stochastic
timed automata that combines real-time constraints and
probabilities. We considered the almost-sure model-checking problem
and designed an abstraction that can be used to prove decidability of
the above, provided fairness is almost-sure in the model. We
identified two main classes of automata for which this is the case,
the class of single-clock timed automata and that of weak reactive
timed automata. In the two cases, the proof of almost-sure fairness is
non-trivial and requires intricate arguments.

A remaining open problem is the decidability status of the almost-sure
model-checking problem for the general class of stochastic timed
automata, already for reachability properties.

As future work, we want to extend our study to quantitative
model-checking, that is, compute the probability of a given property
in an automaton. This has partly been solved for single-clock automata
in~\cite{BBBM08}, but more importantly, we would like to do it for the
class of (weak) reactive stochastic timed automata, which allows for
more complex timed constraints.

Compositionality is often a key for the description of real
systems. Defining a composition of stochastic timed automata seems
non-trivial in general, but the model of reactive stochastic timed
automata seems to be well-suited for compositional design, since time
can never be blocked by a component.

\section*{Acknowledgement}
Nathalie Bertrand and Patricia Bouyer were partly supported by ANR
project ImpRo (number ANR-10-BLAN-0317). Patricia Bouyer was
furthermore supported by the ERC Starting grant EQualIS (number
308087). Thomas Brihaye was partly supported by the ARC project
(number AUWB-2010-10/15-UMONS-3), a grant ``Mission Scientifique''
from the F.R.S.-FNRS, the FRFC project (number 2.4545.11), and the
EU-FP7 project CASSTING (number 601148). Quentin Menet was supported
by a grant of FRIA.  Christel Baier was partly supported by the German
Research Council (DFG) through the collaborative research centre 912
Highly- Adaptive Energy-Efficient Computing (HAEC) and the cluster of
excellence Centre for Advancing Electronics Dresden (cfAED), the
EU-FP7 project MEALS (295261) and by the EU and the State Saxony
through the ESF young researcher groups IMData (100098198) and SREX
(100111037).

\newcommand{\etalchar}[1]{$^{#1}$}


\appendix
\newpage
\makeatletter
\let\lem\relax
\let\defi\relax
\let\thm\relax
\let\c@thm\relax
\let\rem\relax
\let\exa\relax
\let\cor\relax
\let\prop\relax
\theoremstyle{plain}
\newtheorem{thm}{Theorem}[section]
\newtheorem{cor}[thm]{Corollary}
\newtheorem{lem}[thm]{Lemma}
\newtheorem{prop}[thm]{Proposition}

\makeatother

In this technical appendix, statements in boxes refer to statements
given in the core of the paper, and whose proofs were postponed to the
appendix. All other statements are local to the appendix.

\section{Details for Section~\ref{sec:proba-def}}
\label{annex:proba-measure}

\noindent\fbox{\begin{minipage}{.98\linewidth}
\propprobameasure* \end{minipage}}

\begin{proof}
We first recall a basic property in measure theory~\cite{KSK76}.

\begin{prop}
  \label{prop:KSK76}
  Let $\nu$ be a non-negative additive set function defined on some set space
  ${\mathcal F}$ such that for every $A \in {\mathcal F}$, $\nu(A) < \infty$. The
  three following properties are equivalent:
  \begin{enumerate}
  \item $\nu$ is $\sigma$-additive,
  \item for every sequence $(A_n)_n$ of elements of ${\mathcal F}$ such that $A_0
    \subseteq A_1 \subseteq A_2 \subseteq \cdots$ and $A = \bigcup_n A_n \in
    {\mathcal F}$, $\lim_n \nu(A_n) = \nu(A)$,
  \item for every sequence $(B_n)_n$ of elements of ${\mathcal F}$ such that $B_0
    \supseteq B_1 \supseteq B_2 \supseteq \cdots$ and $\bigcap_n B_n =
    \emptyset$, $\lim_n \nu(B_n) = 0$.
  \end{enumerate}
\end{prop}

For every $n \in \IN$, we write ${\mathcal F}_n(s)$ for the
ring\footnote{A \emph{ring} $R \subseteq 2^S$ is such that $\emptyset
  \in R$, $R$ is closed by finite union and by complement.} generated
by the set of (basic) cylinders from $s$ of length $n$, \textit{i.e.},
all $\Cyl(\path[\mathcal{C}]{s,e_1\ldots e_n})$. The elements of
${\mathcal F}_n(s)$ are thus finite unions of basic cylinders of length
$n$. We then define
\[
{\mathcal F}(s) = \bigcup_n {\mathcal F}_n(s)
\]

\begin{lem}\label{proba Fn}
  For every $n$, $\Prob_{\A}$ is a probability measure on ${\mathcal
    F}_n(s)$.
\end{lem}

\begin{proof}
  First, by induction on $n$, it is not difficult to prove that for
  every $n \in \IN$,
  \begin{equation}
    \label{1}
    \sum_{(e_1,\ldots,e_n)} \Prob_\A(\path{s,e_1\ldots e_n}) =
    \Prob_\A(\path{s}) = 1
  \end{equation}

  We fix $n \in \IN$. $\Prob_\A$ is obviously additive, non-negative
  and finite over ${\mathcal F}_n(s)$. Take a sequence $(A_i)_i$ of
  elements of ${\mathcal F}_n(s)$ such that $A_0 \subseteq A_1 \subseteq
  A_2 \subseteq \cdots$ and $A = \bigcup_i A_i \in {\mathcal
    F}_n(s)$. There are finitely many distinct sequences of edges of
  length $n$. Hence, intersecting each of the $A_i$'s with each of the
  symbolic paths $\path{s,e_1\ldots e_n}$ of length $n$, we can assume
  w.l.o.g. that each $A_i$ is a single constrained finite symbolic
  path.

  Let $e_1 \ldots e_n$ be the sequence of edges underlying all constrained
  symbolic paths $A_i$, and write $\mathcal{C}_i$ for the tightest constraint
  defining $A_i$ (\textit{i.e.}, $A_i = \path[\mathcal{C}_i]{s,e_1\ldots
    e_n}$). We have that $\mathcal{C}_i \subseteq \mathcal{C}_{i+1}$, and
  $(\mathcal{C}_i)_i$ converges to $\mathcal{C}$, which corresponds to the
  constraint associated with $A$. We can write, if $\car\alpha$ is the
  characteristic function of set $\alpha$, that:
  \begin{multline*}
    \lim_i \Prob_\A(A_i) = \lim_i \int_{\tau_1 \in I(s,e_1)}
    p_{s+\tau_1}(e_1) \int_{\tau_2 \in I(s_{\tau_1},e_2)}
    p_{s_{\tau_1}+ \tau_2}(e_2) \cdots \\  
    \int_{\tau_n \in I(s_{\tau_1 \cdots
        \tau_{n-1}},e_n)} 
    p_{s_{\tau_1
        \cdots \tau_{n-1}}+\tau_n}(e_n)\, \car{\mathcal{C}_i}(\tau_1,\ldots,\tau_n) \, \ud
    \mu_{s_{\tau_1
        \cdots \tau_{n-1}}}(\tau_n) \cdots \ud \mu_s(\tau_1)
  \end{multline*}
  \begin{multline*}
    = \int_{\tau_1 \in I(s,e_1)}
    p_{s+\tau_1}(e_1) \int_{\tau_2 \in I(s_{\tau_1},e_2)}
    p_{s_{\tau_1}+ \tau_2}(e_2) \cdots \\
    \int_{\tau_n \in I(s_{\tau_1 \cdots
        \tau_{n-1}},e_n)} 
    p_{s_{\tau_1
        \cdots \tau_{n-1}}+\tau_n}(e_n)\, \Big(\lim_i \car{\mathcal{C}_i}(\tau_1,\ldots,\tau_n)\Big)\, \ud
    \mu_{s_{\tau_1
        \cdots \tau_{n-1}}}(\tau_n) \cdots \ud \mu_s(\tau_1) \\
\text{{\small (by dominated convergence and equation~(\ref{1}))}}
  \end{multline*}
  \begin{multline*}
    = \int_{\tau_1 \in I(s,e_1)}
    p_{s+\tau_1}(e_1) \int_{\tau_2 \in I(s_{\tau_1},e_2)}
    p_{s_{\tau_1}+ \tau_2}(e_2) \cdots \\
    \int_{\tau_n \in I(s_{\tau_1 \cdots
        \tau_{n-1}},e_n)} 
    p_{s_{\tau_1
        \cdots \tau_{n-1}}+\tau_n}(e_n)\, \car{\mathcal{C}}(\tau_1,\ldots,\tau_n)\, \ud
    \mu_{s_{\tau_1
        \cdots \tau_{n-1}}}(\tau_n) \cdots \ud \mu_s(\tau_1)
\end{multline*}
\begin{multline*}
= \Prob_{\A}(A) \\
  \end{multline*}
This shows that $\Prob_{\A}$ is a measure on ${\mathcal F}_n(s)$, for all
$n \in \mathbb{N}$. It is moreover a probability measure since
$\Prob_{\A} ({\mathcal F}_n(s)) = \Prob_{\A}(\path{s}) =1$.
\end{proof}

\begin{lem}
  $\Prob_\A$ is a probability measure on $\mathcal{F}(s)$.
\end{lem}

\begin{proof}

  Obviously $\Prob_{\A}$ is non-negative on ${\mathcal F}(s)$, additive
  (because ${\mathcal F}_n(s) \subseteq {\mathcal F}_{n+1}(s)$ for every $n
  \in \IN$) and finite over ${\mathcal F}(s)$. It remains to prove that it
  is $\sigma$-additive.  For this, we use
  Proposition~\ref{prop:KSK76}, and consider a sequence $(B_n)_n$ of
  sets in ${\mathcal F}(s)$ such that $B_0 \supseteq B_1 \supseteq B_2
  \supseteq \cdots$ and $\bigcap_n B_n = \emptyset$. W.l.o.g. we
  assume that for every $n$, $B_n \in {\mathcal F}_n(s)$. We want to prove
  that $\lim_n \Prob_{\A}(B_n) = 0$.  Applying a reasoning similar to
  that of~\cite[Lemmas~2.1, 2.2, 2.3]{KSK76}, it is sufficient, thanks
  to Lemma~\ref{proba Fn}, to do the proof when $B_n$ is some
  $\Cyl(\pi_n)$ where $\pi_n$ is a finite (constrained) symbolic path
  of length $n$. We write $\mathcal{C}_n$ for the tightest constraint
  over variables $(\tau_i)_{i \leq n}$ corresponding to $\pi_n$. We
  define $p_i$ the constraint from $\IR_+^{i+1}$ onto the $i$ first
  components (thus in $\IR_+^i$). Note that this projection is
  continuous (for the product topologies). In $\pi_n$, if $i < n$, the
  $i$ first variables are constrained by $\mathcal{C}_n^i =
  p_i(\mathcal{C}_n^{i+1})$.  Moreover, for every $i \leq n$, we have
  that
  \[
  \mathcal{C}_{n+1}^i \subseteq \mathcal{C}_n^i\quad \text{and}\quad
  \mathcal{C}_n^i \subseteq \mathcal{C}_n^{i-1}
  \]

  Fix some $i$, the sequence $(\mathcal{C}_n^i)_n$ is nested, hence converges
  to $\mathcal{C}^i$, and $\mathcal{C}^i \subseteq \mathcal{C}^{i-1}$. By
  continuity of the projection over the $i$ first components, we have that
  $\mathcal{C}^i = p_i(\mathcal{C}^{i+1})$. If none of the $\mathcal{C}^i$ is
  empty, we can thus construct an element in $\bigcap_i \mathcal{C}^i$ as
  follows: we take some $\tau_1$ satisfying the constraint $\mathcal{C}^1$; we
  have that $\mathcal{C}^1 = p_1(\mathcal{C}^2)$ (and $\mathcal{C}^2$ is a
  constraint over $\tau_1$ and $\tau_2$), hence there exists $\tau_2$ such
  that $(\tau_1,\tau_2)$ satisfies $\mathcal{C}^2$ (while $\tau_1$ still
  satisfies $\mathcal{C}^1$); we do the same step-by-step for all $\tau_i$ and
  construct a sequence $(\tau_i)_i$ which satisfies all constraints
  $\mathcal{C}^i$. This sequence corresponds to a run in $\bigcap_i
  \Cyl(\pi_i)$. As we assumed at the beginning of the paragraph that
  $\bigcap_i \Cyl(\pi_i) = \emptyset$, it thus means that there exists some $i
  \in \IN$ such that $\mathcal{C}^i = \emptyset$.

  We will use the fact that $\mathcal{C}^i = \bigcap_{n \geq i}
  \mathcal{C}^i_n$ is empty to prove that $\lim_n \Prob_\A(\pi_n) = 0$. We
  write, still with the notation that $\car\alpha$ is the characteristic
  function of set $\alpha$:
  \begin{multline*}
    \Prob_\A(\Cyl(\pi_n)) = \int_{\tau_1 \in I(s,e_1)}
    p_{s+\tau_1}(e_1) \int_{\tau_2 \in I(s_{\tau_1},e_2)}
    p_{s_{\tau_1}+ \tau_2}(e_2) \cdots   \\
     \int_{\tau_n \in I(s_{\tau_1 \cdots
        \tau_{n-1}},e_n)} p_{s_{\tau_1
        \cdots \tau_{n-1}}+\tau_n}(e_n)\,
    \car{\mathcal{C}_n}(\tau_1,\ldots,\tau_n)\, 
    \ud \mu_{s_{\tau_1
        \cdots \tau_{n-1}}}(\tau_n) \cdots \ud \mu_s(\tau_1) 
  \end{multline*}
  \begin{multline*}
    \le \int_{\tau_1 \in I(s,e_1)}
    \int_{\tau_2 \in I(s_{\tau_1},e_2)}
    \cdots \int_{\tau_i \in I(s_{\tau_1 \cdots
        \tau_{i-1}},e_i)}
    \car{\mathcal{C}^i_n}(\tau_1,\ldots,\tau_i)\, \ud
    \mu_{s_{\tau_1
        \cdots \tau_{i-1}}}(\tau_i) \cdots \ud \mu_s(\tau_1)
  \end{multline*}
  Applying the dominated convergence theorem, we get that:
  \begin{eqnarray*}
    \lim_n \Prob_\A(\Cyl(\pi_n)) & = & \int_{\cdots} \int_{\cdots} \cdots \int_{\cdots}
    \Big(\lim_n \car{\mathcal{C}^i_n}(\tau_1,\ldots,\tau_i)\Big)\, \ud
    \mu_{s_{\tau_1
        \cdots \tau_{i-1}}}(\tau_i) \cdots \ud \mu_s(\tau_1) \\
    & = & 0
  \end{eqnarray*}
  This concludes the proof that $\Prob_{\A}$ is $\sigma$-additive on
  $\mathcal{F}(s)$, and thus the proof that $\Prob_{\A}$ is a probability measure
  on ${\mathcal F}(s)$.
\end{proof}

We conclude the proof using the following classical measure extension theorem:

\begin{thm}[Carath\'eodory's extension theorem]
  \label{theo:caratheodory}
  Let $S$ be a set, and $\nu$ a $\sigma$-finite measure defined on a ring $R
  \subseteq 2^S$. Then, $\nu$ can be extended in a unique manner to the
  $\sigma$-algebra generated by $R$.
\end{thm}

We apply Theorem~\ref{theo:caratheodory} to the set $S = \Runs(\A,s)$, $R =
\mathcal{F}(s)$, and $\nu = \Prob_{\A}$ which is a $\sigma$-finite measure on
$\mathcal{F}(s)$. Hence, there is a unique extension of $\Prob_{\A}$ on
$\Omega_{\A}^s$, the $\sigma$-algebra generated by the cylinders, which is a
probability measure.
\end{proof}

\ligne 
\label{annex:region}

\noindent\fbox{\begin{minipage}{.98\linewidth}
\lemmaproba* \end{minipage}}

\begin{proof}
  It is sufficient to prove that the measures coincide on finite
  constrained paths, since it implies that they agree on cylinders and
  by uniqueness of the extension on any measurable set of infinite
  runs.

  In this proof we will denote transitions of $\A$ by $e_i$
  and transition in $\RA$ by $f_i$. We prove that
  $\Prob_\A$ and $\Prob_{\RA}$ coincide on finite paths by induction
  on the length $n$ of constrained symbolic paths. When $n=0$, this is
  obvious as, for every $(\ell,\nu)$, there is a single state
  $((\ell,r),\nu)$ in $\RA$ such that $\nu \in r$, and in that case,
  $\iota(\path{(\ell,\nu)}) = \{\path{((\ell,r),\nu)}\}$. We assume
  the induction hypothesis holds for all constrained paths of length
  strictly smaller than $n$.

  We will use the following notations (this will be technical, but
  rather simple): given $s$ a state, we recall that $s+t$ is the
  state reached from $s$ after a delay $t$, $[s]$ is the region to
  which $s$ belongs. If $q$ is a state of the region automaton, we
  write $n_q$ for the number of edges enabled without delay in $q$ in
  $\RA$ (or equivalently in $\A$). If transition $e_1$ can be taken
  from $q$ without delay, $e_1(q)$ denotes the single image region
  reached after firing $e_1$ from $q$, and we write $q \models f_1$ if
  $f_1$ is the unique transition with guard checking that we are in
  $q$ and corresponding to $e_1$ in $\RA$.

  Let $\pi=\path[\mathcal{C}]{s,e_1,\ldots,e_n}$ be a constrained
  symbolic path in $\A$. Constraint $\mathcal{C}$ is on $n$ variables
  $\tau_1 \cdots \tau_n$. We will denote $\mathcal{C}_t$ the
  constraint obtained from $\mathcal{C}$ by replacing $\tau_1$ by $t$.

\[\eqalign{
  \Prob_{\A}(\pi) 
&=\int_{t \in I(s,e_1)} p^{\A}_{s+t}(e_1) \,
  \Prob_{\A}(\path[\mathcal{C}_t]{s_t,e_2 \ldots e_n}) \, \ud
  \mu^{\A}_s(t) \cr
&=\int_{t \in I(s,e_1)} p^{\A}_{s+t}(e_1) \,
  \Prob_{\RA}(\iota(\path[\mathcal{C}_t]{s_t,e_2 \ldots e_n}))\, \ud
  \mu^{\A}_s(t) \quad \text{by induction hypothesis} \cr
&=\int_{\begin{array}{c} {\scriptscriptstyle t \in
        I(s,e_1)} \end{array}} p^{\A}_{s+t}(e_1) \sum_{\pi' \in
    \iota(\path{s_t,e_2 \ldots e_n})}
  \Prob_{\RA}(\pi'_{\mathcal{C}_t}) \, \ud \mu^{\A}_s(t) \cr
&=\sum_q \int_{\begin{array}{c} {\scriptscriptstyle t \in I(s,e_1)}\cr 
  {\scriptscriptstyle s+t \in q} \end{array}}
  p^{\A}_{s+t}(e_1) \hspace*{-.5cm}\sum_{\pi' \in \iota(\path{s_t,e_2
      \ldots e_n})}\hspace*{-.5cm} \Prob_{\RA}(\pi'_{\mathcal{C}_t})
  \, \ud \mu^{\A}_s(t) \cr
&=\sum_q \int_{\begin{array}{c} {\scriptscriptstyle t \in I(s,e_1)}\cr
  {\scriptscriptstyle s+t \in q} \end{array}}
  p^{\A}_{s+t}(e_1) \hspace*{-1cm}\sum_{\begin{array}{c}
      {\scriptscriptstyle (f_2, \ldots f_n) \in
        \iota(e_1(q),e_2,\ldots,e_n)} \end{array}}\hspace*{-1cm}
  \Prob_{\RA}(\path[\mathcal{C}_t]{\iota(s_t),f_2 \ldots f_n}) \, \ud
  \mu^{\A}_s(t) \cr
&=\sum_q \int_{\begin{array}{c} {\scriptscriptstyle t \in
        I(\iota(s),f_1)} \cr
  {\scriptscriptstyle s+t \in q} \\[-.2cm]{\scriptscriptstyle  q \models f_1} \\[-.1cm] {\scriptscriptstyle [s] \xrightarrow{f_1} e_1(q)} \end{array}} p^{\RA}_{\iota(s)+t}(f_1)\hspace*{-1cm}\sum_{\begin{array}{c} {\scriptscriptstyle (f_2,\ldots,f_n) \in \iota(e_1(q),e_2,\ldots,e_n)} \end{array}} \hspace*{-1cm}
 \Prob_{\RA}(\path[\mathcal{C}_t]{\iota(s)_t,f_2,\ldots,f_n}) \, \ud
 \mu^{\RA}_{\iota(s)}(t)
 \llap{\lower 33 pt\hbox{by hypothesis on the measures and weights}} \cr
&=\sum_{\begin{array}{c} {\scriptscriptstyle q \models f_1}\cr
  {\scriptscriptstyle [s] \xrightarrow{f_1} e_1(q)} \\[-.2cm]
  {\scriptscriptstyle (f_2,\ldots,f_n) \in
    \iota(e_1(q),e_2,\ldots,e_n)} \end{array}} \hspace*{-1cm} \int_{t
   \in I(\iota(s),f_1)} p^{\RA}_{\iota(s)+t}(f_1) \,
 \Prob_{\RA}(\path[\mathcal{C}_t]{\iota(s)_t,f_2,\ldots,f_n}) \,\ud
 \mu^{\RA}_{\iota(s)}(t) \cr
&=\sum_{\begin{array}{c} {\scriptscriptstyle q \models f_1} \\[-.1cm]
     {\scriptscriptstyle [s] \xrightarrow{f_1} e_1(q)} \\[-.2cm]
     {\scriptscriptstyle (f_2,\ldots,f_n) \in
       \iota(e_1(q),e_2,\ldots,e_n)}  \end{array}} \hspace*{-.8cm}\Prob_{\RA}(\path[\mathcal{C}]{\iota(s),f_1,\ldots,f_n})
 \cr
&=\Prob_{\RA}(\iota(\pi))
  }
\]
  \noindent where $(f_2,\ldots,f_n) \in \iota(e_1(q),e_2,\ldots,e_n)$
  iff $(f_2,\ldots,f_n)$ is a finite sequence of transitions
  corresponding to $(e_2,\ldots,e_n)$ and which starts in $(e_1(q))$
  (this is a state of $\RA$). 
 \end{proof}

\ligne

\noindent\fbox{\begin{minipage}{.98\linewidth}
\mesurabilite* \end{minipage}}

\label{app:mesurabilite}

\begin{proof}
  It is sufficient to do the proof in the case of specifications given
  as deterministic timed automata. Indeed, it covers also the case of
  $\omega$-regular and \LTL-properties, since they can be turned into
  a deterministic untimed Muller automaton.

  Let $\langle \A,\mu,w \rangle$ be a stochastic timed automaton, and
  $\B$ a specification automaton (that is, a deterministic complete
  timed automaton). We prove that the set of runs in $\A$ that are
  accepted by $\B$ is measurable (for the probability measure defined
  by $\langle \A,\mu,w \rangle$).  To do so, we consider the product
  timed automaton $\A \ltimes \B$ (see definition on
  page~\pageref{def:produit-TA-spec}). Let $\R{\A \ltimes \B}$ be its
  (untimed) region automaton, and $\Fcal$ the accepting condition
  naturally derived from the one of $\B$. The set of paths in $\R{\A
    \ltimes \B}$ satisfying $\Fcal$ is a Boolean combination of
  cylinders $\Cyl(\bs_0, \be_1 \dots \be_n)$.  Indeed, since $\Fcal$
  is an $\omega$-regular condition and seeing $\R{\A \ltimes \B}$ as a
  finite Markov chain (with arbitrary probabilities), this is a
  consequence of the proof of measurability of $\omega$-regular
  properties~\cite{vardi85}. For a fixed finite path $\bs_0,\be_1
  \dots \be_n$ in $\R{\A \ltimes \B}$, we write
  \[
  H(\bs_0, \be_1 \dots \be_n) = \{ \rho \in \Runs(\A,s_0) \mid
  \iota(\rho^{\B}) \in \Cyl(\path{\bs_0, \be_1 \dots \be_n})\} \enspace.
  \]
  Roughly speaking, $H(\bs_0, \be_1 \dots \be_n)$ is the set of all
  runs in $\A$ whose natural projection in $\R{\A \ltimes \B}$ belongs
  to $\Cyl(\path{\bs_0, \be_1 \dots \be_n})$. One can be convinced
  that $H(\bs_0, \be_1 \dots \be_n)$ consists of a finite union of
  cylinders generated by constrained symbolic paths in $\A$. Hence the
  set of runs in $\A$ satisfying the specification $\B$ can be written
  as a Boolean combination of cylinders generated by constrained
  symbolic paths, and is therefore measurable.
\end{proof}

\section{Details for Section~\ref{sec:toposem}}

For Definition~\ref{def:topo} to properly define a topological space,
we prove that the intersection of two basic open sets is still a basic
open set. This is the object of the following result, whose proof
requires several technical intermediary lemmas.

\begin{lem}
  \label{lemma:intersection}
  Let $\pi_{\mathcal{C}} = \path[\mathcal{C}]{s,e_1,\ldots,e_n}$ and
  $\pi_{\mathcal{C}'} = \path[\mathcal{C}']{s,e_1,\ldots,e_n}$ be two
  basic open sets of same length. Then $\pi_{\mathcal{C}} \cap
  \pi_{\mathcal{C}'}$ is an open set.
\end{lem}
  
For the next lemmas, let us fix $\pi_{\mathcal{C}}
= \path[\mathcal{C}]{s,e_1,\ldots,e_n}$ and $\pi_{\mathcal{C}'} =
\path[\mathcal{C}']{s,e_1,\ldots,e_n}$ be two constrained symbolic
paths of same length, where $\mathcal{C}$ and $\mathcal{C}'$ are
Borel-measurable. For all $i \le n$, write $\mathcal{C}_i$
(resp. $\mathcal{C}'_i$) for the projection of $\mathcal{C}$
(resp. $\mathcal{C}'$) on the $i$ first coordinates. Write also
$\pi_{\mathcal{C}_i} = \path[\mathcal{C}_i]{s,e_1 \ldots e_i}$,
$\pi_{\mathcal{C}'_i} = \path[\mathcal{C}'_i]{s,e_1 \ldots e_i}$, and
$\pi_i = \path{s,e_1 \ldots e_i}$.

\begin{lem}
  \label{lemma:croissance}
  Assume $\pi_{\mathcal{C}} \subseteq \pi_{\mathcal{C}'}$ and
  $\dim(\Pol(\pi_{\mathcal{C}})) =
  \dim(\Pol(\pi_{\mathcal{C}'}))$. Then for all $i \leq n$,
  $\dim(\Pol(\pi_{\mathcal{C}_i})) = \dim(\Pol(\pi_{\mathcal{C'}_i}))$
\end{lem}

\begin{proof}
  Assume there exists an index $i \leq n$ such that
  $\dim(\Pol(\pi_{\mathcal{C}_i})) <
  \dim(\Pol(\pi_{\mathcal{C'}_i}))$. As $\dim(\Pol(\pi_{\mathcal{C}}))
  = \dim(\Pol(\pi_{\mathcal{C}'}))$ there must be an index $j$, such
  that $\Pol(\pi_{\mathcal{C}})$ gains some dimension in the $j$-th
  direction, whereas $\Pol(\pi_{\mathcal{C'}})$ does not. But this is
  not possible since $\pi_{\mathcal{C}} \subseteq \pi_{\mathcal{C}'}$
  and therefore $\Pol(\pi_{\mathcal{C}}) \subseteq
  \Pol(\pi_{\mathcal{C}'})$. 
\end{proof}

From this basic result, we get the following corollaries.

\begin{cor}
  \label{cor:full_dimension}
  If $\Pol(\pi_{\mathcal{C}})$ is open in $\Pol(\pi)$, then for all $i
  \leq n$, $\dim(\Pol(\pi_{\mathcal{C}_i})) = \dim(\Pol(\pi_i))$.
\end{cor}
\begin{proof}
  As $\Pol(\pi_\mathcal{C})$ is open in $\Pol(\pi)$, there exists an
  open set $O$ of $\mathbb{R}^n$ such that $\Pol(\pi_\mathcal{C}) = O
  \cap \Pol(\pi)$. This implies that $\dim(\Pol(\pi_\mathcal{C})) =
  \dim(\Pol(\pi))$.\footnote{We use here the following general
    topology result: if $X$ is a convex set and $O$ an open set in
    $\mathbb{R}^n$ such that $X \cap O \neq \emptyset$, then $\dim(X)
    = \dim(X \cap O)$.}  Applying Lemma~\ref{lemma:croissance} to
  $\pi_{\mathcal{C}}$ and $\pi$ yields the expected result. 
\end{proof}

\begin{cor}
  \label{coro2-essai3}
  Assume $\pi_{\mathcal{C}'}$ is a non-empty open set of
  $(\Runs(\A,s),\mathcal{T}_{\A}^s)$ and $\pi_{\mathcal{C}} \subseteq
  \pi_{\mathcal{C}'}$ and $\Pol(\pi_{\mathcal{C}})$ is open in
  $\Pol(\pi)$, then $\pi_{\mathcal{C}}$ is thick (that is,
  $\pi_{\mathcal{C}}$ is a non-empty basic open set of
  $(\Runs(\A,s),\mathcal{T}^s_{\A})$).
\end{cor}

\begin{proof}
  By Corollary~\ref{cor:full_dimension} applied to both
  $\pi_{\mathcal{C}}$ and $\pi$, and $\pi_{\mathcal{C}'}$ and $\pi$,
  we get  for every $1 \le i \le n$:
  \[
  \dim(\Pol(\pi_{\mathcal{C}_i})) = \dim(\Pol(\pi_{\mathcal{C}'_i})) =
  \dim(\Pol(\pi_i))\,.
  \] 
  As $\pi_{\mathcal{C}'}$ is a basic open set it also holds that for
  every $1 \leq i \leq n$:
  \[
  \dim(\Pol(\pi_{\mathcal{C}'_i})) = \dim(\bigcup_e
  \Pol(\path[\mathcal{C}'_{i-1}]{s,e_1\ldots e_{i-1} e}))\,.
  \]
  By containment of $\path[\mathcal{C}_{i-1}]{s,e_1\ldots e_{i-1} e}$
  into $\path[\mathcal{C}'_{i-1}]{s,e_1\ldots e_{i-1} e}$, we get that
  \[
  \dim(\bigcup_e \Pol(\path[\mathcal{C}'_{i-1}]{s,e_1\ldots e_{i-1}
    e})) \ge \dim(\bigcup_e \Pol(\path[\mathcal{C}_{i-1}]{s,e_1\ldots
    e_{i-1} e}))\,.
  \]
  This shows that $\pi_{\mathcal{C}}$ is thick: 
  \[
  \dim(\Pol(\pi_{\mathcal{C}_i})) \ge \dim(\bigcup_e
  \Pol(\path[\mathcal{C}_{i-1}]{s,e_1\ldots e_{i-1} e}))\,.
  \]
\end{proof}

We can now come to the proof of Lemma~\ref{lemma:intersection}.

\begin{proof}[of Lemma~\ref{lemma:intersection}]
  Let us denote in this proof $\mathcal{C''} = \mathcal{C} \cap
  \mathcal{C'}$, and $\pi$ the unconstrained symbolic path
  $\path{s,e_1,\ldots,e_n}$. Write $\pi_{\mathcal{C}''}$ for
  $\pi_{\mathcal{C}} \cap \pi_{\mathcal{C}'} =
  \path[\mathcal{C}'']{s,e_1,\ldots,e_n}$. If $\pi_{\mathcal{C}''}$ is
  empty, we are done since the empty set is an open set. We therefore
  assume that $\pi_{\mathcal{C}''}$ is non-empty.

  We first show that $\Pol(\pi_{\mathcal{C}} \cap \pi_{\mathcal{C}'})$
  is open in $\Pol(\pi)$, which is the second condition for
  $\pi_{\mathcal{C}} \cap \pi_{\mathcal{C}'}$ to be an open set. We
  have that $\Pol(\pi_{\mathcal{C}''}) = \Pol(\pi_{\mathcal{C}}) \cap
  \Pol(\pi_{\mathcal{C}'})$. By assumption both
  $\Pol(\pi_{\mathcal{C}})$ and $\Pol(\pi_{\mathcal{C}'})$ are open in
  $\Pol(\pi)$, hence their intersection too.

  The fact that $\pi_{\mathcal{C}''}$ is thick is a consequence of
  Corollary~\ref{coro2-essai3}.  We conclude that
  $\pi_{\mathcal{C''}}$ is an open set for our topology. 
  \end{proof}

\ligne

\noindent\fbox{\begin{minipage}{.98\linewidth}
\propBaire* \end{minipage}}
\label{app:Baire}

\begin{proof}
  To prove that $(\Runs(\A,s),{\mathcal T}_{\A}^s)$ is a Baire space, we
  prove that every non-empty basic open set in ${\mathcal T}_\A^s$ is not
  meagre. Let $\Cyl(\path[\mathcal{C}]{s,e_1 \ldots e_n})$ be a basic
  open set.
  Using Banach-Mazur games (see page~\pageref{BM:defi}
  or~\cite{oxtoby57}), we prove that $\Cyl(\path[\mathcal{C}]{s,e_1
    \ldots e_n})$ is not meagre by proving that Player~2 does not have
  a winning strategy for the Banach-Mazur game played with basic open
  sets and where the goal set is $C=\Cyl(\path[\mathcal{C}]{s,e_1
    \ldots e_n})$.
  
  Player~1 starts by choosing a set $B_1 = \Cyl(\path[\mathcal{C}]{s,e_1
    \ldots e_n})$. Then Player~2 picks some basic open set $B_2=
  \Cyl(\path[\mathcal{C}^2]{s,e_1 \ldots e_n \ldots e_{n_1}})$ such that $B_1
  \supseteq B_2$.

  Let us now explain how Player~1 can build her move in order to avoid
  to reach the empty set. Since $B_2$ is an open set, we have that
  (i) $\pi_{{\mathcal{C}^2}}$ is thick and (ii)
  $\Pol(\path[\mathcal{C}^2]{s,e_1 \ldots e_{n_1}})$ is open in
  $\Pol(\path{s,e_1 \ldots e_{n_1}}) \subseteq \IR_+^{n_1}$. The
  topology on $\Pol(\path{s,e_1 \ldots e_{n_1}})$ is induced from a
  distance, hence there exists a closed, bounded and convex set
  denoted $K_1$ such that $\mathring{K}_1 \ne \emptyset$ and $K_1
  \subseteq \Pol(\path[\mathcal{C}^2]{s,e_1 \ldots e_{n_1}})$. Let
  $\mathcal{D}^1$ be the constraint associated with $K_1$; clearly
  cylinder $\Cyl(\path[\mathcal{D}^1]{s,e_1 \ldots e_{n_1}})$ is
  included in $B_2$. Let $O$ be an open set included in $K_1$ and
  ${\mathcal C}^3$ be the constraint associated with $O$. It is
  Borel-measurable since it is open.
  Applying Corollary~\ref{coro2-essai3}, we know that
  $\path[\mathcal{C}^3]{s,e_1 \ldots e_{n_1}}$ is thick. Hence clearly
  enough, $\Cyl(\path[\mathcal{C}^3]{s,e_1 \ldots e_{n_1}})$ is an
  open set. Player~1's move will be to take $B_3 =
  \Cyl(\path[\mathcal{C}^3]{s,e_1 \ldots e_{n_1}})$. By iterating this
  process, we define a strategy for Player~1 which satisfies:
  \[
  B_1 \supseteq B_2 \supseteq \Cyl(\pi_{\mathcal{D}^1}) \supseteq B_3
  \supseteq B_4 \supseteq \Cyl(\pi_{\mathcal{D}^2}) \supseteq \ldots \supseteq
  B_{2i-1} \supseteq B_{2i} \supseteq \Cyl(\pi_{\mathcal{D}^i}) \supseteq
  \cdots
  \] 
  where for each $i$, $K_i = \Pol(\pi_{\mathcal{D}^i})$ is a closed and
  bounded subset of $\Pol(\pi(e_1,\ldots,e_{n_{i}})) \subseteq \IR_+^{n_{i}}$
  (where the $n_i$'s form a non-decreasing sequence of $\IN$). We then have
  that:
  \[ 
  \bigcap_{i=1}^\infty B_i \ = \ \bigcap_{i=1}^\infty
  \Cyl(\pi_{\mathcal{D}^i})\,.
  \]

  We would like to guarantee that the above intersection is
  non-empty. This is not completely straightforward since the
  polyhedra $K_i = \Pol(\pi_{\mathcal{D}^i})$ belong to different
  powers of $\IR_+$. We distinguish between two cases:
  \begin{itemize}
  \item either the sequence $(n_i)_{i \geq 1}$ diverges to $+\infty$. In that
    case, we will embed $\bigcap_{i=1}^\infty K_i$ into a compact set of
    $\IR_+^{\IN}$. We first define
    \[
    \widetilde{K}_j = \Proj_{\{n_{j-1}+1,\ldots,n_j\}}K_j \quad \text{ and }
    \quad \widetilde{K} = \prod_{j \ge 1} \widetilde{K}_j\,,
    \]  
    where $\Proj_I(K_J)$ for $I \subseteq \{1 \cdots n_j\}$ is the
    natural projection from $\IR_+^{n_j}$ to the coordinates specified
    by $I$.  Note that $\widetilde{K}_j$ is a compact set, since it is
    the projection of a compact set. Each $K_i$ can naturally be
    embedded in $\widetilde{K}$ by considering the sets $K'_i$ defined
    by
    \[
    K'_i = K_i \times \prod_{j > i} \widetilde{K}_j\,.
    \]
    The decomposition is illustrated on Figure~\ref{fig:decomposition}. The
    $K'_i$'s form a nested chain of closed sets of $\widetilde{K}$. By
    Tychonoff's theorem, $\widetilde{K}$ is compact. Hence we can
    ensure that $\bigcap_{i=1}^\infty K'_i$ is non-empty (Heine-Borel
    theorem). Take a sequence $(\tau_j)_{j \geq 1}$ in $\bigcap_{i=1}^\infty
    K'_i$. Each subsequence $(\tau_j)_{1 \leq j \leq n_i}$ straightforwardly
    belongs to $K_i$. Hence, the run $s \xrightarrow{\tau_1,e_1} s_1
    \xrightarrow{\tau_2,e_2} s_2 \ldots$ is in $\bigcap_{i=1}^\infty B_i$,
    which completes the proof in this case.
    \begin{figure}[h!]
      \begin{center}
        \begin{tikzpicture}[yscale=.8]
          \draw [-latex'] (0,0) -- (7,0);
          \draw [-latex'] (0,0) -- (0,4);
          
          \draw (1,-.1) -- (1,.1);
          \draw (5,-.1) -- (5,.1);
          
          \draw (-.1,.5) -- (.1,.5);
          \draw (-.1,3) -- (.1,3);
          
          \draw [fill=black!15!white] (1,.5) -- (1,3) -- (5,3) --
          (5,.5) -- (1,.5); 
          
          \draw [fill=black!40!white] (2,.5) -- (4.5,2) -- (3,3) --
          (2,.5);
          
          \draw [dotted] (2,.5) -- (0,.5);
          \draw [dotted] (3,3) -- (0,3);
          
          \path (7.2,.3) node[] {$\IR_+^{n_1}$};
          \path (6.5,3.7) node[] {$\IR_+^{n_2}$};
          \path (-.3,4.2) node[] {$\IR_+^{n_2-n_1}$};
          
          \path (3,-.5) node[] {$K_1$};
          \path (3.25,1.8) node[] {$K_2$};
          \path (-.4,2) node[] {$\widetilde{K}_2$};
             \end{tikzpicture}
      \end{center}
      \caption{The decomposition of the $K_i$'s.}
      \label{fig:decomposition}
    \end{figure}
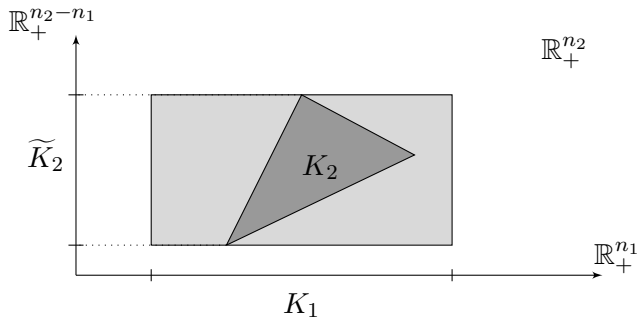
  \item either the sequence $(n_i)_{i \geq 1}$ is upper bounded. In
    that case, we embed $\bigcap_{i=1}^\infty K_i$ into a compact set
    of $\IR_+^{N}$ where $N = \lim_{i \rightarrow +\infty}n_i$. We let
    the details to the reader, as they are very similar to (and easier
    than) the previous case.   \end{itemize}
\end{proof}

\ligne

\noindent\fbox{\begin{minipage}{.98\linewidth}
\basicopensets* \end{minipage}}
\label{app:basicopensets} 

\begin{proof}
  First notice that if $\mathcal{C}$ is an open constraint of
  $\mathbb{R}^n$, then $\mathcal{C}$ is Borel-measurable, and
  $\Pol(\path[\mathcal{C}]{s,e_1 \dots e_n}) = \Pol(\path{s,e_1 \dots
    e_n}) \cap \mathcal{C}$ is open in $\Pol(\path{s,e_1 \dots
    e_n})$. Furthermore, applying Corollary~\ref{coro2-essai3},
  $\path[\mathcal{C}]{s,e_1 \dots e_n}$ is thick if $\path{s,e_1 \dots
    e_n}$ is thick as well.

  Now, assume that $\pi = \path[\mathcal{C}]{s,e_1 \dots e_n}$ is a
  basic open set. This means in particular that $\Pol(\pi)$ is open in
  $\Pol(\path{s,e_1 \dots e_n})$: there exists an open set $\delta$ of
  $\mathbb{R}^n$ such that $\Pol(\pi) = \Pol(\path{s,e_1 \dots e_n})
  \cap \delta$. As it is open, $\delta$ is Borel-measurable, and we
  get that $\Pol(\pi) = \Pol(\path[\delta]{s,e_1 \dots e_n})$, and
  therefore $\pi = \path[\delta]{s,e_1 \dots e_n}$, which is the
  expected result. 
\end{proof}

\ligne

\label{app:homeo}
\noindent \fbox{\begin{minipage}{.98\linewidth}
\homeo* \end{minipage}}

\begin{proof}
  \textit{We first prove that $\iota$ is continuous.}  Let
  $\pi_{\mathcal{C}} = \path[\mathcal{C}]{\iota(s),f_1 \ldots f_n}$ be
  a symbolic path in $\RA$ such that $\Cyl(\pi_{\mathcal{C}})$ is a
  basic open set of
  $(\Runs(\iota(s),\RA),\mathcal{T}^{\iota(s)}_{\RA})$. We need to
  prove that $\iota^{-1}(\Cyl(\pi_{\mathcal C}))$ is an open set of
  $\mathcal{T}_{\A}^s$.  One can easily be convinced that
  $\iota^{-1}(\Cyl(\pi_{\mathcal C})) = \Cyl(\iota^{-1}(\pi_{\mathcal
    C}))$. Thus proving the continuity of $\iota$ consists in proving
  that $\iota^{-1}(\pi_{\mathcal C})$ is a thick finite symbolic path
  whose polyhedron is open in its ambient space.

  First notice that there are unique edges $e_1,\ldots,e_n$ such that
  $\iota^{-1}(\pi_{\mathcal C}) \subseteq \path{s,e_1\ldots e_n}$, we
  can then set $\pi' \egdef \iota^{-1}(\pi_{\mathcal C})$. Then
  obviously, $\Pol(\pi_{\mathcal{C}}) = \Pol(\pi')$.

  Let $\gamma$ be the tightest constraint which defines
  $\Pol(\path{\iota(s),f_1 \ldots f_n})$. We have for every $i \le n$:
  \[
  \left\{\begin{array}{lcl}
      \iota^{-1}(\path[\mathcal{C}_i]{\iota(s),f_1 \ldots f_i}) & = &
      \path[\mathcal{C}_i \wedge \gamma_i]{s,e_1 \ldots e_i} \\
      \bigcup_f \iota^{-1}(\path[\mathcal{C}_i]{\iota(s),f_1 \ldots
        f_{i-1}f}) & = & \bigcup_e \path[\mathcal{C}_{i-1} \wedge
      \gamma_{i-1}]{s,e_1 \ldots e_{i-1}e}
    \end{array}\right.
  \]
  where $\mathcal{C}_i$ and $\gamma_i$ are the projection of
  $\mathcal{C}$ and $\gamma$ on the $i$ first coordinates (in
  particular, $\gamma_i$ is the tightest constraint defining
  $\path{\iota(s),f_1\ldots f_i}$ since this is in $\RA$). As
  $\pi_{\mathcal{C}}$ is thick, the two dimensions on the left are equal,
  and therefore so are the two dimensions on the right. We deduce that
  $\pi'$ is thick.

  Now, as $\path{\iota(s),f_1 \ldots f_n}$ is thick (since
  $\pi_{\mathcal{C}}$ is thick), we can prove by induction on its
  length that there is some open constraint $\delta$ such that
  \[  
  \Pol(\path{\iota(s),f_1 \ldots f_n}) = \Pol(\path[\delta]{s,e_1
    \ldots e_n}) = \Pol(\path{s,e_1 \ldots e_n}) \cap \delta
  \] Indeed:
  \begin{itemize}
  \item if $\delta_i$ is an open constraint such that
    $\Pol(\path[\delta_i]{s,e_1 \ldots e_i}) = \Pol(\path{\iota(s),f_1
      \ldots f_i})$, we have that $\dim(\Pol(\path[\delta_i]{s,e_1
      \ldots e_i e_{i+1}})) = \dim(\Pol(\path{\iota(s),f_1 \ldots f_i
      f_{i+1}}))$ (by thickness of $\path{\iota(s),f_1 \ldots f_n}$);
  \item furthermore this value is either equal to
    $\dim(\Pol(\path{\iota(s),f_1 \ldots f_i}))$ or to\linebreak
    $\dim(\Pol(\path{\iota(s),f_1 \ldots f_i}))+1$;
  \item in the first case, we set $\delta_{i+1}=\delta_i$, whereas in
    the second case $\delta_{i+1}$ is obtained by adding to $\delta_i$
    the open constraint derived from the last transition $f_{i+1}$
    (which is then open);
  \item we can easily check that this concludes the induction.
  \end{itemize}

  As $\pi_{\mathcal{C}}$ is thick, $\Pol(\pi_{\mathcal{C}})$ is open
  in $\Pol(\path{\iota(s),f_1 \ldots f_n})$: there exists an open set
  $O$ of $\mathbb{R}^n$ such that $\Pol(\pi_{\mathcal{C}}) =
  \Pol(\path{\iota(s),f_1 \ldots f_n}) \cap O$. 

  We infer that $\Pol(\pi') = \Pol(\pi_{\mathcal{C}}) =
  \Pol(\path{s,e_1 \ldots e_n}) \cap \delta \cap O$, and $\delta \cap
  O$ is open in $\mathbb{R}^n$: $\Pol(\pi')$ is open in
  $\Pol(\path{s,e_1 \ldots e_n})$.

  This concludes the proof: $\Cyl(\pi')$ is a basic open set in
  $(\Runs(\A,s),\mathcal{T}_{\A}^s)$.

  \bigskip \textit{We now prove that for every non-empty open set $O
    \in \mathcal{T}_{\A}^s$, $\widering{\iota(O)} \neq \emptyset$.}
  Again it is sufficient to prove that for each basic open set
  $\Cyl(\pi_{\mathcal C})$ of $\mathcal{T}^s_{\A}$,
  $\iota(\Cyl(\pi_{\mathcal C}))$ contains a basic open set $\Cyl(\pi')$
  of $\mathcal{T}^{\iota(s)}_{\RA}$, that is, there is a thick
  symbolic path $\pi'$ whose polyhedron is open in its ambient space
  and such that $\Cyl(\pi') \subseteq \iota(\Cyl(\pi_{\mathcal C}))$.

  Let $\pi_{\mathcal{C}} = \path[\mathcal{C}]{s,e_1 \ldots e_n}$ be a
  constrained symbolic path such that $\Cyl(\pi_{\mathcal{C}})$ is a
  basic open set of $(\Runs(s,\A),\mathcal{T}^s_{\A})$. We have that
  \[
  \iota(\pi_{\mathcal{C}}) = \bigcup_{f_1,\ldots, f_n}
  \path[\mathcal{C}]{\iota(s),f_1 \ldots f_n}
  \]
  where the (finite) union is taken over all sequences of edges
  $f_1,\ldots,f_n$ corresponding to $e_1,\ldots,e_n$. There exist thus
  edges $f_1,\ldots,f_n$ such that
  \[
  \dim(\Pol(\pi_{\mathcal{C}})) =
  \dim\Big(\Pol\big(\path[\mathcal{C}]{\iota(s),f_1 \ldots f_n}\big)\Big)
  \]
  and we write $\pi'_{\mathcal{C}} = \path[\mathcal{C}]{\iota(s),f_1
    \ldots f_n}$. We will prove that $\pi'_{\mathcal{C}}$ is an open
  set. Note that as $\Pol(\pi_{\mathcal{C}})$ is open in
  $\Pol(\path{s,e_1 \ldots e_n})$, we can assume w.l.o.g. that
  $\mathcal{C}$ defines an open set of $\IR^n$. Hence
  $\pi'_{\mathcal{C}}$ is open in $\path{\iota(s),f_1 \ldots
    f_n}$. Assume that it is thin. Then, there exists some $i$ such
  that
  \[
  \dim \Big( \Pol\big(\path[\mathcal{C}'_i]{\iota(s),f_1 \ldots
    f_i}\big) \Big) < \dim \Big( \bigcup_f
  \Pol\big(\path[\mathcal{C}'_{i-1}]{\iota(s),f_1 \ldots
    f_{i-1},f}\big) \Big)
  \]
  where $\mathcal{C}'_i$ corresponds to the projection on the $i$
  first coordinates of the tightest constraint defining
  $\pi'_{\mathcal{C}}$.  Moreover, as $\Pol(\pi'_{\mathcal{C}})
  \subseteq \Pol \big(\pi_{\mathcal{C}}\big)$ and
  $\dim(\Pol(\pi'_{\mathcal{C}})) = \dim(\Pol(\pi_{\mathcal{C}}))$,
  applying Lemma~\ref{lemma:croissance}, we get that for all $i$'s,
  $\dim (\Pol(\pi'_{\mathcal{C}'_i})) = \dim
  (\Pol(\pi_{\mathcal{C}_i}))$. Furthermore, $ \bigcup_f
  \Pol\big(\path[\mathcal{C}'_{i-1}]{\iota(s),f_1 \ldots f_{i-1}
    f}\big) \subseteq \bigcup_e
  \Pol\big(\path[\mathcal{C}_{i-1}]{s,e_1 \ldots e_{i-1} e}\big)$
  (this is a property of the region automaton). Finally, we get that
  \[
  \dim \Big( \Pol\big(\pi_{\mathcal{C}_i}\big) \Big) < \dim \Big( \bigcup_e
  \Pol\big(\path[\mathcal{C}_{i-1}]{s,e_1 \ldots e_{i-1} e}\big) \Big)
  \]
  which contradicts the hypothesis that $\pi$ is thick. We deduce that
  $\Cyl(\pi')$ is a basic open set of
  $(\Runs(\RA,\iota(s)),\mathcal{T}^{\iota(s)}_{\RA})$, hence the
  result. 
\end{proof}

\ligne

\noindent\fbox{\begin{minipage}{.98\linewidth}
    \toporegions* \end{minipage}}
\label{app:toporegions}

\begin{proof}
  We prove both implications using characterisation of meagre sets by
  Banach-Mazur games and Lemma~\ref{lemma:homeo}. To play this game,
  we choose as basis all open sets.

  Assume Player~$2$ has a winning strategy in $\A$ to avoid
  $\sem{P}_{\A,s}$.  We will show that Player~$2$ also has a winning
  strategy in $\RA$ to avoid $\sem{P}_{\RA,\iota(s)}$. Before starting
  the simulation, we recall that $\sem{P}_{\RA,\iota(s)} =
  \iota(\sem{P}_{\A,s})$.

  The first move of Player~$1$ in $\RA$ is some open set $B_1$ (in
  $\mathcal{T}_{\RA}^{\iota(s)}$). That move can be transported in
  $\A$: thanks to Lemma~\ref{lemma:homeo}, $B'_1 \egdef
  \iota^{-1}(B_1)$ is a legal move of the game in $\A$. Then,
  Player~$2$ plays according to her strategy in $\A$ with move
  $B'_2$. This move cannot directly be transported to $\RA$ (since
  $\iota(B'_2)$ may not be an open set), but thanks to
  Lemma~\ref{lemma:homeo}, there is a non-empty open set $B_2$ such
  that $B_2 \subseteq \iota(B'_2)$. We continue the simulation that
  way. Finally we get that $\bigcap_{i} B_i \subseteq \bigcap_i
  \iota(B'_i) = \iota(\bigcap_i B'_i)$. As Player~$2$ plays with a
  winning strategy in $\A$, we get that $\bigcap_i B'_i \cap
  \sem{P}_{\A,s} = \emptyset$, which implies that $\bigcap_i B_i \cap
  \sem{P}_{\RA,\iota(s)} = \emptyset$.

  On the contrary, assume that Player~$2$ has a winning strategy in
  $\RA$ to avoid $\sem{P}_{\RA,\iota(s)}$.  We will show that
  Player~$2$ also has a winning strategy in $\A$ to avoid
  $\sem{P}_{\A,s}$.  Assume that Player~$1$ plays $\path[\gamma]{s,e_1
    \ldots e_n}$, then applying Lemma~\ref{lemma:homeo}, Player~$2$
  can play as if it was $\path[\gamma]{\iota(s),f_1 \ldots f_n}$ in
  $\RA$ for some $f_1,\ldots,f_n$. The game then plays as in $\RA$,
  and all moves are legal thanks to Lemma~\ref{lemma:homeo}. We
  conclude that this strategy avoids $\sem{P}_{\A,s}$ as well, which
  concludes the proof. 
\end{proof}

\section{Details for Section~\ref{sec:thickgraph}}
\label{app:thickgraph}

\noindent \fbox{\begin{minipage}{.98\linewidth}
\thinlocal* \end{minipage}}

\proof
  Let $q$ be the target region of $\path{s,e_1 \ldots e_{n-1}}$. There
  are two possible cases:
  \begin{itemize}
  \item $\dim(I(q,e_n)) = 0$: for every $s' \in q$, there is a unique
    delay $\tau_n(s')$ such that $s'
    \xrightarrow{\tau_n(s'),e_n}$. Also, for all delays
    $\tau_1,\ldots,\tau_{n-1}$ such that $s \xrightarrow{\tau_1,e_1}
    \ldots \xrightarrow{\tau_{n-1},e_{n-1}}$ there is a unique $s' \in
    q$ such that $s \xrightarrow{\tau_1,e_1} \ldots
    \xrightarrow{\tau_{n-1},e_{n-1}} s'$. We can therefore define the
    function $g$ with $g(\tau_1,\ldots,\tau_{n-1})=\tau_n(s')$.  We
    then write:
    \begin{eqnarray*}
      \Pol(\path{s,e_1 \ldots e_n}) & = & \{(\tau_1,\ldots,\tau_n) \mid
      \tau_n=\tau_n(s')\ \text{where}\ s \xrightarrow{\tau_1,e_1} \ldots
      \xrightarrow{\tau_{n-1},e_{n-1}} s'\} \\
      & = & \{(\tau_1,\ldots,\tau_n) \mid (\tau_1,\ldots,\tau_{n-1}) \in \Pol(\path{s,e_1 \ldots e_{n-1}}) \\
      & & \hspace*{2cm} \text{and}\ \tau_n=g(\tau_1,\ldots,\tau_{n-1})\} \\
      & = & \{(\tau_1,\ldots,\tau_{n-1},g(\tau_1,\ldots,\tau_{n-1})) \mid (\tau_1,\ldots,\tau_{n-1}) \in \Pol(\path{s,e_1 \ldots e_{n-1}})\}
    \end{eqnarray*}
    The second equality holds because $\path{s,e_1 \ldots e_n}$ is a
    symbolic path in the region automaton $\RA$.  We deduce that
    $\dim(\Pol(\path{s,e_1 \ldots e_n})) = \dim(\Pol(\path{s,e_1
      \ldots e_{n-1}}))$, which is the expected value.
  \item $\dim(I(s',e_n)) = 1$ for every $s' \in q$: there is
    non-punctual open interval $(\mu_n(s'),\nu_n(s'))$ such that $s'
    \xrightarrow{\tau_n,e_n}$ iff $\tau_n \in (\mu_n(s'),\nu_n(s'))$.
    We can then rewrite $\Pol(\path{s,e_1,\ldots,e_n})$ as follows:
    \begin{eqnarray*}
      \Pol(\path{s,e_1 \ldots e_n}) & = & \{(\tau_1,\ldots,\tau_{n-1},\tau_n) \mid (\tau_1,\ldots,\tau_{n-1}) \in \Pol(\path{s,e_1 \ldots e_{n-1}}) \\
      & & \hspace*{2cm} \text{and}\ \tau_n \in g(\tau_1,\ldots,\tau_{n-1})\}
    \end{eqnarray*}
    where $g(\tau_1,\ldots,\tau_{n-1})$ defines an open interval.
    We then get \[
    \dim(\Pol(\path{s,e_1 \ldots e_n})) =
    \dim(\Pol(\path{s,e_1 \ldots e_{n-1}}))+1.\eqno{\qEd}\]
  \end{itemize}

\noindent\fbox{\begin{minipage}{.98\linewidth}
\thicknesslocal* \end{minipage}}
\label{app:thicknesslocal}

\begin{proof}
  The proof is done by induction on the length $n$ of $\pi$.

  The case $n=0$ is obvious since $\path{s}$ is thick and $\path{s}$
  surely contains no thin edge. Assume $n>0$ and the result holds for
  any $0 \le j \leq n-1$, and let $\pi = \path{s,e_1 \ldots e_n}$ be a
  symbolic path of length $n$. For every $1 \le j \le n$, write
  $\pi_{\le j}$ for $\path{s,e_1 \ldots e_j}$.
  \begin{itemize}
  \item Let us first assume that $\pi$ is thin in $\RA$. In case there
    exists $j<n$ such that $\pi_{\le j}$ is thin, then we are done by
    applying the induction hypothesis to $\pi_{\le j}$.  We therefore
    assume that $\pi_{\le j}$ is thick for every $j<n$. Let $k=
    \dim(\Pol(\pi_{\le n-1}))$.  Applying Lemma~\ref{lem:thin-local},
    $\dim(\Pol(\pi)) \in \{k,k+1\}$, and if $\dim(\Pol(\pi))=k+1$,
    then $\pi$ would be thick ($k+1$ is the maximal dimension of any
    possible $\Pol(\path{s,e_1 \ldots e_{n-1}e})$). We thus get that
    $\dim(\Pol(\pi))=k$, and there exists $e$ such that
    $\dim(\Pol(\path{s,e_1 \ldots e_{n-1} e}))=k+1$ (this witnesses
    the fact that $\pi$ is thin). By Lemma~\ref{lem:thin-local}, we
    therefore infer that $\dim(I(q,e_n))=0$ whereas $\dim(I(q,e))=1$
    where $q$ is the target region of $\path{s,e_1 \ldots e_{n-1}}$:
    edge $e_n$ is thin, which concludes the left-to-right implication.
  \item Assume now that $\pi$ is thick. By definition of thickness,
    for every $j \le n$, $\pi_{\le j}$ is thick.  By induction
    hypothesis, all edges $e_1,\ldots,e_{n-1}$ are then thick.  Assume
    towards a contradiction that $e_n$ is thin.  Let $q$ be the
    source-region of $e_n$, there exists an edge $e$ such that for
    every $s' \in q$, $\dim(I(s',e_n)) < \dim(I(s',e))$. Now we know
    that for every $s'$ with $s \xrightarrow{\tau_1,e_1} \ldots
    \xrightarrow{\tau_{n-1},e_{n-1}} s'$ for some
    $\tau_1,\ldots,\tau_{n-1}$, $s' \in q$ (this is a property of
    region automata).  Applying Lemma~\ref{lem:thin-local}, we get
    that \[\dim(\Pol(\path{s,e_1\ldots e_n})) =
    \dim(\Pol(\path{s,e_1\ldots e_{n-1}})) +\dim(I(q,e_n))\] whereas
    \[\dim(\Pol(\path{s,e_1\ldots e_{n-1}e})) =
    \dim(\Pol(\path{s,e_1\ldots e_{n-1}})) +\dim(I(q,e)).\] We deduce
    that $\pi$ is thin, contradicting the assumption. This concludes
    the proof of the right-to-left implication.\qedhere
  \end{itemize}
\end{proof}

\section{Details for Section~\ref{sec:match}}

\subsection{Safety properties}
\label{app:safety}
~\\

\noindent \fbox{\begin{minipage}{.98\linewidth}
\probthick* \end{minipage}}

\begin{proof}
  We assume that the probability distributions in $\RA$ are those used
  in Lemma~\ref{lemma:proba}, and we write $\mu$ for the distributions
  over delays.

  We first prove that $\Prob_{\RA}(\Cyl(\pi)) >0$ implies that $\pi$
  is thick. Towards a contradiction, assume that $\pi$ is
  thin. Following Proposition~\ref{prop:thickness-local}, there exists
  $1 \le i \le n$ such that $e_i$ is thin. Let $q$ be the target set
  of $\path{s,e_1 \ldots e_{i-1}}$: $\dim(I(q,e_i)) < \dim(I(q))$. By
  hypothesis on the measures $\mu$ (condition $(\star)$, \textit{cf}
  page~\pageref{subsec:proba}), for every $s' \in q$,
  $\mu_{s'}(I(s',e_i))=0$. Hence, for every $s' \in q$,
  $\Prob_{\RA}(\Cyl(\path{s',e_i \ldots e_n})) = 0$. This implies that
  $\Prob_{\RA}(\Cyl(\pi)) = 0$.

  Assume now that $\pi$ is thick.  For every $1 \le i \le n$, we write
  $q_i$ for the target region of $\path{s,e_1 \ldots e_i}$. Following
  Proposition~\ref{prop:thickness-local}, for every $1\le i \le n$,
  $e_i$ is thick, which means that $\dim(I(q_i,e_{i+1})) =
  \dim(I(q_i))$. As in the first implication, by assumption $(\star)$
  on the measure $\mu$, for every $s_i \in q_i$,
  $\mu_{s_i}(I(s_i,e_{i+1}))>0$.  We then use the definition of the
  probability inductively on suffixes of $\pi$ starting in some $s_i$
  to obtain a sequence of integral computation over non negligible set
  of a positive function, hence $\Prob_{\RA}(\Cyl(\pi)) >0$.  
\end{proof}

\ligne

\noindent\fbox{\begin{minipage}{.98\linewidth}
\theosafety* \end{minipage}}

\label{app:safety2}
\begin{proof}
  Equivalence between $(c)$ and $(d)$ is by construction of
  $\thickgraph(\A)$.

  Thanks to Lemmas~\ref{lemma:proba} and \ref{prop-topo-A-RA}, to
  prove the equivalence between $(a)$ and $(b)$, it is equivalent (and
  hence sufficient) to prove that $\RA,\iota(s) \robust_{{\mathcal T}} P$
  iff $\RA,\iota(s) \robust_{\Prob} P$. Since $P$ is a property over
  \AP, for every infinite symbolic path $\pi = \path{s,e_1 \ldots e_i
    \ldots}$, either all realisations $\rho \in \pi$ satisfy $P$, or
  none of them satisfies $P$. Now, using the fact that $P$ is a safety
  property, we can write $\sem{\neg P}_{\RA,\iota(s)}$ as a
  denumerable union of cylinders:
  \[
  \sem{\neg P}_{\RA,\iota(s)} = \bigcup_{i \in I} \Cyl(\pi_i)
  \]
  where $\pi_i$ is a finite (unconstrained) symbolic path from
  $\iota(s)$, and $I$ is a denumerable set.  Thus, $\RA,\iota(s)
  \robust_{\Prob} P$ is equivalent to $\Prob_{\RA}(\bigcup_{i \in I}
  \Cyl(\pi_i))=0$. Since $I$ is denumerable, we obtain $\RA,\iota(s)
  \robust_{\Prob} P$ iff for all $i \in I$,
  $\Prob_{\RA}(\Cyl(\pi_i))=0$. By Proposition~\ref{prop:prob-thick},
  we have that $\RA,\iota(s) \robust_{{\Prob}} P$ iff for all $i \in
  I$, $\pi_i$ is thin (this by-the-way proves equivalence between
  $(a)$ and $(c)$).  Thus proving the theorem amounts to proving that
  $\sem{\neg P}_{\RA,\iota(s)}$ is meagre iff for all $i \in I$,
  $\pi_i$ is thin.

  Let us assume first that for all $i \in I$, $\pi_i$ is thin. To
  prove that , $\sem{\neg P}_{\RA,\iota(s)}$ is meagre we use a
  Banach-Mazur game and Theorem~\ref{BM:thm}, playing with the set
  ${\mathcal B}$ of basic open sets of $\mathcal{T}_{\RA}^{\iota(s)}$. The
  objective of the game is set to be $\sem{\neg P}_{\RA,\iota(s)}$.
  By hypothesis, $\sem{\neg P}_{\RA,\iota(s)} = \bigcup_{i \in I}
  \Cyl(\pi_i)$ and all $\pi_i$ are thin. As every basic open set is of
  the form $\Cyl(\pi)$ with $\pi$ thick, it holds that for every $B
  \in {\mathcal B}$ such that $B \neq \Runs(\RA,\iota(s))$, we have $B
  \cap \sem{\neg P}_{\RA,\iota(s)} = \emptyset$.  Thus, if the first
  move of Player~1 is $\Runs(\RA,\iota(s))$, Player~2 picks some
  cylinder of a finite thick path. If the first move of Player~1 is a
  cylinder $\Cyl(\pi)$, then Player~2 just chooses the same set. Then,
  Player~2 wins the game by mimicking at each round the choices of
  Player~1, \textit{i.e.}, whatever set $B_{2j-1}$ Player~1 chooses in
  the $j$-th round, Player~2 answers with the same choice $B_{2j} =
  B_{2j-1}$.  For such a play we clearly have $\bigcap_{i=1}^\infty
  B_i \cap \sem{\neg P}_{\RA,\iota(s)} = \emptyset$.  Thus Player~2
  has a winning strategy and Theorem~\ref{BM:thm} implies that
  $\sem{\neg P}_{\RA,\iota(s)}$ is meagre.

  Let us now prove the other implication. For a contradiction we
  assume that $\pi_i$ is thick for some $i \in I$. In particular
  $\sem{\neg P}_{\RA,\iota(s)}$ would contain the open set
  $\Cyl(\pi_i)$, which is not meagre because our topological space is
  Baire (see Proposition~\ref{prop:Baire}).  Since the notion of being
  meagre is closed under subset, the set $\sem{\neg P}_{\RA,\iota(s)}$
  would not be meagre, which is a contradiction. 
\end{proof}

\subsection{Details for the counter-example
  \texorpdfstring{$\A_{\mathsf{unfair}}$}{A(unfair)} of
  Figure~\ref{fig:pacman}}

\label{app:pacman}

\begin{prop}
  $\Prob_{\A_{\mathsf{unfair}}} \big(\path{s_0,(e_3 e_4 e_5)^\omega}
  \big) > 0$, where $s_0 = (\ell_0,(0,0))$.
\end{prop}
\begin{proof}
  In this proof, we write $\Prob$ for $\Prob_{\A_{\mathsf{unfair}}}$,
  and we first show that given $0 < t_0 < 1$, $\Prob
  \big(\path{s_{t_0},(e_3 e_4 e_5)^\omega} \big) > 0$ where $s_{t_0} =
  (\ell_0,(0,t_0))$.  Obviously, we have that 
  \[\Prob \big(\path{s_{t_0},(e_3 e_4
    e_5)^\omega}\big) = \lim_{N\rightarrow +\infty} \Prob
  \big(\path{s_{t_0},(e_3 e_4 e_5)^N} \big)\] where we write $\Prob
  \big(\path{s_{t_0},(e_3 e_4 e_5)^N} \big)$ for
  $\Prob\big(\Cyl(\path{s_{t_0},(e_3 e_4 e_5)^N} )\big)$.

  We now would like to express $\Prob(\path{s_{t_0},(e_3 e_4 e_5)^N})$
  as a multiple integral. In order to take the leftmost loop, we need
  to choose a first delay ensuring that the valuation of the clock $y$
  satisfies the guard $1<y<2$. The location $\ell_4$ is then reached
  with the clock valuation $(2-t_0,0)$. From there a second positive
  time delay has to be chosen in order to reach location $\ell_0$. We
  thus have that:
  \begin{align*}
    \Prob\bigl(\path{s_{t_0},(e_3 e_4 e_5)^N}\bigr) & = \frac{1}{2-t_0}
    \int_{\tau=1-t_0}^{2-t_0} \frac{1}{1-t_0} \int_{t_1=t_0}^1
    \Prob\big(\path{s_1,(e_3 e_4 e_5)^{N-1}}\big)  \ud t_1  \ud \tau \\
    & = \frac{1}{2-t_0} \cdot \frac{1}{1-t_0} \int_{t_1=t_0}^1
    \Prob\big(\path{s_1,(e_3 e_4 e_5)^{N-1}}\big) \ud t_1
  \end{align*}
  where $s_1 = (\ell_0,(0,t_1))$. By iterating this process, we obtain that:
  \begin{multline*}
  \Prob\big(\path{s_{t_0},(e_3 e_4 e_5)^N}\big) =\\ \frac{1}{2-t_0}.\frac{1}{1-t_0}
  \int_{t_1=t_0}^{1} \frac{1}{2-t_1} \cdot \frac{1}{1-t_1} \int_{t_2=t_1}^{1}
  \ldots \frac{1}{2-t_{N-1}} \cdot \frac{1}{1-t_{N-1}} \int_{t_N=t_{N-1}}^1
  \ud t_N \ldots \ud t_1.
  \end{multline*}
  We write
  \[
  \gamma_i^N = \frac{1}{1-t_{i-1}} \int_{t_i = t_{i-1}}^1 \frac{1}{2-t_i} \cdot
  \frac{1}{1-t_i} \int_{t_{i+1}=t_i}^1 \ldots \frac{1}{2-t_{N-1}} \cdot
  \frac{1}{1-t_{N-1}} \int_{t_N=t_{N-1}}^1 \ud t_N \ldots \ud t_i
  \]
  and we can prove by a descending induction on $i$ (see
  Lemma~\ref{lm:rec-annexe} below) that
 \[
  \gamma_i^N\, \geq\, t_{i-1}.
  \]
  Thus, we deduce that
  \begin{equation*}
    \Prob\big(\path{s_{t_0},(e_3 e_4 e_5)^N}\big) = \frac{1}{2-t_0} \cdot \gamma_1^N \geq \frac{t_0}{2-t_0} > 0.
  \end{equation*}
  It remains to show that $\Prob\big(\path{s_0,(e_3 e_4 e_5)^\omega}
  \big)>0$, \emph{i.e.} that the above result extends to the case $t_0
  =0 $. Roughly speaking, after one loop, we will have $t_1 > 0$,
  hence we can use the above inequality from the second loop on by
  writing:
  \begin{eqnarray*}
    \Prob(\path{s_0,(e_3 e_4 e_5)^\omega}) & = & \frac{1}{2-0} \cdot \frac{1}{1-0}. \int_{0}^1
    \Prob(\path{s_{t_1},(e_3 e_4 e_5)^\omega})\, \ud t_1 \\
    & \geq & \frac{1}{2}. \int_0^1 \frac{t_1}{2-t_1} \, \ud t_1 \\
    & \geq & \frac{1}{2}. \int_0^1 \left(-1 + \frac{2}{2-t_1}\right) \, \ud
    t_1 \\
    & \geq & \frac{1}{2}. \left[-t_1 -2 \log(2-t_1) \right]_{t_1=0}^1 \\
    & \geq & \log(2) - \frac{1}{2}>0.
  \end{eqnarray*}
 This concludes the proof.
\end{proof}

\begin{lem}
\label{lm:rec-annexe}
If $\gamma_i^N = \frac{1}{1-t_{i-1}} \int_{t_i = t_{i-1}}^1 \frac{1}{2-t_i} \cdot
  \frac{1}{1-t_i} \int_{t_{i+1}=t_i}^1 \ldots \frac{1}{2-t_{N-1}} \cdot
  \frac{1}{1-t_{N-1}} \int_{t_N=t_{N-1}}^1 \ud t_N \ldots \ud t_i$
  with $0<t_{i-1}<1$, then:
\[
  \gamma_i^N\, \geq\, t_{i-1}.
  \]
\end{lem}

\begin{proof}
 The base case is when $i=N$. In that case, 
  \[
  \gamma_N^N = \frac{1}{1-t_{N-1}} \int_{t_N = t_{N-1}}^1 \ud t_N = 1
  \]
  which proves the desired property.

  We assume we have proved the property for $i+1$, and want to prove it for
  $i$.
  \begin{eqnarray*}
    \gamma_i^N & = & \frac{1}{1-t_{i-1}} \int_{t_i = t_{i-1}}^1
    \frac{1}{2-t_i} \cdot \gamma_{i+1}^N \, \ud t_i \\
    & \geq & \frac{1}{1-t_{i-1}} \int_{t_i = t_{i-1}}^1
    \frac{t_i}{2-t_i} \ud t_i \qquad \text{(by i.h.)} \\
    &\geq &  \frac{1}{1-t_{i-1}} \int_{t_i = t_{i-1}}^1
    \left(-1+\frac{2}{2-t_i}\right) \ud t_i\\
    & \geq & \frac{1}{1-t_{i-1}} \left[ - t_i-2\log(2-t_i)\right]_{t_i=t_{i-1}}^1 \\
    & \geq & \frac{1}{1-t_{i-1}} \left( -1+t_{i-1}+2\log(2-t_{i-1})\right)
  \end{eqnarray*}
  Now, when $0 \leq x \leq1$, we know that $\log(1+x) \geq x - \frac{x^2}{2}$ (see
  Lemma~\ref{lemma:log}). Applying this inequality to $x = 1-t_{i-1}$,
  we get the following inequality:
  \begin{eqnarray*}
    \gamma_i^N & \geq & \frac{1}{1-t_{i-1}} \left( -1+t_{i-1}+2(1-t_{i-1})-(1-t_{i-1})^2\right)\\
    & \geq & 1-(1-t_{i-1})=t_{i-1}.
  \end{eqnarray*}
  This concludes the inductive case.
\end{proof}

\begin{lem}
  \label{lemma:log}
  Let $0 \leq x \leq 1$. Then $\log(1+x) \geq x - \frac{x^2}{2}$.
\end{lem}

\begin{proof}
Let $f(x)=\log(1+x)$ and $g(x)=x-\frac{x^2}{2}$. Since $f'(x)=\frac{1}{1+x}$ and $g'(x)=1-x$, we remark that for any $0\le x\le 1$, we have $f'(x)\ge g'(x)$. Indeed, for any $x\ge 0$, $\frac{1}{1+x}\ge1-x$ if and only if $1\ge 1-x^2$. Since $f(0)=0=g(0)$, the result follows.
\end{proof}

\subsection{Extension to specification timed automata}
\label{app:specauto}

In this part, we aim at proving Theorem~\ref{theo:specauto} below.
\medskip

\noindent \fbox{\begin{minipage}{.98\linewidth}
\specauto* \end{minipage}}

\medskip The proof of this theorem will require several lemmas that we
present below.  Before that, we define $\iota_{\mathcal{B}}$ the
application (bijection) which assigns to every run $\rho$ in $\A$ its
unique image $\rho^{\mathcal{B}}$ in $\A \ltimes \B$.

\begin{lem}
  \label{lemma:ltl2muller-proba}
    Let $s$ be a state of $\A$, and $\B$ be a specification B\"uchi or
  Muller timed automaton. Assume measures and weights in $\A \ltimes
  \B$ are properly set.
    Then:
  \[
  \Prob_\A(s \models \B) = \Prob_{\A \ltimes \B}(\mathsf{init}_{\A
    \ltimes \B}(s) \models P_{\A \ltimes \B})\,.
  \]
\end{lem}

\begin{proof}
  This lemma is a consequence of the two following properties:
  \begin{enumerate}
  \item for every measurable set $E$ of $\Runs(\A,s)$,
    \[
    \Prob_{\A}(E) = \Prob_{\A \ltimes \B}\{\iota_\B(\varrho) \mid
    \varrho \in E\}
    \]
  \item for every run $\varrho$ in $\A$, $\varrho \models \B$ iff
    $\iota_\B(\varrho) \models P_{\A \ltimes \B}$.
  \end{enumerate}
  The second item is a direct consequence of the definition of $P_{\A
    \ltimes \B}$. The proof for the first item is similar to that for
  region automata, that is, to the proof of
  Lemma~\ref{lemma:proba}. We therefore skip it. 
\end{proof}

A similar result holds for the topological semantics:

\begin{lem}
  \label{lemma:ltl2muller-topo}
  Let $\A$ be a timed automaton, let $s$ be a state of $\A$, and let
  $\B$ be a specification B\"uchi or Muller timed automaton. Then, for
  every set $S \subseteq \Runs(\A,s)$:
  \[
  S\ \text{is large in}\ \mathcal{T}_{\A}^s\ \Leftrightarrow\
  \{\rho^{\B} \mid \rho \in S\}\ \text{is large in} \
  \mathcal{T}_{\A \ltimes \B}^{\textsf{init}_{\A \ltimes \B}(s)}
  \]
\end{lem}

\begin{proof}
  The proof of this lemma is similar to that of
  Proposition~\ref{prop-topo-A-RA}, even though projection
  $\iota_{\B}$ might not be continuous. The next lemma
  (Lemma~\ref{lemma:iotaB}) is sufficient to make the simulation
  between the two Banach-Mazur games (as in the proof of
  Proposition~\ref{prop-topo-A-RA}), and to prove the expected
  result. 
\end{proof}

\begin{lem}
  \label{lemma:iotaB}\hfill
  \begin{enumerate}
  \item If $O$ is a non-empty open set of $\Runs(\A \ltimes
    \B,\textsf{init}_{\A \ltimes \B}(s))$, then $\iota_{\B}^{-1}(O)$
    has a non-empty interior.
  \item If $O$ is a non-empty open set of $\Runs(\A,s)$, then
    $\iota_{\B}(O)$ has a non-empty interior.
  \end{enumerate}
\end{lem}

\begin{proof}
  We show the first property. Let $\pi$ be a basic open set of
  $\Runs(\A \ltimes \B,\textsf{init}_{\A \ltimes \B}(s))$. Following
  the notations of Section~\ref{sec:richer} and applying
  Lemma~\ref{lemma:basicopensets}, we can write the basic open set
  $\pi$ as $\path[\mathcal{C}]{\textsf{init}_{\A \ltimes \B}(s),
    e^1_{\mathsf{e}^1} \dots e^n_{\mathsf{e}^n}}$, where $\mathcal{C}$
  is an open constraint of $\mathbb{R}^n$. Note that as $\B$ is
  complete, none of the transitions $\mathsf{e}^i$ have equality
  constraints (otherwise $\pi$ would not be an open set). Also notice
  that all $e^i$'s are also thick (for the same reason). Let $\gamma$
  be the constraint of $\mathbb{R}^n$ generated by the transitions
  $\mathsf{e}^1,\dots,\mathsf{e}^n$. Due to the previous remark, the
  interior of $\gamma$ is non-empty. Also, as $\mathcal{C} \wedge
  \gamma$ is non-empty, this implies that $\mathcal{C} \wedge \gamma$
  has a non-empty interior.  Now, we should just notice that
  $\iota_{\B}^{-1}(\pi) = \path[\mathcal{C} \wedge \gamma]{s,e^1 \dots
    e^n}$, which contains an open set, since $\mathcal{C} \wedge
  \gamma$ has a non-empty interior, and all the $e^i$'s are thick.

  Fix now a non-empty open set $O$ in $\Runs(\A,s)$, and pick
  $\pi=\path[\mathcal{C}]{s,e^1 \dots e^n}$ a basic open set included
  in $O$. The image of $\pi$ by $\iota_{\B}$ can be decomposed as the
  following finite union over all sequences of edges in $\B$ of length
  $n$:
  \[
  \iota_{\B}(\pi) = \bigcup_{(\mathsf{e}^1,\dots,\mathsf{e}^n)}
  \path[\mathcal{C}]{\textsf{init}_{\A \ltimes \B}(s),
    e^1_{\mathsf{e}^1} \dots e^n_{\mathsf{e}^n}}
  \]
  Also note that we don't lose any behaviour, in the sense that
  \begin{eqnarray*}
    \Pol(\pi) &= &\Pol\left(\bigcup_{(\mathsf{e}^1,\dots,\mathsf{e}^n)}
      \path[\mathcal{C}]{\textsf{init}_{\A \ltimes \B}(s),
        e^1_{\mathsf{e}^1} \dots e^n_{\mathsf{e}^n}}\right) \\ &=&
    \bigcup_{(\mathsf{e}^1,\dots,\mathsf{e}^n)}
    \Pol\left(\path[\mathcal{C}]{\textsf{init}_{\A \ltimes \B}(s),
        e^1_{\mathsf{e}^1} \dots e^n_{\mathsf{e}^n}}\right)
  \end{eqnarray*}
  In particular, there is some $(\mathsf{e}^1,\dots,\mathsf{e}^n)$
  such that $\dim(\Pol\left(\path[\mathcal{C}_i]{\textsf{init}_{\A
        \ltimes \B}(s), e^1_{\mathsf{e}^1} \dots
      e^i_{\mathsf{e}^i}}\right)) = \dim(\Pol(\pi_{|\le i}))$ for
  every $i$, where we consider here the projections over the $i$ first
  components. In particular, we can easily prove now that
  $\path[\mathcal{C}]{\textsf{init}_{\A \ltimes \B}(s),
    e^1_{\mathsf{e}^1} \dots e^n_{\mathsf{e}^n}}$ is a basic open set
  in $\Runs(\A \ltimes \B,\textsf{init}_{\A \ltimes \B}(s))$. 
\end{proof}

\noindent\emph{Proof of Theorem~\ref{theo:specauto}.}
  We can now prove Theorem~\ref{theo:specauto}. The property $P_{\A
    \ltimes \B}$ is prefix-independent. We can therefore apply
  Corollary~\ref{coro:1} to $\A \ltimes \B$. We can now combine with
  the previous technical lemmas, and we get the expected equivalence.\qed

\section{Details for single-clock timed automata}
\label{app:single}

\noindent \fbox{\begin{minipage}{.98\linewidth}
\technica* \end{minipage}}

\begin{proof}
  We write $\Prob_{\A}(s \xrightarrow{e} (q',J'))$ for the probability
  of the set of runs starting from $s$ with a move $s
  \xrightarrow{\tau,e} s'$ with $s' \in (q',J')$ and for some $\tau
  \in \IR_+$.\footnote{Note that this set is $\Prob_{\A}$-measurable
    because it can be seen as $\Cyl(\path[\mathcal{C}_{J'}]{s,e})$ for
    some constraint $\mathcal{C}_{J'}$ enforcing the first move to
    lead to $J'$.}

  Let $\lambda \egdef \inf_{s \in (q,J)} \Prob_{\A}(s \xrightarrow{e}
  (q',J'))$. Since $\overline{J} \subseteq q$ is compact and $\forall
  s \in q,\ \Prob_{\A}(s \xrightarrow{e} (q',J')) > 0$ (because $e$ is
  thick and $e(s) \cap J'$ is non-empty and open), $\lambda > 0$.
  Indeed we have supposed that for all $\ell \in L$, for all $[a,b]
  \subseteq \IR_+$, the function $v \mapsto \mu_{(\ell,v)}([a,b])$ is
  continuous, see hypothesis $(H3)$ in~$(\dag)$
  (page~\pageref{dagger}), hence $s \mapsto \Prob_{\A}(s
  \xrightarrow{e} (q',J'))$ is continuous.
  
  Denote $E_k$ the set of paths in $\A$ that visit $(q,J)$ infinitely often,
  but from the $k$-th passage in $(q,J)$ on never fire $(q,J) \xrightarrow{e}
  (q',J')$ anymore. Note that the set $E_k$ is $\Prob_{\A}$-measurable, and
  that $\Prob_\A(E_k) \leq \prod_k^\infty (1 - \lambda) = 0$. Then note that
  the set $\bigcup_{k \geq 1} E_k$ can be equivalently defined by $B \wedge
  \neg A$ where $B$ is `$\G \F (q,J)$' and $A$ is `$\G \F
  (q,J) \xrightarrow{e} (q',J')$'. Hence, we get that $\Prob_{\A}(s,B \wedge
  \neg A) \leq \lim_{k \rightarrow +\infty} \Prob_{\A}(E_k) = 0$, and thus
  \[\begin{array}{lcl}
    \Prob_{\A}(s,A \mid B) & = & \frac{\Prob_{\A}(s,A \wedge
      B)}{\Prob_{\A}(s,B)}  \quad \text{(by definition)}\\[.1cm]
    & = & \frac{\Prob_{\A}(s,A \wedge B)}{\Prob_{\A}(s,A \wedge B) +
      \Prob_{\A}(s, \neg A \wedge B)} \\
    & & \hspace*{1cm} \text{(by Bayes formulas)} \\[.1cm]
    & = & 1 \quad \text{(because}\ \Prob_{\A}(s,B \wedge \neg A) = 0\text{)}
  \end{array}\]
  which is exactly $\Prob_{\A}( s, \G \F (q,J) \xrightarrow{e} (q',J') \mid
  \G \F (q,J)) = 1$.  
\end{proof}

\ligne

\noindent \fbox{\begin{minipage}{.98\linewidth}
\oneclock* \end{minipage}}

\begin{proof}
  Let $s$ be a state in $\A$. We want to prove that $\Prob_{\A}(s
  \models \mathsf{fair})=1$. We will equivalently prove
  $\Prob_{\RA}(\iota(s) \models \mathsf{fair})=1$.  To that purpose,
  we decompose the set of infinite runs in $\RA$ from $\iota(s)$ into:
  \begin{enumerate}[label=$(F_{\arabic*})$]
  \item[$(F_1)$] the set of runs with infinitely many resets,
  \item[$(F_2)$] the set of runs with finitely many resets, and which
    are ultimately in the unbounded region $(c_k,+\infty)$,
  \item[$(F_3)$] the set of runs with finitely many resets, and which
    ultimately stay forever in a bounded region, either $\{c_i\}$ with
    $0 \leq i \leq k$, or $(c_i,c_{i+1})$ with $0 \leq i < k$. We
    write $(F_3^{(c_i,c_{i+1})})$ (resp. $(F_3^{c_i})$) for condition
    $F_3$ restricted to $(c_i,c_{i+1})$ (resp.  $\{c_i\}$).
  \end{enumerate}
  
  \noindent We write $\Prob_{\RA}(s,F_j)$ for the probability of the runs
  starting in $s$ and satisfying condition~$F_j$. The three sets of
  runs above are measurable and partition the set of all runs. Hence
  $\sum_{j=1,2,3} \Prob_{\RA}(s,F_j) = 1$, and applying Bayes formula:
  \begin{equation}
    \label{app-bayes}\tag{$\bullet$}
    \Prob_{\RA}(s \models \mathsf{fair}) = \sum_{j=1,2,3} \Prob_{\RA}(s
    \models \mathsf{fair} \mid F_j) \cdot \Prob_{\RA}(s,F_j)\,.
  \end{equation}
  We now distinguish between the three cases to prove that
  $\Prob_{\RA}(s \models \mathsf{fair} \mid F_j)=1$ (in case
  $\Prob_{\RA}(s,F_j)=0$ we remove the corresponding term
  from~\eqref{app-bayes}).

  \begin{description}
  \item[Case $F_1$] 
    We consider the set of runs with infinitely many resets.  Let $\rho
    = s_0 \xrightarrow{\tau_1,e_1} s_1 \xrightarrow{\tau_2,e_2}
    \ldots$ be such a run. There exists $q$ such that for infinitely
    many $i$ with $i \in \IN$, $s_i = (q,0)$ (since $\A$ is
    single-clock). Now, fix a state $(q,0)$ and assume that
    $\Prob_{\RA}(s,\G \F (q,0)) > 0$ (otherwise the set of runs
    visiting infinitely often $(q,0)$ will be negligible). For every
    sequence $\sigma$ of edges and compact sets (as in the statement
    of Lemma~\ref{lemma:technic2}), we get that
    \[
    \Prob_{\RA}( s, \G \F \sigma \mid \G \F (q,0) ) = 1\,.
    \]
    Hence, for sequences of edges $(e_i)_{1 \leq i \leq p}$ such that
    such a $\sigma$ exists, we get that
    \begin{equation}
      \Prob_{\RA}( s, \G \F (q,0) \xrightarrow{e_1} q_1 \ldots
      \xrightarrow{e_p} q_p \mid \G \F (q,0) ) = 1\,.
      \label{star}\tag{$\star$}
    \end{equation}
    Now notice that such a $\sigma$ always exists whenever these edges
    are thick, hence~\eqref{star} holds for every sequence of
    consecutive thick edges.

    Now, fix a thick edge $e$, and assume that the set of paths
    passing through $(q,0)$ infinitely often and enabling $e$
    infinitely often, has a positive probability. We will then prove
    that
    \[
    \Prob_{\RA}(s, (\G \F e\ \text{enabled}) \Rightarrow (\G \F
    \xrightarrow{e} ) \mid \G \F (q,0)) = 1\,,
    \]
    which will imply that $\Prob_{\RA}(s \models \mathsf{fair} \mid
    F_1) = 1$.
    \begin{itemize}
    \item Assume that $e$ is reachable from $(q,0)$ following thick
      edges, say $(e_i)_{1 \leq i \leq p}$ with $e_p=e$. Then,
      applying~\eqref{star}, we get that $\Prob_{\RA}( s, \G \F
      (q,0) \xrightarrow{e_1} q_1 \ldots \xrightarrow{e_p} q_p \mid
      \G \F (q,0) ) = 1$, hence that $\Prob_{\RA}(s, \G
      \F \xrightarrow{e} \mid \G \F (q,0)) = 1$.
    \item Assume on the contrary that $e$ is not reachable from
      $(q,0)$ following thick edges. If $e$ is not reachable from
      $(q,0)$, then $\Prob_{\RA}(s,\G \F e\ \text{enabled} \mid
      \G \F (q,0)) = 0$. Let $W$ be the set of finite
      sequences of edges $(e_i)_{1 \leq i \leq p}$ leading from
      $(q,0)$ to a state where $e$ is enabled. Then:
      $$\begin{array}{l}
        \Prob_{\RA}(\G \F e\ \text{enabled} \mid \G \F (q,0))
        \\[.1cm]
        \hspace*{1cm} = \Prob_{\RA}(\G \F \bigcup_{w \in W} w \mid \G
        \F (q,0)) \\[.1cm]
        \hspace*{1cm} \leq \Prob_{\RA}(\F \bigcup_{w \in W} w \mid \G
        \F (q,0)) \\[.1cm]
        \hspace*{1cm} = 0 \qquad \text{because one of the edges in}\ w\ 
         \text{is thin}.
      \end{array}$$
    \end{itemize}
    In both cases, we get the expected property.

    \smallskip
  \item[Case $F_2$] 
    We consider the set of runs with finitely many resets and which
    end up in the unbounded region $(c_k,+\infty)$. Let $\rho = s
    \xrightarrow{\tau_1,e_1} s_1 \xrightarrow{\tau_2,e_2} \ldots$ be
    such a run, and assume that from $s_n$ on, all states are in the
    unbounded region. From that state on, all edges which are enabled
    are thick and have guard $x > c_k$. Let $e$ be such an edge, and
    $q$ be the region-state source of $e$: the probability
    $\Prob_{\RA}(\Cyl(\path{s,e}))$ for every $s \in q$ is independent
    of the choice of $s$ (as there is no restriction on delays, it is
    equal to $w_e/(\sum_{e'\ \text{enabled at}\ q} w_{e'})$). Hence,
    ultimately, after having reached the unbounded region (and never
    leave it anymore), it will behave like a finite Markov chain.

    Assume now that a resetting edge $e$ is enabled infinitely often along
    $\rho$. Then, by a similar argument to the one in the proof of
    Lemma~\ref{lemma:technic} with the $E_k$, as the probability distribution
    of taking an edge is lower-bounded (because we are now in a finite Markov
    chain), then any edge will be almost surely taken infinitely often. Hence,
    $$\Prob_{\RA} (s, \G\F \, \text{resetting edge enabled} \mid 
    F_2 ) = 0\, ,$$ and thus
    $$\Prob_{\RA} ( s, \neg (\G\F \, \text{resetting edge enabled}) \mid 
    F_2 ) = 1\,.$$ Once more, due to the distribution over edges (which is a
    finite Markov chain), when there is no more resetting edges, we get
     $$\Prob_{\RA} ( s \models \mathsf{fair} \mid F_2 ) = 1\,.$$
   \item[Case $F_3$] 
     We consider the set of runs with finitely many resets and which
     end up in a bounded region. We assume the region $r \egdef
     (c_i,c_{i+1})$. Let $\rho = s \xrightarrow{e_1} s_1
     \xrightarrow{e_2} \ldots$ be a witness run, and we assume that
     from $s_n$ on, we are in region $r$. If $s_{j_1}$ and $s_{j_2}$
     with $n \leq j_1 < j_2$ correspond to the same location, then the
     clock value of $s_{j_1}$ is less than (or equal to) that of
     $s_{j_2}$. Hence, if a thick edge $e$ and whose guard is included
     in $[c_{i+1},+\infty)$ is enabled in $s_{j_1}$ (and thus also in
     $s_{j_2}$), the probability of taking $e$ from $s_{j_2}$ is
     greater than (or equal to) the probability of taking $e$ from
     $s_{j_1}$ (due to $(H4)$ in~$(\dag)$ (page~\pageref{dagger}) on
     $\mu$'s and to the fact that the discrete probability over edges
     is constant by regions). Hence, there is a positive lower bound
     for the probability of taking $e$, and if $e$ is enabled
     infinitely often, it will be taken infinitely often. Such an
     enabled edge is thus only possible with probability $0$ under the
     assumption made in this case. Hence, with probability $1$, only
     edges with guard $x \in r$ are enabled. For these edges, as
     previously, the system behaves like a finite Markov chain. We
     thus get that
     \[
     \Prob_{\RA} (s \models \mathsf{fair} \mid F_3^{(c_i,c_{i+1})} ) =
     1\,.
     \]

     If we now assume the region $r = \{c_i\}$, the reasoning is very
     similar to the previous one. Given a location $\ell$ along the
     suffix of the path where $x=c_i$ always holds, the edges enabled in
     $(\ell,x=c)$ are equipped with a distribution defining a finite
     Markov chain. Hence any edge enabled infinitely often will be
     taken infinitely often almost surely, which implies that
     \[
     \Prob_{\RA} (s \models \mathsf{fair} \mid F_3^{c_i} ) = 1\,.
     \]
  \end{description}
  Gathering all cases, we get the desired property, \textit{i.e.},
  $\Prob_{\RA} (s \models \mathsf{fair}) = 1$. 
\end{proof}

\subsection{Details for Zenoness in single-clock timed automata}
~\\

\noindent \fbox{\begin{minipage}{.98\linewidth}
\lemmazeno* \end{minipage}}

\label{app:zeno-lemma}

\begin{proof}
  From Lemma~\ref{lemma:proba} (resp. Lemma~\ref{prop-topo-A-RA}), we
  know that $(a)$ (resp. $(b)$) is equivalent to $\RA,\iota(s)
  \robust_{\Prob} \neg \Zeno$ (resp. $\RA,\iota(s) \robust_{{\mathcal T}}
  \neg \Zeno$).

  We first remove syntactically all resets from edges of $\RA$
  labelled by $x=0$ since they are useless. We borrow the notations
  used in the proof of Theorem~\ref{theorem:fair-oneclock}, and
  following that proof, we decompose the set of infinite runs from $s$
  into:
  \begin{enumerate}[label=$(F_{\arabic*})$]
  \item[$(F_1)$] the set of runs with infinitely many resets,
  \item[$(F_2)$] the set of runs with finitely many resets, and which
    are ultimately in the unbounded region $(c_k,+\infty)$,
  \item[$(F_3)$] the set of runs with finitely many resets, and which
    ultimately stay forever in a bounded region, either $\{c_i\}$ with
    $0 \leq i \leq k$, or $(c_i,c_{i+1})$ with $0 \leq i < k$.
  \end{enumerate}
  We then also have:
  \begin{equation}
    \label{sum}
    \Prob_{\RA}(\iota(s) \models \Zeno) = \sum_{i=1,2,3} 
    \Prob_{\RA}(\iota(s) \models \Zeno \mid
    F_i) \cdot \Prob_{\RA}(\iota(s),F_i)
  \end{equation}
  when these conditional probabilities are well-defined (otherwise it is
  correct to remove the term from the sum).

  The proof of Lemma~\ref{lemma:zeno1} is then decomposed into two parts, first we
  prove that the two first terms of the above sum are always equal to
  $0$, and then that we can decide whether the last term is equal to
  $0$.

  \begin{lem}
    $\Prob_{\RA}(\iota(s) \models \Zeno \mid F_1) = 0$ and
    $\Prob_{\RA}(\iota(s) \models \Zeno \mid F_2) = 0$.
  \end{lem}

  \begin{proof}
    We distinguish two cases.
    \begin{description}
    \item[Case $F_1$] We consider the set of runs with infinitely many
      resets.  This set can be decomposed according to the states
      $(q,0)$ (where $q\in Q$ is a region) that are visited infinitely
      often. We show that $\Prob_{\A}(\iota(s) \models \Zeno \mid \G\F
      (q,0)) =0$. In order to prove this, we distinguish the four
      following subcases depending on the set $I((q,0))$: either (i)
      $I((q,0)) \cap [0,1) = \emptyset$, or (ii) $(0,1) \subseteq
      I((q,0))$, or (iii) $\{0\} \subsetneq I((q,0))$, or (iv)
      $\{0\} = I((q,0))$.

      Let us first treat the easy case (i). If $I((q,0)) \cap
      [0,1)=\emptyset$, since the timed automaton is non-blocking,
      this means that each time the automaton arrives in state $(q,0)$
      at least $1$ time unit elapses before the next transition. Hence
      a run visiting infinitely often such state $(q,0)$ is
      necessarily non-Zeno.

      Let us now consider case (ii), \textit{i.e.}, we assume that
      $(0,1) \subseteq I((q,0))$. Since the probability distribution
      over the delays is then equivalent to the Lebesgue measure (see
      hypothesis $(\star)$), the probability of waiting a time delay
      $\tau \le \frac{1}{2}$ in $(q,0)$ is positive and strictly
      smaller than $1$ (we write $\lambda_{(q,0)}$ for this value: $0
      < \lambda_{(q,0)}< 1$). Let $E_k$ be the set of runs starting
      from $\iota(s)$, visiting $(q,0)$ infinitely often, and such
      that from the $k$-th passage on, the time elapsed from state
      $(q,0)$ (before taking an action) is less than $\frac{1}{2}$. We
      have $\Prob_{\RA}(E_k) \le \prod_k^\infty\lambda_{(q,0)}=0$, and
      as a consequence
      \[
      \Prob_{\RA}\big( \iota(s) \models \Zeno \mid \G\F (q,0) \wedge
      (ii)\big) \le \sum_{k=0}^\infty \Prob(E_k)=0\,.
      \]

      In case (iii), we assume that $\{0\} \subsetneq I((q,0))$. If
      $(0,1) \subseteq I((q,0))$, we are done by case (ii). We can
      thus suppose that if $0 \ne \tau \in I((q,0))$, we have that
      $\tau \ge 1$. If $I((q,0))$ reduces to a finite union of points,
      the probability $\lambda_0$ of waiting a delay greater than or
      equal to $1$ is positive and strictly smaller than $1$ (because
      the measure is then equivalent to the uniform measure over those
      points, see hypothesis~$(\star)$). When going infinitely often
      through $(q,0)$, we will thus wait infinitely often a time
      greater than or equal to $1$. If $I((q,0))$ contains an open
      interval, the probability of waiting a delay greater or equal
      than $1$ from $(q,0)$ is $1$ (by hypothesis $(\star)$). From
      this we can easily derive that:
      \[
      \Prob_{\RA}\big( \iota(s) \models \Zeno \mid \G \F(q,0) \wedge
      (iii)\big) =0\,.
      \]

      Let us conclude with case (iv) where $I((q,0)) = \{0\}$. Since no
      positive delay can elapse from $(q,0)$, the probability of taking any
      edge enabled in $(q,0)$ is positive (the distribution over edges indeed
      becomes uniform). Hence, any state $(q_e,0)$ reachable from $(q,0)$
      taking edge $e$, is almost surely infinitely often visited (as soon as
      $(q,0)$ is). From $(q_e,0)$, again two situations are possible: either
      $I((q_e,0)) = \{0\}$ or not. In the first case, note that it is
      necessarily the case that such a chain $(q,0) \xrightarrow{0,e_1}
      (q_1,0) \xrightarrow{0,e_2} (q_2,0) \cdots$ is finite, otherwise the run
      would contain only finitely many resets\footnote{Recall that edges
        labelled with $x=0$ are not labelled with a reset.}. Thus we surely
      reach infinitely often a state $(q',0)$ such that $I((q',0)) \ne \{0\}$
      allowing us to rely on the previous cases to obtain the desired results.

      Gathering the four cases, we conclude that $\Prob_{\A}(\iota(s)
      \models \Zeno \mid \G\F (q,0)) = 0$. Hence
      \[
      \Prob_{\A}(\iota(s) \models \Zeno \mid F_1)=0\,.
      \]  
    \item[Case $F_2$] We consider the set of runs with finitely many
      resets and which end up in the unbounded region. From any state
      in the unbounded region, the set of potential delays is
      necessarily of the form $[0,+\infty)$\footnote{Otherwise the
        clock would be compared to a constant greater than the maximal
        one}. From hypothesis~$(H5)$ in $(\dag)$ on the distributions
      over delays, the probability of waiting a time delay $\tau \le
      \frac{1}{2}$ from $s$, denoted $\lambda_s$, can be bounded by a
      constant: $0 < \lambda_s \le \lambda_0 < 1$. Let $E_k$ denote
      the set of executions which, at the $k$-th step, are in the
      unbounded region without leaving it afterwards, and such that
      all delays afterwards are less than $\frac{1}{2}$. The
      probability of being Zeno when in $E_k$ satisfies: $\Prob(E_k )
      \le \prod_{i > k} \lambda_0 = 0$, from which we derive:
      \[
      \Prob(\iota(s) \models \Zeno \mid F_2 ) \le \sum_{k=0}^\infty
      \Prob(E_k) =0\,.
      \]
    \end{description}
    This concludes the proof of the Lemma~\ref{lemma:zeno1}. 
  \end{proof}

  The case of condition $F_3$ is not similar to the two previous
  cases.  Indeed, it is worth noticing that every execution satisfying
  the condition $F_3$ is Zeno. Hence, if $\Prob_{\RA}(\iota(s) \models
  F_3) \neq 0$ (otherwise the term $\Prob_{\RA}(\iota(s) \models \Zeno
  \mid F_3) \cdot \Prob_{\RA}(\iota(s) \models F_3)$ does not appear
  in the sum~\ref{sum}), then $\Prob_{\RA}(\iota(s) \models \Zeno \mid
  F_3) = 1$. It remains to compute or characterise the value
  $\Prob_{\RA}(\iota(s) \models F_3)$.

  A BSCC $B$ in $\bluegraph(\A)$ is called a \emph{Zeno BSCC} if it is bounded
  and contains no resetting edges. Note that in a Zeno BSCC the value of the
  clock lies in a unique interval $(c,c+1)$ (with $0 \leq c < M$) or $\{c\}$
  (with $0 \leq c \leq M$).

  \begin{lem}
    $\Prob_\RA(\iota(s) \models F_3) = \displaystyle\sum_{B\
      \text{Zeno BSCC of}\ \thickgraph(\A)}\Prob_{\RA}(\iota(s)
    \models \F B)$.
  \end{lem}

  \begin{proof}
    Runs in $\RA$ are almost surely fair (thanks to
    Theorem~\ref{theorem:fair-oneclock}), hence $\Prob_{\RA}(\iota(s)
    \models F_3) = \Prob_{\RA}(\iota(s) \models F_3 \wedge
    \mathsf{fair})$. Now, a fair run in $\RA$ actually ends up in a
    BSCC of $\thickgraph(\A)$.
    It is now sufficient to remark that fair runs in $F_3$ end up in a
    BSCC that is bounded and does not reset the clock.  Indeed, if one
    of these condition does not hold, the run would not be in $F_3$
    (either it would end up in an unbounded region, or have infinitely
    many resets). Conversely, any run ending up in a Zeno BSCC
    is in~$F_3$. Hence, the mentioned equality holds. 
  \end{proof}
  This concludes the proof of the lemma. 
\end{proof}

\ligne

\noindent \fbox{\begin{minipage}{.98\linewidth}
\zeno* \end{minipage}}

\label{app:zeno2}

\begin{proof}
  The equivalence of $(a)$ and $(c)$ is a consequence of
  Lemma~\ref{lemma:zeno1} and of Theorem~\ref{th:safety}:
  $\Prob_{\A}(s \models \Zeno) = \sum_{B\ \text{Zeno BSCC of}\
    \thickgraph(\A)} \Prob_{\RA}(\iota(s) \models \F B)$. Therefore,
  $\A ,s \robust_{\Prob} \neg\Zeno$ iff $\Prob_{\A}(s \models \Zeno) =
  0$, which is then equivalent to ``for every Zeno BSCC $B$ of
  $\thickgraph(\A)$, $\Prob_{\RA}(\iota(s) \models \F B)=0$'', which
  is itself equivalent to ``for every Zeno BSCC $B$ of
  $\thickgraph(\A)$, $\Prob_{\RA}(\iota(s) \models \G \neg B)=1$''; it
  remains to realise that $\G \neg B$ is a simple safety property, and
  to apply Theorem~\ref{th:safety}.

  We now show the equivalence with $(b)$.  We remove syntactically all
  resets from edges of $\RA$ labelled by $x=0$ since they are
  useless. We also borrow the notations used in the proof of
  Theorem~\ref{theorem:fair-oneclock}.
  Assume first that $\Prob_{\A}(s \models \Zeno) = 0$. Then no BSCC of
  $\thickgraph(\A)$ is Zeno. We once more play a Banach-Mazur game
  using the basic open sets. Player~$1$ plays some move $B_1$, and
  player 2 then plays a move $B_2$ leading to a BSCC $B$ of
  $\thickgraph(\A)$. By hypothesis, $B$ is not a Zeno BSCC, hence
  either it is not bounded, or it contains resetting edges.
  \begin{itemize}
  \item We first consider the case where $B$ contains no resetting
    edges. In that case, it means that the clock value when in $B$ is
    always above the maximal constant. Hence, the game can keep going
    on, and each time Player~$2$ chooses a move, she first chooses a
    move which constrains the cylinder saying that the delay has to be
    larger than $1$. This is always possible, due to the form of the
    constraints, which all include $(c_k,+\infty)$. In that case, it
    is not difficult to check that the resulting runs are all
    non-Zeno.
  \item We now consider the case where $B$ has resetting edges. Note
    that the clock can then become larger than $0$. In that case,
    Player~$2$ can always choose a move so that it terminates with a
    resetting edge, but has visited a positive region, and has
    enforced that the value of the clock in that precise region was
    larger than $1/2$. In that case also, all runs resulting from that
    play are non-Zeno.
  \end{itemize}
  Hence, we get that Player~$2$ has a strategy to avoid the set of
  Zeno runs, hence this set is meagre.

  \medskip Conversely assume that the set of Zeno runs is meagre, but
  assume also that $\Prob_{\A}(s \models \Zeno) > 0$. Once more, let's
  play the Banach-Mazur game. Player 2 has a strategy to avoid Zeno
  behaviours.  However, as $\Prob_{\A}(s \models \Zeno) > 0$, Player 1
  can play a first move leading to a Zeno BSCC $B$ of
  $\thickgraph(\A)$. Then $B$ has no resetting edges and lies within
  an interval $(c_i;c_{i+1})$ or $\{c_i\}$. Then whatever move Player
  2 chooses, the resulting runs will all be Zeno, hence contradicting
  the assumption that the set of Zeno runs is meagre. The claim
  follows. 
\end{proof}

\section{Details for Subsection~\ref{subsec:reactive} (reactive timed
  automata)}
\label{app:reactive}

Let $s$ be a state of $\RA$, that we will take as initial.
If $e$ is a thick edge in $T$, and $q\in Q$, we write
$\mathfrak{R}^e(s)$ for the set of runs in $\RA$ that start in $s$ and
take $e$ infinitely often, and $\mathfrak{R}^{q}(s)$ for the set of
runs of $\RA$ that start in $s$ and visit $q$ infinitely often. In
particular, we write $\mathfrak{R}^{\source(e)}(s)$ for the set of
runs that start in $s$ and visit $\source(e)$ infinitely often (hence
along which $e$ is enabled infinitely often).

We fix a thick edge $e$ in $T$, and we let $\mathcal{Q}$ be the set of
pairs $q = (\ell,r)$ where $r$ is memoryless and $\mathcal{Q}'$ the
set of elements $q=(\ell,r)\in \mathcal{Q}$ such
that \[\Prob(\mathfrak{R}^{q}(s))>0 \quad\text{and}\quad
\Prob(\mathfrak{R}_0^{q,e}(s_q))>0\] where $\mathfrak{R}_0^{q,e}(s_q)$
is the set of runs that start from $s_q$ and take $e$ before any other
visit to $q$.
\medskip

\noindent \fbox{\begin{minipage}{.98\linewidth}
\equaun* \end{minipage}}

\proof
  We let $\mathfrak{D}^{>M}_n(s)$ be the set of runs from $s$ that
  delay more than $M$ time units before taking the $n$-th transition
  (i.e.  $\mathfrak{D}^{>M}_n(s) = \{\varrho \in \Runs(\A,s) \mid
  \varrho = s \xrightarrow{\tau_1,e_1} s_1 \xrightarrow{\tau_2,e_2}
  \dots \mid \tau_n > M\}$), and $\mathfrak{K}_n^\ell(s)$
  (resp. $\mathfrak{K}_n^q(s)$ if $q = (\ell,r) \in \mathcal{Q}$) the
  set of runs from $s$ such that the $n$-th configuration is of the
  form $(\ell,v)$ (resp. $(\ell,v)$ with $v \in r$).
    
  If we denote $\mathfrak{D}^{\le M}_{n,N}(s):=\bigcap_{n\le k\le
    N}(\mathfrak{D}^{>M}_k(s))^c$ (where $\mathfrak{A}^c$ is the
  complement of any set $\mathfrak{A}$)
and $S_N$ the set of locations $l\in
  L$ such that
  \[
  \Prob\left(\mathfrak{D}^{\le M}_{n,N}(s)\cap
    \mathfrak{K}_{N}^\ell(s) \right)>0
  \]
  then for any $\ell\in S_{N-1}$, we have
  \begin{align*}
    \Prob\left(\mathfrak{D}^{\le M}_{n,N}(s)\mid\mathfrak{D}^{\le
        M}_{n,N-1}(s)\cap \mathfrak{K}_{N-1}^\ell(s)\right)
    &=\Prob\left(\mathfrak{D}^{\le M}_{N,N}(s)\mid\mathfrak{D}^{\le M}_{n,N-1}(s)\cap \mathfrak{K}_{N-1}^\ell(s)\right)\\
    &= \Prob\left(\mathfrak{D}^{\le M}_{N,N}(s)\mid \mathfrak{K}_{N-1}^\ell(s)\right)\\
    &=\mu_{\ell}([0,M])
  \end{align*}
  and thus
  \begin{align*}
    \Prob\left(\mathfrak{D}^{\le M}_{n,N}(s)\right)
    & =\sum_{\ell\in S_{N-1}} \Prob\left(\mathfrak{D}^{\le M}_{n,N}(s)\mid\mathfrak{D}^{\le M}_{n,N-1}(s)\cap \mathfrak{K}_{N-1}^\ell(s)\right)\\
    &\quad\quad\cdot \Prob\left(\mathfrak{D}^{\le M}_{n,N-1}(s)\cap \mathfrak{K}_{N-1}^\ell(s)\right)\displaybreak[0]\\
    & = \sum_{\ell\in S_{N-1}} \mu_{\ell}([0,M])
    \Prob\left(\mathfrak{D}^{\le M}_{n,N-1}(s)\cap \mathfrak{K}_{N-1}^\ell(s)\right)\displaybreak[0]\\
    & \le \max_{\ell\in L} \mu_{\ell}([0,M])\cdot
    \sum_{\ell\in S_{N-1}}  \Prob\left(\mathfrak{D}^{\le M}_{n,N-1}(s)\cap \mathfrak{K}_{N-1}^\ell(s)\right)\\
    & = \max_{\ell\in L} \mu_{\ell}([0,M])\cdot
    \Prob\left(\mathfrak{D}^{\le M}_{n,N-1}(s)\right)\,.
  \end{align*}
  As for any $\ell\in L$, we have assumed $\mu_{\ell}([0,M])<1$, we
  conclude that
  \[
  \Prob\left(\mathfrak{D}^{\le M}_{n,N}(s)\right)\le (\max_{\ell\in L}
  \mu_{\ell}([0,M]))^{N+1-n}\underset{N\rightarrow
    \infty}{\longrightarrow} 0
  \]
  and thus that
  \[
  \Prob\left(\bigcup_{n\in \mathbb{N}}\bigcap_{k\ge
      n}(\mathfrak{D}^{>M}_k(s))^c\right)=0\,.
  \]

  We get that 
  \[
  \Prob\left(\bigcap_{n \in \IN} \bigcup_{k \geq n}
    \mathfrak{D}^{>M}_k(s)\right) = 1.
  \]
  Now it is just a matter of noticing that if $\varrho = s
  \xrightarrow{t_1,a_1} s_1 \cdots \xrightarrow{t_n,a_n} s_n \cdots$
  is in $\mathfrak{D}^{>M}_n(s)$, then $s_n$ is of the form
  $(\ell_n,v_n)$ with $(\ell_n,[v_n]_{\A}) \in \mathcal{Q}$ (a clock
  is either reset on the $n$-transition (hence its value is $0$), or
  it is above $M$) and thus $\varrho \in
  \mathfrak{K}_n^{(\ell_n,[v_n]_{\A})}(s)$. Hence, we have
  $\mathfrak{D}^{>M}_k(s)\subset \bigcup_{q \in
    \mathcal{Q}}\mathfrak{K}^{q}_k(s)$ and
  \[
  \Prob\left(\bigcap_{n \in \IN} \bigcup_{k \geq n} \left(\bigcup_{q
        \in \mathcal{Q}}\mathfrak{K}^{q}_k(s)\right)\right) = 1\,.
  \]
  As $\mathcal{Q}$ is a finite set, we deduce
  \[
  \Prob\left(\bigcup_{q \in \mathcal{Q}} \left( \bigcap_{n \in \IN}
      \bigcup_{k \geq n} \mathfrak{K}^{q}_k(s)\right)\right) = 1
  \]
  and as $\bigcap_{n \in \IN} \bigcup_{k \geq n}
  \mathfrak{K}^{q}_k(s)= \mathfrak{R}^{q}(s)$, we get
 \[
 \Prob\left(\bigcup_{q \in \mathcal{Q}} \mathfrak{R}^{q}(s)
 \right) = 1\,.\eqno{\qEd}
 \]

\ligne 

Now we give the proof of Lemma~\ref{lemma:eq2}. It will require
quite long developments that we give in details.
\medskip

\noindent \fbox{\begin{minipage}{.98\linewidth}
\equadeux* \end{minipage}}

\begin{proof}
Let $q \in \mathcal{Q}'$. We want prove that
\begin{equation*}
  \Prob\left(\mathfrak{R}^e(s) \mid \mathfrak{R}^{q}(s)\right) = 1
\end{equation*}
or equivalently that
\begin{equation*}
  \Prob\left(\mathfrak{R}^e(s)\cap\mathfrak{R}^{q}(s) \mid \mathfrak{R}^{q}(s)\right) = 1.
\end{equation*}

We notice that the event $\mathfrak{R}^e(s) \cap \mathfrak{R}^{q}(s)$ coincides with
\[
\bigcap_{n \in \IN} \bigcup_{k \geq n} \mathfrak{R}_k^{q,e}(s)
\]
where $\mathfrak{R}_k^{q,e}(s)$ is the set of runs starting in $s$
along which an occurrence of edge $e$ is preceded by precisely $k$
visits to $q$, i.e.  $\mathfrak{R}_k^{q,e}(s) = \{\varrho \in
\Runs(\A,s) \mid \varrho = s \xrightarrow{\tau_1,e_1} s_1 \dots
\xrightarrow{\tau_m,e_m} s_m \dots \text{and there exists}\ j\
\text{s.t.}\ e_j = e\ \text{and}\ \#\{1\leq i < j \mid
\mathsf{loc}(s_i)=q\}=k\}$, where $\mathsf{loc}(s_i)$ is the location
of state $s_i$.  We recall the following lemma, which is well-known in
probability theory (see for example~\cite{Bi95}):
\begin{lem}[Borel-Cantelli]
  Assume $(\mathcal{E},\Prob)$ is a probabilistic space, and that the
  measurable events $(E_k)_{k \in \IN}$ are independent. If
  $\displaystyle \sum_{k \in \IN} \Prob(E_k) = +\infty$, then
 \[
 \Prob\left(\bigcap_{n \in \IN} \bigcup_{k \geq n} E_k\right) = 1\,.
 \]
\end{lem}

With the aim to apply this lemma, we will prove that the events
$\mathfrak{R}_k^{q,e}(s)$ are independent in the
$\mathfrak{R}^{q}(s)$-conditional $\sigma$-algebra, and that $\sum_{k
  \in \IN} \Prob(\mathfrak{R}_k^{q,e}(s)\mid\mathfrak{R}^{q}(s)) =
+\infty$, which will imply Lemma~\ref{lemma:eq2}. This is non-trivial
and will require several technical lemmas that we present now. The
following arguments rely on result that will be given as
Corollary~\ref{cor:decomposition} (which is technical, and therefore
postponed).

\paragraph{Independence of events.}

\begin{lem}\label{indpt}
  The events $\mathfrak{R}_k^{q,e}(s)$ are conditionally independent
  given $\mathfrak{R}^{q}(s)$.
\end{lem}
\begin{proof}
  Defining $\mathfrak{R}_{\geq n}^{q}(s)$ as the set of runs starting
  in $s$ and visiting $q$ at least $n$ times, we compute:
  \begin{align*}
    \Prob\left(\mathfrak{R}_k^{q,e}(s) \mid \mathfrak{R}^{q}(s)\right)
    &=\Prob\left(\mathfrak{R}_k^{q,e}(s)
      \mid \bigcap_{n>k} \mathfrak{R}_{\geq n}^{q}(s)\right) \displaybreak[0]\\
    & = \frac{\displaystyle \Prob\left(\mathfrak{R}_k^{q,e}(s) \cap
        \bigcap_{n>k} \mathfrak{R}_{\geq
          n}^{q}(s)\right)}{\displaystyle
      \Prob\left(\bigcap_{n>k} \mathfrak{R}_{\geq n}^{q}(s)\right)} \\
    & = \frac{\displaystyle \Prob\left(\bigcap_{n>k}
        \left(\mathfrak{R}_k^{q,e}(s) \cap \mathfrak{R}_{\geq
            n}^{q}(s)\right)\right)}{\displaystyle
      \Prob\left(\bigcap_{n>k} \mathfrak{R}_{\geq n}^{q}(s)\right)}\,.
  \end{align*}
  The two sequences $\left(\mathfrak{R}_k^{q,e}(s) \cap
    \mathfrak{R}_{\geq n}^{q}(s)\right)_{n>k}$ and
  $\left(\mathfrak{R}_{\geq n}^{q}(s)\right)_{n>k}$ are
  non-increasing, hence:
  \[
  \left\{\begin{array}{l}
      \Prob\left(\bigcap_{n>k} \left(\mathfrak{R}_k^{q,e}(s)
          \cap \mathfrak{R}_{\geq n}^{q}(s)\right)\right)
           =\displaystyle \lim_{n \rightarrow \infty}
      \Prob\left(\mathfrak{R}_k^{q,e}(s) \cap
        \mathfrak{R}_{\geq n}^{q}(s)\right)\\
        \text{}\\
      \displaystyle \Prob\left(\bigcap_{n>k} \mathfrak{R}_{\geq
          n}^{q}(s)\right)=\displaystyle \lim_{n \rightarrow
        \infty} \Prob\left(\mathfrak{R}_{\geq n}^{q}(s)\right).
    \end{array}\right.
  \]
  Thus, by Corollary~\ref{cor:decomposition},
  \begin{align*}
    \Prob\left(\mathfrak{R}_k^{q,e}(s) \mid
      \mathfrak{R}^{q}(s)\right)
    & =
    \frac{\displaystyle \lim_{n \rightarrow \infty} 
      \Prob\left(\mathfrak{R}_k^{q,e}(s) 
        \cap \mathfrak{R}_{\geq
          n}^{q}(s)\right)}{\displaystyle \lim_{m \rightarrow \infty} 
      \Prob\left(\mathfrak{R}_{\geq m}^{q}(s)\right)} \\
    & = \frac{\displaystyle \lim_{n \rightarrow \infty} 
      \Prob\left(\mathfrak{E}_{q}(s)\right) \cdot
      \Prob\left(\mathfrak{R}_0^{q,e}(s_q)\right)}
    {\displaystyle \lim_{m \rightarrow \infty} 
      \Prob\left(\mathfrak{E}_{q}(s)\right)} \\
    & = \Prob\left(\mathfrak{R}_0^{q,e}(s_q)\right)
  \end{align*}
  where $\mathfrak{E}_{q}(s)$ is the set of runs starting in $s$ and
  visiting $q$ at least once, and $s_q$ is the canonical configuration
  $(\ell,v_r)$ for $q=(\ell,r)$.

  Similarly we prove, using Corollary~\ref{cor:decomposition}, that
  for $k\ne k'$:
  \begin{align*}
    \Prob\left(\mathfrak{R}_k^{q,e}(s) \cap \mathfrak{R}_{k'}^{q,e}(s) \mid
      \mathfrak{R}^{q}(s)\right) &= 
    \Prob\left(\mathfrak{R}_0^{q,e}(s_q)\right)^2\\
    & = \Prob\left(\mathfrak{R}_k^{q,e}(s) \mid
      \mathfrak{R}^{q}(s)\right) \cdot 
    \Prob\left(\mathfrak{R}_{k'}^{q,e}(s) \mid
      \mathfrak{R}^{q}(s)\right).
  \end{align*}
  We conclude that the two events $\mathfrak{R}_k^{q,e}(s)$ and
  $\mathfrak{R}_{k'}^{q,e}(s)$ are conditionally independent given
  $\mathfrak{R}^{q}(s)$. 
\end{proof}

\paragraph{Divergence of the series.}

\begin{lem}
  $\displaystyle \sum_{k \in \IN} \Prob(\mathfrak{R}_k^{q,e}(s)\mid
  \mathfrak{R}^{q}(s)) = +\infty$.
\end{lem}

\begin{proof}
  As in the previous lemma, we can deduce from equalities of
  Corollary~\ref{cor:decomposition}
  that \[\Prob\left(\mathfrak{R}_k^{q,e}(s)
    \mid\mathfrak{R}^{q}(s)\right)=
  \Prob\left(\mathfrak{R}_0^{q,e}(s_q)\right).\] However, as $q\in
  \mathcal{Q}'$, we know by definition that we have
  \[\Prob(\mathfrak{R}_0^{q,e}(s_q))>0.\]
  The result follows.
\end{proof}

\paragraph{Decomposition using basic sets.}
This section aims to prove Corollary~\ref{cor:decomposition} and so to
complete the previous proofs.

\begin{lem}
  \label{lemma:vm}
  Let $r$ be a memoryless region, $v\in r$, and $s'=(\ell,v)$ a
  configuration of $\A$. Writing $q$ for region-state $(\ell,r)$, we
  have for every sequence $(e_1,\dots,e_n) \in E^n$,
  \[
  \Prob(\path{s',e_1\dots e_n}) = \Prob(\path{s_q,e_1\dots e_n}).
  \]
\end{lem}
\begin{proof}
  We will prove a stronger result.  We show that for every pair
  $s'=(\ell,v)$, $s''=(\ell,v')$ satisfying for every $x\in X$,
  $v(x)=v'(x)$ or $\min(v(x),v'(x))>M$, we have for every sequence
  $(e_1,\dots,e_n) \in E^n,$
  \[ \Prob(\path{s',e_1\dots e_n}) = \Prob(\path{s'',e_1\dots
    e_n})\,.\] We prove this result by induction on the length $n$ of
  the sequence of edges. The result trivially holds for $n=0$. Assume
  that $n>0$, and that the lemma holds for sequences of edges of
  length $n-1$.  The constraint labelling the edge $e_1$ is equivalent
  to some $\displaystyle \bigwedge_{x \in X} (x \in I_x)$ where $I_x$
  is an interval of the form $[c;c]$ for some integer $0 \leq c \leq
  M$, or $(c;c+1)$ for some integer $0 \leq c < M$, or
  $(M;+\infty)$. We have that:
  \[
   \Prob(\path{s',e_1 \dots e_n})
   = \int_{t \in I(s',e_1)} p_{s'+t}(e_1)\,
    \Prob(\path{{s'}_t^{e_1},e_2 \dots e_n})\, \ud\mu_{\ell}(t)
  \]
  where $s' \xrightarrow{t,e_1} {s'}_t^{e_1}$. Now, it is not
  difficult to check that $p_{s'+t}(e_1)=p_{s''+t}(e_1)$ and $I(s',e_1)
  = I(s'',e_1)$ by hypothesis on $s'$ and $s''$. Also, writing $s''
  \xrightarrow{t,e_1} {s''}_t^{e_1}$, we easily get that if
  ${v'}^{e_1}_t$ and ${v''}^{e_1}_t$ are the valuations of
  ${s'}^{e_1}_t$ and ${s''}^{e_1}_t$ then for every $x\in X$, we have
  ${v'}^{e_1}_t(x)={v''}^{e_1}_t(x)$ or
  $\min({v'}^{e_1}_t(x),{v''}^{e_1}_t(x))>M$. We deduce by induction
  hypothesis:
  \[\Prob(\path{{s'}_t^{e_1},e_2 \dots e_n}) =
  \Prob(\path{{s''}_t^{e_1},e_2 \dots e_n}).\] Hence, we have
  \begin{align*}
    \Prob(\path{s',e_1 \dots e_n})
    &= \int_{t \in I(s'',e_1)} p_{s''+t}(e_1)\,
     \Prob(\path{{s''}_t^{e_1},e_2 \dots e_n})\, \ud\mu_{\ell}(t) \\
    &= \Prob(\path{s'',e_1 \dots e_n}).
  \end{align*}
  This concludes the proof. 
\end{proof}

We define $\mathfrak{E}_{q}(s)$ the set of runs starting in $s$ and
visiting $q$ at least once.

\begin{prop}
  \label{prop:decomposition}
  Let $q\in \mathcal{Q}'$. The following equalities hold true:
  \begin{enumerate}
  \item For every $n \geq 1$,
    \[
    \Prob\left(\mathfrak{R}_{\geq n}^{q}(s)\right) =
    \Prob\left(\mathfrak{E}_{q}(s)\right) \cdot
    \Prob\left(\mathfrak{R}^{q}_{\geq 1}(s_q)\right)^{n-1}.
    \]
  \item For every $1 \leq k < n$,
    \[
      \Prob\left(\mathfrak{R}_k^{q,e}(s) \cap
        \mathfrak{R}_{\geq n}^{q}(s)\right) =
      \Prob\left(\mathfrak{E}_{q}(s)\right)
      \cdot \Prob\left(\mathfrak{R}_{\geq 1}^{q}(s_q)\right)^{n-2}
      \cdot \Prob\left(\mathfrak{R}^{q}_{\geq 1}(s_q) \cap
        \mathfrak{R}_0^{q,e}(s_q)\right).
    \]
  \item For every $1 \leq k < k' < n$,
    \begin{align*}
      \Prob\left(\mathfrak{R}_k^{q,e}(s) \cap
        \mathfrak{R}_{k'}^{q,e}(s) \cap \mathfrak{R}_{\geq
          n}^{q}(s)\right) &=
      \Prob\left(\mathfrak{E}_{q}(s)\right)\\
      &\quad \cdot \Prob\left(\mathfrak{R}_{\geq
          1}^{q}(s_q)\right)^{n-3} \cdot
      \Prob\left(\mathfrak{R}^{q}_{\geq 1}(s_q) \cap
        \mathfrak{R}_0^{q,e}(s_q)\right)^2.
    \end{align*}
  \end{enumerate}
\end{prop}
\begin{proof}
  We will only prove the first equality. The other equalities can be
  handled, using decompositions, similarly to the first equality.

  We define the set $\varepsilon^q(s)$ as:
  \begin{multline*}
    \bigcup_{h \in \IN} \left\{(e_1,\dots,e_h) \in E^h \mid
      \mathsf{source}(e_1)= [s]_{\A},\right.\\
    \left.\mathsf{target}(e_h)= q,\ \text{and}\ \mathsf{target}(e_i)
      \ne q\ \text{for any}\ i<h\right\}\,.
  \end{multline*}
  Note that a run $\varrho$ is in $\mathfrak{E}_{q}(s)$ iff there is
  some $(e_1,\dots,e_h) \in \varepsilon^q(s)$ with $\varrho \in
  \Cyl\left(\path{s,e_1 \dots e_h}\right)$. The set $\varepsilon^q(s)$
  is a \emph{basic set} for $\mathfrak{E}_{q}(s')$.  For every state
  $s'$, we define the set $\gamma_n^q(s')$ as:
  \begin{multline*}
    \bigcup_{h \in \IN} \big\{(e_1,\dots,e_h) \in E^h \mid
    \mathsf{source}(e_1) = [s']_{\A},\\
    \mathsf{target}(e_h) = q,\ \text{and}\ \#\{1 \leq i \leq h \mid
    \mathsf{target}(e_i) = q\} = n \big\}\,.
  \end{multline*}
  Note that a run $\varrho$ is in $\mathfrak{R}_{\geq n}^{q}(s')$ iff
  there is some $(e_1,\dots,e_h) \in \gamma_n^q(s')$ such that
  $\varrho \in \Cyl\left(\path{s',e_1 \dots e_h}\right)$. The set
  $\gamma_n^q(s')$ is a \emph{basic set} for $\mathfrak{R}_{\geq
    n}^{q}(s')$.  Note also that if $s'$ and $s''$ are
  region-equivalent, then the two sets $\gamma_n^q(s')$ and
  $\gamma_n^q(s'')$ coincide. We will thus write $\gamma_n^q(q)$,
  where $q$ is the region-state of $s'$ and $s''$.

  We will decompose runs in $\mathfrak{R}_{\geq n}^{q}(s)$ using the
  basic sets $\varepsilon^q(s)$ and $\gamma_n^q(s_q)$. Indeed, if
  $\varrho \in \mathfrak{R}_{\geq n}^{q}(s)$, then we can decompose
  $\varrho$ into $\varrho' \cdot \varrho''$ such that there exist
  $(e_1,\dots,e_h) \in \varepsilon^q(s)$ with $\varrho' \in
  \path{s,e_1 \dots e_h}$, and $(f_1,\dots,f_k) \in
  \gamma_{n-1}^q(\mathsf{last}(\varrho'))$ with $\varrho'' \in
  \Cyl\left(\path{\mathsf{last}(\varrho'),f_1,\dots,f_k}\right)$. By
  definition of $\varepsilon^q(s)$, $\mathsf{last}(\varrho') \in q$,
  and by applying Lemma~\ref{lemma:vm}, we get
  \[
  \Prob\left(\path{\mathsf{last}(\varrho'),f_1 \dots f_k}\right)
  = \Prob\left(\path{s_q,f_1,\dots,f_k}\right).
  \]
  If we denote $\sum_{\varepsilon^q(s)}$ instead of $\sum_{(e
    _1,\dots,e_h) \in \varepsilon^q(s)}$ and
  $\sum_{\gamma_{n-1}^q(q)}$ instead of $ \sum_{(f_1,\dots,f_k) \in
    \gamma_{n-1}^q(q)}$, we have the following simplification:
  \begin{align*}
    \Prob\left(\mathfrak{R}_{\geq n}^{q}(s)\right)
    & = \sum_{\varepsilon^q(s)} \ \sum_{\gamma_{n-1}^q(q)}  \Prob \left\{\varrho = \varrho' \cdot \varrho'' \mid \varrho' \in \path{s,e_1\dots e_h},\right.\\
    &\qquad\qquad\qquad\qquad\qquad\left. \varrho'' \in \Cyl\left(\path{\last(\varrho'),f_1 \dots f_k}\right)\right\} \\
    &= \sum_{\varepsilon^q(s)} \ \sum_{\gamma_{n-1}^q(q)} \left( \int_{t_1} \int_{t_2} \dots \int_{t_h} \Big(\prod_{i=0}^{h-1} p_{s_{i}+t_{i+1}}(e_{i+1})\Big)\right.\displaybreak[0]\\
    &\left.\vphantom{\prod^h}\qquad\qquad\cdot\Prob\left(\path{s',f_1 \dots f_k}\right)\, \ud\mu_{s_{h-1}}(t_h)\, \dots \ud\mu_{s}(t_1)\right)\\
    & \quad \text{(where $s'=\last(s \xrightarrow{t_1,e_1}s_1 \dots s_{h-1}\xrightarrow{t_h,e_h}) \in q$)}\displaybreak[0]\\
    &= \sum_{\varepsilon^q(s)} \ \sum_{\gamma_{n-1}^q(q)} \left( \int_{t_1} \int_{t_2} \dots \int_{t_h} \Big(\prod_{i=0}^{h-1} p_{s_{i}+t_{i+1}}(e_{i+1})\Big)\right.\\
    &\left.\vphantom{\prod^h}\qquad\qquad\cdot\Prob\left(\path{s_q,f_1\dots f_k}\right)\, \ud\mu_{s_{h-1}}(t_h)\, \dots \ud\mu_{s}(t_1)\right)\displaybreak[0]\\
    &= \left(\sum_{\gamma_{n-1}^q(q)} \Prob\left(\path{s_q,f_1\dots f_k}\right) \right)\\
    &\quad\quad\cdot \left(\sum_{\varepsilon^q(s)} \int_{t_1} \int_{t_2} \dots \int_{t_h} \Big(\prod_{i=0}^{h-1} p_{s_{i}+t_{i+1}}(e_{i+1})\Big)\ud\mu_{s_{h-1}}(t_h)\, \dots \ud\mu_{s_0}(t_1) \right)\\
    \displaybreak[0]\\
    &\quad= \Prob\left(\mathfrak{R}_{\geq n-1}^{q}(s_q)\right) \cdot
    \Prob\left(\mathfrak{E}_{q}(s)\right).
  \end{align*}

  By induction on $n$, using a similar decomposition, we can prove that:
  \begin{equation*}\label{eqn}
  \Prob\left(\mathfrak{R}_{\geq n}^{q}(s_q)\right) =
  \Prob\left(\mathfrak{R}_{\geq 1}^{q}(s_q)\right)^n
  \end{equation*}  
  which concludes the proof for the first equality.
\end{proof}

We can simplify equalities of the previous proposition thanks to
following lemma:

\begin{lem}
  \label{limitok}
  Let $q\in\mathcal{Q}'$. We have $\Prob(\mathfrak{R}^{q}_{\ge
    1}(s_q)) = 1$.
\end{lem}
\begin{proof}
  For a contradiction, we assume that $\Prob(\mathfrak{R}^{q}_{\ge
    1}(s_q)) = \alpha_0 < 1 $. By
  Proposition~\ref{prop:decomposition}, we thus have that:
\begin{align*}
  \Prob(\mathfrak{R}^q(s)) \
  &= \Prob\left(\bigcap_n \mathfrak{R}^{q}_{\ge n}(s)\right)
  = \lim_{n \to \infty} \Prob\left(\mathfrak{R}_{\geq
      n}^{q}(s)\right) \\
  &= \lim_{n \to \infty} \Prob\left(\mathfrak{E}_{q}(s)\right)
  \cdot \Prob\left(\mathfrak{R}^{q}_{\geq 1}(s_q)\right)^{n-1} \\
  &= \Prob\left(\mathfrak{E}_{q}(s)\right) \lim_{n \to \infty}
  (\alpha_0)^{n-1} = 0
\end{align*}
which contradicts the fact that $q\in\mathcal{Q}'$.
\end{proof}

As an immediate corollary we get the following result.

\begin{cor}
  \label{cor:decomposition}
  Let $q\in\mathcal{Q}'$. The following equalities hold true:
  \begin{enumerate}
  \item For every $n \geq 1$,
    \[
    \Prob\left(\mathfrak{R}_{\geq n}^{q}(s)\right) =
    \Prob\left(\mathfrak{E}_{q}(s)\right).
    \]
  \item For every $1 \leq k < n$,
    \[
      \Prob\left(\mathfrak{R}_k^{q,e}(s) \cap
        \mathfrak{R}_{\geq n}^{q}(s)\right)
        = \Prob\left(\mathfrak{E}_{q}(s)\right) \cdot
      \Prob\left(\mathfrak{R}_0^{q,e}(s_q)\right).
    \]
  \item For every $1 \leq k < k' < n$,
    \[
      \Prob\left(\mathfrak{R}_k^{q,e}(s) \cap
        \mathfrak{R}_{k'}^{q,e}(s) \cap
        \mathfrak{R}_{\geq n}^{q}(s)\right)
        = 
      \Prob\left(\mathfrak{E}_{q}(s)\right) \cdot
      \Prob\left(\mathfrak{R}_0^{q,e}(s_q)\right)^2.
    \]
  \end{enumerate}
\end{cor}
This concludes the proof of Lemma~\ref{lemma:eq2}.
\end{proof}

\ligne

We now extend the previous study to weak reactive stochastic timed
automata.  \medskip

\noindent \fbox{\begin{minipage}{.98\linewidth}
\weakreactive* \end{minipage}}
\medskip

\begin{proof}
  We seek to show that, with probability $1$, we delay infinitely many
  times more than $M$ time units before taking a transition. Let
  $\mathfrak{D}^{>M}_n(s)$ be the set of runs from state $s$ that
  delays more than $M$ time units before taking the $n$-th transition
  (i.e.  $\mathfrak{D}^{>M}_n(s) = \{\varrho \in \Runs(\A,s) \mid
  \varrho = s \xrightarrow{\tau_1,e_1} s_1 \xrightarrow{\tau_2,e_2}
  \dots \mid \tau_n > M\}$). We denote $\mathfrak{D}^{\le
    M}_{k,n}(s):=\bigcap_{k\le j\le n}(\mathfrak{D}^{>M}_j(s))^c$.

  Assume $s=(\ell,v)$ with $\ell \in L_u$.  Since $\A$ is a weak
  reactive stochastic timed automaton, there exists $0<\lambda_0\le 1$
  such that 
  for every $\ell \in L_u$, for every $v \models \mathcal{I}(\ell)$,
  we have
  \[
  \mu_{(\ell,v)}\left(\left[M,+\infty\right[\,\right) \ge \lambda_0,
  \]
  and there exist $0<\lambda_1\le 1$ and $N\ge 1$ such that for any
  $\ell\in L_b$, any $v \models \mathcal{I}(\ell)$, we
  have 
  \[
  \Prob_{\A}\Big(\bigcup_{(e_1,\dots,e_N)\in
    E_u}\path{(\ell,v),e_1,\dots,e_N}\Big)\ge \lambda_1
  \]
  where $E_u=\{(e_1,\dots,e_N) \mid \target(e_i)\in L_u\text{\ for
    some\ } 1\le i< N\}$.

  Let $s=(\ell,v)$ be a state of $\A$. If $\ell\in L_u$, we have
  \[
  \Prob\Big(\bigcup_{1\le j\le N}\mathfrak{D}^{>M}_j(s)\Big)\ge
  \Prob(\mathfrak{D}^{>M}_1(s))\ge\lambda_0\ge \lambda_0\lambda_1\,.
  \]
  If $\ell \in L_b$, we first remark that for every $v \models
  \mathcal{I}(\ell)$, we have
  \[
  \Prob_{\A}\Big(\bigcup_{(e_1,\dots,e_N)\in
    E_u}\path{s,e_1,\dots,e_N}\Big)= \sum_{(e_1,\dots,e_k)\in
    F_u}\Prob_{\A}(\path{s,e_1,\dots,e_k})
  \]
  where \[F_u=\{(e_1,\dots,e_k) \mid 1\le k< N, \target(e_k)\in L_u
  \text{\ and for any\ } {1\le i< k,}\ {\target(e_i)\notin
    L_u}\}.\]

  Therefore, if $\ell \in L_b$, we have
  \begin{align*}
    &\Prob(\bigcup_{1\le j\le N}\mathfrak{D}^{>M}_j(s)))\\
    &\quad\ge \sum_{(e_1,\dots,e_k)\in F_u}\sum_{e\in E}\Prob(\path[t_{k+1}>M]{s,e_1,\dots,e_k,e}\\
    &\quad\ge \sum_{(e_1,\dots,e_k)\in F_u}\sum_{e\in E}\int_{t_1\in I(s,e_1)} p_{s+t_1}(e_1)\cdots\\
    &\phantom{\sum_{(e_1,\dots,e_k)\in F_u}\sum_{e\in E}\int_{t_1\in
        I(s,e_1)}} \cdots
    \int_{t_{k+1}\in I(s_{k},e)\cap ]M,+\infty[}p_{s_k+t_{k+1}}(e)\ud\mu_{s_{k}}(t_{k+1})\dots\ud\mu_{s}(t_1)\\
    &\quad\ge \sum_{(e_1,\dots,e_k)\in F_u}\int_{t_1\in I(s,e_1)} p_{s+t_1}(e_1)\cdots\\
    &\phantom{\sum_{(e_1,\dots,e_k)\in F_u}\sum_{e\in E}\int_{t_1\in
        I(s,e_1)}} \cdots
    \int_{t_{k+1}\in ]M,+\infty[}\sum_{e\in E} p_{s_k+t_{k+1}}(e)\ud\mu_{s_{k}}(t_{k+1})\dots\ud\mu_{s}(t_1)\\
    &\quad\ge \sum_{(e_1,\dots,e_k)\in F_u}\int_{t_1\in I(s,e_1)} p_{s+t_1}(e_1)\cdots
    \int_{t_{k+1}\in ]M,+\infty[}\ud\mu_{s_{k}}(t_{k+1})\dots\ud\mu_{s}(t_1)\\
    &\quad\ge \sum_{(e_1,\dots,e_k)\in F_u}\int_{t_1\in I(s,e_1)} p_{s+t_1}(e_1)\cdots\\
    &\phantom{\sum_{(e_1,\dots,e_k)\in F_u}\sum_{e\in E}} \cdots
    \int_{t_{k}\in I(s_{k-1},e_k)}p_{s_{k-1}+t_{k}}(e_k) \mu_{s_k}(]M,+\infty[)\ud\mu_{s_{k-1}}(t_{k})\dots\ud\mu_{s}(t_1)\\
    &\quad\ge \lambda_0 \sum_{(e_1,\dots,e_k)\in F_u} \Prob_{\A}(\path{s,e_1,\dots,e_k})\\
    &\quad \ge \lambda_0\lambda_1\,.
  \end{align*}
  We deduce that for every state $s$, $\Prob(\mathfrak{D}^{\le
    M}_{1,N}(s))\le 1-\lambda_0\lambda_1$.  Therefore for any state
  $s$, for any $k\ge 2$, we get
  \begin{align*}
    &\Prob(\mathfrak{D}^{\le M}_{1,kN}(s))\\
    &\quad\le \sum_{(e_1,\dots,e_{(k-1)N)}}\int_{t_1\in I(s,e_1)\cap
      [0,M]} p_{s+t_1}(e_1) \cdots
    \int_{t_{(k-1)N}\in I(s_{(k-1)N-1},e_{(k-1)N})\cap [0,M]}\\
    &\quad\quad
    \Big(p_{s_{(k-1)N-1}+t_{(k-1)N}}(e_{(k-1)N})\Prob(\mathfrak{D}^{\le
      M}_{1,N}(s_{(k-1)N}))\Big)
    \ud\mu_{s_{(k-1)N-1}}(t_{(k-1)N})\dots\ud\mu_{(\ell,v)}(t_1)\\
    &\quad\le (1-\lambda_0\lambda_1) \Prob(\mathfrak{D}^{\le M}_{1,(k-1)N}(s))\\
&\quad\le (1-\lambda_0\lambda_1)^k.
\end{align*}

\noindent In the same way, we deduce that for any $j,k\ge 1$, 
\[
\Prob(\mathfrak{D}^{\le M}_{j,j+ kN}(s))\le
(1-\lambda_0\lambda_1)^k\,.
\]
We conclude that for any $k\ge 1$, $\Prob(\mathfrak{D}^{\le
  M}_{k,n}(s))$ converges to $0$ when $n$ tends to infinity and we
finish the proof with the same arguments as in the proof of
Lemma~\ref{lemma:eq1}.  
\end{proof}

\noindent \fbox{\begin{minipage}{.98\linewidth}
\reactivenonzeno* \end{minipage}}

\label{app:reactivenonzeno}

\begin{proof}
  This is a consequence of the proofs of Lemma~\ref{lemma:eq1} (for
  reactive automata) and of Lemma~\ref{lemma:weakreactive} (for weak
  reactive automata).

  We should just notice that, writing $\mathsf{Zeno}(s)$ for the set
  of Zeno runs from $s$,
  \[
  \mathsf{Zeno}(s)^c \supseteq \bigcap_{n \in \IN} \bigcup_{k \geq n}
  \mathfrak{D}^{>M}_k(s)\,.
  \]
  And in the two proofs that are mentioned, it is proven that
  $\Prob(\mathfrak{D}^{\le M}_{k,n}(s))$ tends to $0$ when $n$ tends
  to $\infty$, which implies that $\Prob(\mathsf{Zeno}(s))=0$. 
\end{proof}

\section{Details for Section~\ref{sec:algo}}
\label{app:lower-bounds}

\noindent \fbox{\begin{minipage}{.98\linewidth}
\main* \end{minipage}}\medskip

The upper bounds have already been explained in the core of the paper
(they are obvious consequences of the previous developments). We will
now explain several lower bounds.

\subsection*{Hardness in reactive timed automata}

We prove that the almost-sure model-checking problem in reactive timed
automata against simple safety and simple reachability properties is
\PSPACE-hard. To that aim we simulate a linearly-bounded Turing
machine $\mathcal{M}$ on an input word $w_0$. The general reduction is
rather standard~\cite{AL02} but it is required to work out the details
so that there is no equality constraints in the constructed timed
automaton, so that $\RA$ and $\thickgraph(\A)$ coincide.

Let $N$ be the bound on the tape of $\mathcal{M}$ when simulating on
input word $w_0$. We assume the alphabet is $\{a,b\}$ and we encode
the content of $j$-th cell $C_j$ using a clock $x_j$ with the
following convention: when we enter a module, cell $C_j$ contains an
$a$ whenever $x_j<1$ and it contains a $b$ whenever $x_j>2$.  To
simulate a transition $q'=\delta(q,\alpha,\beta,\mathsf{dir})$ where
$\alpha,\beta \in \{a,b\}$, we construct a module as in
Fig.~\ref{fig:simul} for every index $i$ such that $1 \le i \le N$ and
$1 \le \mathsf{dir}(i) \le N$, where $\mathsf{dir}(i)$ is either $i+1$
(if the head goes to the left), or $i-1$ (if the head goes to the
right).

\begin{figure}[h!]
\begin{center}
  \begin{tikzpicture}[scale=.7]
    \everymath{\scriptstyle}
    \draw (0,0) node [draw,rounded corners=1.5mm] (A) {$q,i$};
    \draw (3,0) node [draw,circle] (B) {};
    \draw (6,0) node [draw,circle] (C) {};
    \draw (9,0) node [draw,circle] (D) {};
    \draw (10,0) node {$\dots$};
    \draw (11,0) node [draw,circle] (E) {};
    \draw (13.5,0) node [draw,circle] (F) {};
    \draw (14.5,0) node {$\dots$};
    \draw (15.5,0) node [draw,circle] (G) {};
    \draw (18.5,0) node [draw,rounded corners=1.5mm] (H) {$q',\mathsf{dir}(i)$};

    \draw [latex'-] (A) -- +(-1.5,0) node [midway,above] {$u:=0$};
    \draw [-latex'] (A) -- (B) node [midway,above] {$2<u<3$};

    \draw [-latex'] (B) .. controls +(60:1cm) and +(120:1cm) .. (C) node [midway,above] {$x_1<4,u<3$} node [midway,below] {$x_1:=0$};
    \draw [-latex'] (B) .. controls +(-60:1cm) and +(-120:1cm) .. (C) node [midway,below] {$x_1>4,u<3$};

    \draw [-latex'] (C) .. controls +(60:1cm) and +(120:1cm) .. (D) node [midway,above] {$x_2<4,u<3$} node [midway,below] {$x_2:=0$};
    \draw [-latex'] (C) .. controls +(-60:1cm) and +(-120:1cm) .. (D) node [midway,below] {$x_2>4,u<3$};

    \draw [-latex'] (E) -- (F) node [midway,above] {$g_{\alpha,i},Y_{\beta,i}$};

    \draw [-latex'] (G) .. controls +(60:1cm) and +(120:1cm) .. (H) node [midway,above] {$x_N<4,u<3$} node [midway,below] {$x_N:=0$};
    \draw [-latex'] (G) .. controls +(-60:1cm) and +(-120:1cm) .. (H) node [midway,below] {$x_N>4,u<3$};

  \end{tikzpicture}
\end{center}
\caption{Simulation of $\mathcal{M}$-transition
  $q'=\delta(q,\alpha,\beta,\mathsf{dir})$ where $\alpha,\beta \in
  \{a,b\}$. Index $i$ is such that $1 \le i \le N$ and $1 \le
  \mathsf{dir}(i) \le N$. Guard $g_{a,i}$ is $x_i<4,u<3$ and guard
  $g_{b,i}$ is $x_i>4,u<3$. Set $Y_{a,i}$ is $\{x_i\}$ and set
  $Y_{b,i}$ is $\emptyset$.}
\label{fig:simul}
\end{figure}
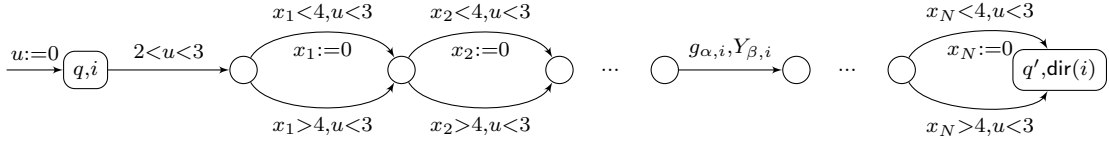

We complete the construction with an initialisation module (with input
$w_0$ on the tape), and we complete the automaton so that it is
reactive (by adding transitions to $\mathsf{sink}$ location). We note
$\A$ for this timed automaton and write $\mathsf{halt}$ for the
halting location (which can be made a sink). Let $P_{\mathsf{safety}}$
be the safety property that $\mathsf{halt}$ is not visited, and
$P_{\mathsf{reach}}$ be the property of reaching $\mathsf{sink}$. We
attach the exponential distribution with parameter, say $1$, to every
state, and assume weight $1$ for every edge. Let $s_0$ be the initial
state of $\A$ where all clocks are set to $0$.

\begin{lem}
  The following equivalences hold:
  \begin{eqnarray*}
    \mathcal{M}\ \text{does not halt on input}\ w_0 & 
    ~\Leftrightarrow~ & \Prob_{\A}(s_0 \models P_{\mathsf{safety}})=1 
    \\
    & 
    ~\Leftrightarrow~ & \Prob_{\A}(s_0 \models P_{\mathsf{reach}})=1
  \end{eqnarray*}
\end{lem}

\begin{proof}
  Assume $\mathcal{M}$ halts on input $w_0$. Then there is a finite
  run $\rho$ in $\A$ leading to $\mathsf{halt}$. Due to the special
  form of $\A$, the probability of the cylinder generated by
  $\pi_\rho$ is positive ($\A$ has only strict guards). This implies
  that $\Prob_{\A}(s_0 \models P_{\mathsf{reach}})<1$ (since state
  $\mathsf{sink}$ is not reached by that cylinder), and
  $\Prob_{\A}(s_0 \models P_{\mathsf{safety}})<1$.

  If $\mathcal{M}$ does not halt on input $w_0$, then location
  $\mathsf{halt}$ is never visited, and therefore $\mathsf{sink}$ is
  visited with probability $1$. This implies that $\Prob_{\A}(s_0
  \models P_{\mathsf{reach}})=1$, and that $\Prob_{\A}(s_0 \models
  P_{\mathsf{safety}})=1$. 
\end{proof}

This shows hardness results for (i) and (iv).

\subsection*{Hardness in single-clock timed automata}

The \NLOGSPACE-hardness result of (ii) already holds for finite
Markov chains (since checking reachability properties in finite graphs
is \NLOGSPACE-hard).

Similarly, the result of (iii) concerning \PSPACE-hardness already holds for
finite Markov chains~\cite{vardi85}.
\end{document}